\documentclass[12pt]{article}
\usepackage[authoryear, round]{natbib}
\usepackage[english]{babel}
\usepackage{amsmath,amssymb}
\usepackage{ifthen}
\usepackage{graphicx}
\usepackage{amsthm}
\usepackage{hyperref}
\usepackage{mathrsfs}
\usepackage{amssymb}
\usepackage{amsfonts}
\usepackage{mathtools}
\usepackage{rotating}
\usepackage[font=small,labelfont=bf]{caption}
\usepackage{subcaption}
\usepackage{geometry}
\usepackage{setspace}
\usepackage{multirow}
\usepackage{longtable}
\usepackage{booktabs}
\usepackage{enumitem}
\usepackage{color}
\usepackage{xcolor}
\usepackage{tikz}
\usetikzlibrary{calc}

\newcommand{\E}{\mathbb{E}}

\graphicspath{{figs/}}

\newcommand{\sizeCBi}{57}  %
\newcommand{\sizeCBii}{62}  %
\newcommand{\sizeCBiii}{68}  %
\newcommand{\sizeCBiv}{85}  %
\newcommand{\sizeCBv}{91}  %
\newcommand{\sizeCSi}{60}  %
\newcommand{\sizeCSii}{71}  %
\newcommand{\sizeCSiii}{109}  %
\newcommand{\sizeCSiv}{34}  %
\newcommand{\sizeCSv}{89}  %
\newcommand{\sizeSBi}{97}  %
\newcommand{\sizeSBii}{154}  %
\newcommand{\sizeSBiii}{75}  %
\newcommand{\sizeSBiv}{21}  %
\newcommand{\sizeSBv}{16}  %
\newcommand{\sizePWi}{89}  %
\newcommand{\sizePWii}{33}  %
\newcommand{\sizePWiii}{109}  %
\newcommand{\sizePWiv}{64}  %
\newcommand{\sizePWv}{68}  %
\newcommand{\sizeLEi}{25}  %
\newcommand{\sizeLEii}{114}  %
\newcommand{\sizeLEiii}{59}  %
\newcommand{\sizeLEiv}{101}  %
\newcommand{\sizeLEv}{64}  %
\newcommand{\sizeRKi}{113}  %
\newcommand{\sizeRKii}{17}  %
\newcommand{\sizeRKiii}{170}  %
\newcommand{\sizeRKiv}{62}  %
\newcommand{\sizeRKv}{1}  %
\newcommand{\pctRKlargest}{47}  %

\newcommand{\ratioObsSeventy}{15.5}

\newcommand{\ratioObsMillennium}{18.6}
\newcommand{\vintageGain}{20}

\newcommand{\piRankDesignOne}{0.88}
\newcommand{\piRankDesignTwo}{0.91}

\newtheorem{theorem}{Theorem}

\newtheorem{corollary}{Corollary}
\newtheorem{lemma}{Lemma}
\newtheorem{proposition}{Proposition}
\theoremstyle{definition}

\newtheorem{remark}{Remark}
\newtheorem{example-continued}{Example}

\newtheorem{assumption}{Assumption}

\author{
        Simon Freyaldenhoven\\
        \textit{Federal Reserve Bank of Philadelphia\thanks{The views expressed herein are those of the author and do not necessarily reflect the views of the Federal Reserve Bank of Philadelphia, or the Federal Reserve System. Email: \href{mailto:simon.freyaldenhoven@phil.frb.org}{simon.freyaldenhoven@phil.frb.org}}}
}
\title{When Can We Work in Embedding Space? What Text Embeddings Preserve}

\date{August 2026}
\begin{document}

\maketitle

\begin{abstract}
\noindent When do text embeddings work as inputs to empirical analysis? 
Their use rests on an assumption: that we can trade text for its low-dimensional embedding, and lose little in doing so. 
I make that assumption precise under a generative model in which documents are mixtures of latent topics.
I study two uses---\emph{clustering} units in embedding space and \emph{controlling} for high-dimensional text. A cluster of embeddings is a set of documents with similar topic mixtures; controlling for the embedding is equivalent to controlling for the topic mixture, so validity reduces to whether that mixture captures the confounding.
In an application to 363 U.S.\ metropolitan areas, embedding-based clusters of LLM-generated economic descriptions recover interpretable economic archetypes and separate local employment dynamics more sharply than clustering on model residuals, or on a curated set of industry and demographic covariates.
\end{abstract}

JEL-Classification:  C21, C38, C45, C55


\thispagestyle{empty}
\newpage

\setcounter{page}{1}

\section{Introduction}

Text embeddings are now a standard tool in empirical economics. Researchers turn product descriptions \citep{bajari2025hedonic, bach2025adventures}, central bank communication (\cite{casella2026structural}), labor-market histories \citep{vafa2022career, vafa2025wage}, and word meanings \citep{kozlowski2019geometry} into vectors (embeddings) and use those in subsequent analysis. 
The appeal is dimension: While text itself is extremely high-dimensional, its embedding is (relatively) low-dimensional. 
The justification, either explicitly or implicitly, is that little should be lost by working in embedding space if the information contained in a text can be represented in a low-dimensional space.\footnote{For example, \cite{bajari2025hedonic} states that the success of their approach ``depends on the existence of parsimonious structures behind images and text'' and that they ``believe that information in images and sentences can be effectively represented in a much lower-dimensional space''.} I make this argument precise and study when text embeddings work as inputs to empirical analysis.

Any reduction of text to a vector may discard something the analysis needed. Existing theoretical work rules this out via a high level sufficiency assumption: that a low-dimensional attribute vector exists and the embedding proxies it \citep{christensen2026unstructured}, or that whatever the representation discards induces a bias vanishing faster than ${n}^{-1/2}$ \citep{vafa2025wage}. But how plausible is such an assumption? In order to answer this, my starting point is a simple generative model of text, which allows me to derive exactly what embeddings preserve.

I first show that, under a standard model in which documents are mixtures of $K$ latent topics \citep{blei2003latent, hofmann1999probabilistic}, the corpus's matrix of word co-occurrence ratios, centered at independence, is positive semi-definite of rank exactly $K-1$ (Proposition~\ref{thm:factorization}). %
Any embedding that matches this matrix inherits the topic-loading geometry --- plainly speaking: words with similar loadings receive similar embeddings (Theorem~\ref{thm:word_embedding_guarantee}). 
Standard embedding methods do not target this matrix directly, and instead factorize a \emph{logarithmic} transform of the co-occurrence ratios (Table \ref{tab:targets}), which is generically full rank. Focusing on one leading embedding --- Word2Vec \citep{mikolov2013efficient, mikolov2013distributed}, in its skip-gram-with-negative-sampling (SGNS) form --- I then prove that it recovers the same geometry up to a distortion (Proposition \ref{prop:oe_general}).

At the document level, averaging word embeddings yields an invertible linear image of the topic mixture (Proposition~\ref{prop:doc_embedding_general_freq}). 
I then consider two concrete use cases. 
The first \emph{groups} units by latent type: clustering in embedding space discretizes unobserved heterogeneity into peer groups. Because distinct topic mixtures cannot share an embedding, a cluster is a set of documents with similar mixtures, and its centroid corresponds to a well-defined mixture (Corollary~\ref{cor:doc_embedding_metric}). Recovering the latent partition exactly is a stronger claim: when documents concentrate tightly enough around a small number of common mixture-types, that partition is a fixed point of $k$-means (Proposition~\ref{prop:archetype_identification}), but separation alone does not make it the unique optimum.
The second \emph{conditions} on textual embeddings: the embedding enters a regression as a control. Adjusting for the embedding is equivalent to adjusting for the topic mixture, so the high-level assumption that ``the embedding is a sufficient control'' reduces to a transparent condition that the topic mixture captures the confounding. (Corollary~\ref{cor:embedding_adjustment}).

To make the grouping use concrete, consider a researcher who wants to model local employment dynamics. A natural starting point is an autoregressive model for log employment in location $i$, e.g.:
\begin{equation}\label{eq:ar_intro}
    y_{it} = \alpha_i + \rho_1^{i} y_{i,t-1} + \cdots + \rho_p^{i} y_{i,t-p} + \varepsilon_{it},
\end{equation}
where $y_{it}$ is log employment in year $t$ for a given Core-Based Statistical Area (CBSA) $i$. %
Pooling across all locations imposes implausible homogeneity, while estimating each city separately throws away cross-sectional information and can result in very noisy estimates.
I propose the following compromise: cluster units by their textual description. Concretely, for each unit CBSA (i) generate a 500-word economic narrative using a large language model, (ii) embed the narrative in $\mathbb{R}^r$, and (iii) apply $k$-means in embedding space. Our theoretical results from Section \ref{sec:theory} provide a lens into this algorithm:  If the corpus can be reasonably approximated by a topic model, $k$-means clustering recovers peer groups with similar topic mixtures.
I find that the peer groups that emerge indeed encode similarity (e.g., ``deindustrialized Eds-and-Meds'', ``Sunbelt growth'', ``energy and resource extraction''), and that the resulting clusters capture meaningful heterogeneity in employment dynamics which clustering on outcome residuals, or on observed covariates, misses.%

\textbf{Related Literature.} This paper connects several literatures. The generative model builds on the topic model literature, including probabilistic latent semantic indexing \citep{hofmann1999probabilistic} and Latent Dirichlet Allocation \citep{blei2003latent}. Our use of a topic model as the data-generating process is not only analytical convenience: a recent literature argues that large language models themselves behave as implicit latent-variable models, inferring a latent concept from context and generating text conditional on it \citep{xie2022, wang2023large}.

On the embedding side, the Word2Vec model was introduced by \cite{mikolov2013efficient, mikolov2013distributed}, and \cite{levy2014neural} established that skip-gram with negative sampling implicitly factorizes a shifted pointwise mutual information matrix when the embedding dimension is unrestricted.
\cite{arora2016rand} model the writing of a document as a random walk over a latent topic vector in a way that justifies what Word2Vec and GloVe \citep{pennington2014glove} estimate. I instead take the topic model as the data-generating process---a model that is easy to interpret and that economists already use to summarize text---and ask what it implies for embeddings. \cite{dieng2020topic} put words and topics in the same embedding space by assumption, giving each topic its own vector. I derive what they assume: under our model each topic has a centroid in embedding space, and document embeddings lie in the simplex those centroids span (Proposition~\ref{prop:doc_embedding_general_freq})---for an embedding estimated without any knowledge of the topics. Finally, \cite{li2023transformers} show that the attention layers of a transformer can also recover topic information, but under the restrictive assumption that each word belongs to a single topic, so that topics have no vocabulary in common.

\cite{veitch2020adapting} use text embeddings as a control for confounding. They construct an embedding that captures the confounding by explicitly supervising it on both treatment and outcome, while I consider ``standard'' unsupervised embeddings in this paper. Other papers studying causal inference with text include \cite{roberts2020adjusting} and \cite{egami2022how}.
Finally, \cite{battaglia2024inference}, \cite{vafa2025wage} and \cite{christensen2026unstructured} study the gap between the representation a researcher has and the object the model requires, and each closes it downstream: a bias correction for a noisy estimate of the right target, a characterization of the omitted-variable bias induced by coarsening, and a correction for an imperfect proxy, respectively.
Our paper is complementary to this literature: there, the object the representation recovers is taken as given; here, I ask what an unsupervised embedding recovers in the first place.
\section{Theoretical Results}\label{sec:theory}

\subsection{Generative Model}\label{sec:generative}

Suppose we generate a corpus of text, meaning a collection of $D$ documents from a vocabulary of size $V$. Each document $d$ consists of $N_d$ terms: $\{w_{d,1}, \ldots, w_{d,N_d}\}$. Within document $d$, each term is drawn i.i.d.\ according to the column-stochastic distribution $\Pi_{\bullet d}$, where $\Pi$ is a column-stochastic $V \times D$ matrix; that is, $\Pi_{vd}$ denotes the probability that a randomly drawn term in document $d$ equals word $v$. Throughout, I use $P(\cdot)$ exclusively as the probability operator and reserve $\Pi$ for the matrix of word--document probabilities. For a matrix $A$, $\sigma_{\max}(A)$ and $\sigma_{\min}(A)$ denote its largest and smallest singular values.

I further assume that $\Pi=B \Theta$, where $B$ is $V \times K$ (word-topic matrix) and $\Theta$ is $K \times D$ (topic-document matrix); I take $K \geq 2$ throughout to avoid degeneracy. Both $B$ and $\Theta$ have non-negative entries, and each column of $\Pi$ and $B$ sums to 1. The probability for a given term $v$ is thus given by $P(w = v | d) = (B \Theta_{\bullet d})_v.$ and does not depend on the position in the document. 

Following \cite{hofmann1999probabilistic}, I assume that for each document $d$
\begin{equation} \label{eq:multinomial}
    N_{\bullet d} | (B,\Theta) \sim \textrm{Multinomial}\left(N_d, B \Theta_{\bullet d} \right),
\end{equation}
where $N_{\bullet d} = (N_{1d}, \ldots, N_{Vd})^\top$ is the vector of word counts in document $d$, and the vectors of counts $N_{\bullet d}$ are independent across documents, conditional on $(B,\Theta)$.  

Once we add a Dirichlet prior on the per-document topic distribution we obtain a proper generative model for new documents in the form of standard LDA \citep{blei2003latent}.
On the other hand, treating the topic mixture $\theta_d$ as a fixed (unknown) parameter rather than a random variable, this has been studied in the Non-negative Matrix Factorization (NMF) literature \citep{lee1999learning, donoho2003does, arora2013practical}. 

\textbf{Notation.} Consider two words $v$ and $u$ appearing in the same document $d$. Under the i.i.d.-given-document model, conditional on $d$ a target word and a separately drawn context word are independent, so:
\begin{align}
P(w = v, c = u | d) = \Pi_{vd} \cdot \Pi_{ud} = (B\Theta_{\bullet d})_v (B\Theta_{\bullet d})_u.
\end{align}
Define the \textit{word co-occurrence matrix} $M \in \mathbb{R}^{V \times V}$ by aggregating over documents:
\begin{align}\label{eq:cooccurrence_nmf}
M_{vu} \;=\; P(w=v, c=u) \;=\; \sum_{d=1}^D p_d (B\Theta_{\bullet d})_v (B\Theta_{\bullet d})_u,
\end{align}
where $p_d = N_d / \sum_{d'} N_{d'}$ weights documents by length and $\sum_d p_d = 1$. Equivalently, with $p = [p_1, \ldots, p_D]^\top$ and $G := \sum_d p_d \Theta_{\bullet d}\Theta_{\bullet d}^\top = \Theta\,\mathrm{diag}(p)\,\Theta^\top \in \mathbb{R}^{K\times K}$ denoting the corpus second-moment matrix of topic mixtures,
\begin{align}\label{eq:M_BGB}
M \;=\; B\,G\,B^\top.
\end{align}

Let $\bar\theta := \sum_d p_d \Theta_{\bullet d}$ denote the corpus-mean topic mixture, let $q := B\bar\theta$ collect the marginal word probabilities, $q_v = P(w{=}v) = P(c{=}v)$, and let $D_q := \mathrm{diag}(q)$ denote the corresponding diagonal matrix. Define the \emph{probability-ratio matrix}
\begin{align}
R \;:=\; D_q^{-1}\, M\, D_q^{-1} \;\in\; \mathbb{R}_{>0}^{V \times V}, \qquad R_{vu} \;=\; \frac{P(w=v,\, c=u)}{P(w=v)\,P(c=u)}.
\end{align}

Finally, define the \emph{marginal-normalized loading matrix} $\tilde B := D_q^{-1} B \in \mathbb{R}^{V \times K}$ %
and the \emph{topic-mixture covariance}
\begin{align}\label{eq:def_sigma_theta}
\Sigma_\Theta \;:=\; \sum_{d=1}^D p_d\, (\Theta_{\bullet d} - \bar\theta)(\Theta_{\bullet d} - \bar\theta)^\top = G - \bar\theta\bar\theta^\top \;\in\; \mathbb{R}^{K \times K}.
\end{align}
\subsection{Word Embeddings}\label{sec:embeddings_lg}

This section derives our theoretical results on word embeddings. I first introduce a matrix I call the centered probability ratio and show that it is positive semi-definite (PSD) of rank $K-1$ (Proposition~\ref{thm:factorization}). I then construct an explicit embedding $\beta$ whose entries are written directly in the topic-model primitives $B$ and $\Theta$ (Lemma~\ref{thm:topic_primitive}) that factorize the centered probability ratio. Our main word-embedding result (Theorem~\ref{thm:word_embedding_guarantee}) then shows that any embedding $\hat\beta$ matching this factorization inherits the marginal-normalized topic-loading geometry. Simply stated: words with similar topic loadings receive similar embeddings.

 I maintain the following regularity conditions.

\begin{assumption}[Topic regularity]\label{ass:topic_reg} I maintain the following regularity conditions on the topic model throughout:
\begin{enumerate}[label=(\alph*)]
\item\label{ass:rank_B} $\mathrm{rank}\,B = K$;
\item\label{ass:rank_Sigma} $\mathrm{rank}\,\Sigma_\Theta = K-1$;
\item\label{ass:marg_pos} $q_v > 0$ for all $v \in [V]$.
\end{enumerate}
\end{assumption}

The full column rank of $B$ rules out redundant topics. Because the columns of $\Theta$ lie on the simplex $\Delta^{K-1}$, $\Sigma_\Theta\, \mathbf{1}_K = 0$, so $\mathrm{rank}\,\Sigma_\Theta \leq K-1$.  Assumption \ref{ass:topic_reg}\ref{ass:rank_Sigma} is thus a maximal rank condition that is implied by the existence of $K$ documents whose topic mixtures are affinely independent in $\Delta^{K-1}$. The marginal-positivity condition simply rules out terms with zero marginal probability and is used to invert $D_q$.\footnote{Note that, if there exists a term $v$ with $q_v = 0$, this term is not used in any document. Removing any unused terms from the dictionary and using the smaller vocabulary $V'$ immediately implies that $q_v > 0$ for all $v \in [V']$.}

\begin{proposition}[Exact rank-$(K-1)$ factorization]\label{thm:factorization}
Under Assumption~\ref{ass:topic_reg}, the centered probability ratio $(R - \mathbf{1}_V\mathbf{1}_V^\top)$ is PSD of rank $K-1$, such that the top-$(K-1)$ eigendecomposition $R - \mathbf{1}_V\mathbf{1}_V^\top = \Phi\Lambda \Phi^\top$ yields a factor $\beta_{\mathrm{SVD}} \;:=\; \Phi\,\Lambda^{1/2} \;\in\; \mathbb{R}^{V \times (K-1)},$ with
\[
\beta_{\mathrm{SVD}}\ \beta_{\mathrm{SVD}}^\top \;=\; R - \mathbf{1}_V\mathbf{1}_V^\top.
\]
\end{proposition}

\begin{proof}
By \eqref{eq:M_BGB}, $M = BGB^\top$, hence
\begin{align}\label{eq:uncentered_identity}
R = D_q^{-1}(BGB^\top)D_q^{-1} = \tilde B G\tilde B^\top.
\end{align}
Substituting $G = \Sigma_\Theta + \bar\theta\bar\theta^\top$ and using
$\tilde B\bar\theta = D_q^{-1}q = \mathbf{1}_V$,
\begin{align}\label{eq:centering_identity}
R - \mathbf{1}_V \mathbf{1}_V^\top \;=\; \tilde B\, \Sigma_\Theta\, \tilde B^\top.
\end{align}
Since $\Sigma_\Theta$ is a covariance matrix it is PSD, hence
$\tilde B\Sigma_\Theta \tilde B^\top$ is PSD. By Assumptions~\ref{ass:topic_reg}\ref{ass:rank_B} and \ref{ass:marg_pos}, $\tilde B$ has full column rank $K$.It follows that
\[
\mathrm{rank}(R - \mathbf{1}_V\mathbf{1}_V^\top) \;=\; \mathrm{rank}(\Sigma_\Theta) \;=\; K-1
\]
by Assumption~\ref{ass:topic_reg}\ref{ass:rank_Sigma}. The factorization follows.
\end{proof}

Proposition~\ref{thm:factorization} delivers an embedding $\beta_{\mathrm{SVD}}$ from the centered probability ratio $R - \mathbf{1}_V\mathbf{1}_V^\top$ alone, with no reference to the topic-model primitives. The next result gives a second, equivalent representative whose entries are written directly in those primitives ($B$, $\Theta$). %

\begin{lemma}[Topic-primitive representation]\label{thm:topic_primitive}
Under Assumption~\ref{ass:topic_reg}, take the compact eigendecomposition of the topic-mixture covariance defined in \eqref{eq:def_sigma_theta}:
\begin{align*}
\Sigma_\Theta \;=\; \Phi_\Theta\, \Lambda_\Theta\, \Phi_\Theta^\top. \footnotemark
\end{align*}
\footnotetext{Formally, $\Phi_\Theta \in \mathbb{R}^{K \times (K-1)}$ holds orthonormal eigenvectors of the nonzero eigenvalues ($\Phi_\Theta^\top \Phi_\Theta = I_{K-1}$) and $\Lambda_\Theta \in \mathbb{R}^{(K-1)\times(K-1)}$ is the diagonal matrix of corresponding (positive) eigenvalues.}
 For each $v \in [V]$, define
\begin{align}\label{eq:beta_def}
\beta_v \;:=\; \frac{\Lambda_\Theta^{1/2}\, \Phi_\Theta^\top\, B_{v\bullet}^\top}{q_v} \;\in\; \mathbb{R}^{K-1}.
\end{align}
Then, the matrix $\beta \in \mathbb{R}^{V \times (K-1)}$ with rows $\beta_v^\top$ satisfies
\begin{align}\label{eq:first_embedding_result}
R - \mathbf{1}_V\mathbf{1}_V^\top \;=\; \beta\,\beta^\top,
\end{align}
and is related to $\beta_{\mathrm{SVD}}$ from Proposition~\ref{thm:factorization} by an orthogonal transformation: there exists an orthogonal $Q$ ($Q^\top Q = I_{K-1}$) with $\beta = \beta_{\mathrm{SVD}}\, Q$.
\end{lemma}

\begin{proof}
For any $v, u \in [V]$,
\begin{align*}
\beta_v^\top \beta_u \;=\; \frac{B_{v\bullet}\, \Phi_\Theta \Lambda_\Theta \Phi_\Theta^\top\, B_{u\bullet}^\top}{q_v q_u} \;=\; \frac{B_{v\bullet}\, \Sigma_\Theta\, B_{u\bullet}^\top}{q_v q_u} \;=\; (\tilde B\,\Sigma_\Theta\,\tilde B^\top)_{vu} \;=\; (R - \mathbf{1}_V\mathbf{1}_V^\top)_{vu},
\end{align*}
where the last step follows from \eqref{eq:centering_identity}. Both $\beta$ and $\beta_{\mathrm{SVD}}$ are full-column-rank $V \times (K-1)$ factors of the same rank-$(K-1)$ PSD matrix, so they coincide up to an orthogonal transformation.
\end{proof}

The inner product $\beta_v^\top \beta_u = \frac{P(w=v,\,c=u)}{P(w=v)\,P(c=u)} - 1$ can be interpreted as the lift of the joint above independence: zero if $w$ and $c$ are independent, positive if they co-occur more than under independence, and negative if less. %
I next show that \emph{any} $(K-1)$-dimensional embedding matching this factorization inherits the marginal-normalized topic-loading geometry.

\begin{theorem}[Topic-loading geometry]\label{thm:word_embedding_guarantee}
Suppose Assumption~\ref{ass:topic_reg} holds and let $\hat\beta \in \mathbb{R}^{V \times (K-1)}$ satisfy $\hat\beta\,\hat\beta^\top = R - \mathbf{1}_V\mathbf{1}_V^\top$. Then, for all $v, u \in [V]$,
\begin{align}
\|\hat\beta_v - \hat\beta_u\|^2 \;=\; (\tilde B_{v\bullet} - \tilde B_{u\bullet})^\top \Sigma_\Theta\, (\tilde B_{v\bullet} - \tilde B_{u\bullet}). \label{eq:beta_distance}
\end{align}
In particular, $B_{v\bullet} = \lambda B_{u\bullet}$ for some $\lambda > 0$ implies $\hat\beta_v = \hat\beta_u$: words with proportional topic loadings receive identical embeddings.
\end{theorem}

\begin{proof}
By Lemma~\ref{thm:topic_primitive}, the topic-primitive representative $\beta$ satisfies $\beta\beta^\top = R - \mathbf{1}_V\mathbf{1}_V^\top = \hat\beta\hat\beta^\top$, so both are full-column-rank $V \times (K-1)$ factors of the same rank-$(K-1)$ PSD matrix and coincide up to an orthogonal transformation: $\hat\beta = \beta Q$ with $Q^\top Q = I_{K-1}$. From the closed form, $\beta_v - \beta_u = \Lambda_\Theta^{1/2} \Phi_\Theta^\top (\tilde B_{v\bullet} - \tilde B_{u\bullet})$, so
\[
\|\beta_v - \beta_u\|^2 \;=\; (\tilde B_{v\bullet} - \tilde B_{u\bullet})^\top \Phi_\Theta \Lambda_\Theta \Phi_\Theta^\top (\tilde B_{v\bullet} - \tilde B_{u\bullet}) \;=\; (\tilde B_{v\bullet} - \tilde B_{u\bullet})^\top \Sigma_\Theta\, (\tilde B_{v\bullet} - \tilde B_{u\bullet}), 
\]
using $\Sigma_\Theta = \Phi_\Theta \Lambda_\Theta \Phi_\Theta^\top$. Finally,  $\|\hat\beta_v - \hat\beta_u\| = \|\hat\beta_v Q - \hat\beta_u Q\| = \|(\hat\beta_v - \hat\beta_u) Q \| = \|\beta_v - \beta_u\|$ since orthogonality of $Q$ preserves the Euclidean norm. This completes the proof of \eqref{eq:beta_distance}. The proportional case follows from $B_{v\bullet} = \lambda B_{u\bullet} \Rightarrow \tilde B_{v\bullet} = \tilde B_{u\bullet} \Rightarrow \beta_v = \beta_u \Rightarrow \hat\beta_v = \hat\beta_u$.
\end{proof}

I call the metric on the right of \eqref{eq:beta_distance} --- the $\Sigma_\Theta$-weighted distance between marginal-normalized topic loadings --- the \emph{topic-loading geometry}. Theorem~\ref{thm:word_embedding_guarantee} provides an algorithm-agnostic reading: any procedure that matches the centered probability ratio $R - \mathbf{1}_V\mathbf{1}_V^\top$ in dimension $K-1$ recovers the topic-loading geometry, regardless of the construction. The direct example is the rank-$(K-1)$ singular value decomposition of $R - \mathbf{1}_V\mathbf{1}_V^\top$ (Proposition~\ref{thm:factorization}). I next analyze some of the word embeddings commonly used in the literature.

\begin{remark}[Dimensions above $K-1$]\label{rem:dimension_slack}
Theorem~\ref{thm:word_embedding_guarantee} is stated at $r = K-1$, but nothing in it requires equality. Let $\hat\beta \in \mathbb{R}^{V \times r}$ with $r \geq K-1$ satisfy $\hat\beta\hat\beta^\top = R - \mathbf{1}_V\mathbf{1}_V^\top$. The right-hand side has rank $K-1$ by Proposition~\ref{thm:factorization}, so $\hat\beta = \beta Q$ for some $Q \in \mathbb{R}^{(K-1) \times r}$ with orthonormal rows, $QQ^\top = I_{K-1}$, and
\[
\|\hat\beta_v - \hat\beta_u\|^2 \;=\; (\beta_v - \beta_u)^\top QQ^\top (\beta_v - \beta_u) \;=\; \|\beta_v - \beta_u\|^2,
\]
so \eqref{eq:beta_distance} holds unchanged. The proof above is the special case $r = K-1$, where $Q$ is square and orthogonal. The extra $r - (K-1)$ coordinates are identically zero and drop out of every distance.

Thus, any $r \ge K-1$ delivers the same geometry. This matters because $K$ is rarely known in practice. %
 Section~\ref{sec:application} sets $r = 50$ on this basis.
\end{remark}

\begin{remark}[Identification-invariance]\label{rem:identification_invariance}
The factorization $\Pi = B\Theta$ is not unique absent further structure: for any invertible $K \times K$ matrix $A$ (that preserves the column-stochasticity of both factors), the reparameterization $B' = BA$, $\Theta' = A^{-1}\Theta$ generates the same observable joint distribution (see, e.g., \cite{donoho2003does}, \cite{fu2019nonnegative}). However, under any such reparameterization the quadratic form in Theorem~\ref{thm:word_embedding_guarantee} transforms as
\begin{align*}
(\tilde B'_{v\bullet} - \tilde B'_{u\bullet})^\top \Sigma'_\Theta\,(\tilde B'_{v\bullet} - \tilde B'_{u\bullet})
&= (\tilde B_{v\bullet} - \tilde B_{u\bullet})^\top A\, (A^{-1}\Sigma_\Theta A^{-\top})\, A^\top (\tilde B_{v\bullet} - \tilde B_{u\bullet}) \\
&= (\tilde B_{v\bullet} - \tilde B_{u\bullet})^\top \Sigma_\Theta\, (\tilde B_{v\bullet} - \tilde B_{u\bullet}),
\end{align*}
so it is invariant. Theorem~\ref{thm:word_embedding_guarantee} therefore holds for any valid factorization, and the geometry it characterizes is identification-free. %
\end{remark}

\subsection{Relationship to Standard Algorithms}\label{sec:other_embeddings}

Word2Vec \citep{mikolov2013efficient, mikolov2013distributed} and GloVe \citep{pennington2014glove} are some of the most widely used word embeddings in practice, and this section asks what they recover under our model. I treat three objectives: skip-gram with a full softmax (the idealized version that is costly to estimate), its negative-sampling approximation SGNS \citep{levy2014neural} (which is what practitioners usually run), and GloVe.\footnote{Word2Vec's other architecture, CBOW, generally admits no closed-form target, and I return to it at the end of the section.}

None of these fit $R - \mathbf{1}_V\mathbf{1}_V^\top$.  Each scores a word--context pair $(v,u)$ by an inner product $X_{vu} = w_v^\top \tilde w_u$ between a target vector $w_v$ and a context vector $\tilde w_u$, the rows of $W, \tilde W \in \mathbb{R}^{V \times r}$, and fits that score to the corpus. They differ in only three ingredients: a \emph{response} $t_{vu}$ read off the corpus, a strictly convex \emph{generator} $\phi$ that fixes the loss, and \emph{weights} $\omega_{vu}$ on pairs. Where the rank constraint of the inner product does not bind, the fitted scores are $\phi'(t_{vu})$ entrywise. I therefore call that matrix the algorithm's \emph{target}. Table~\ref{tab:targets} lists the three ingredients and the target for each algorithm,  alongside $\beta_{\mathrm{SVD}}$ from the previous section. I derive its entries in Appendix \ref{app-app:log_approximation}. %

\begin{table}[tb!]
\centering
\footnotesize
\setlength{\tabcolsep}{3pt}
\begin{tabular}{@{}llllll@{}}
\toprule
& Response $t_{vu}$ & Generator $\phi$ & Weight $\omega_{vu}$ & Target $\phi'(t_{vu})$ & Free offsets \\
\midrule
$\beta_{\mathrm{SVD}}$ & $R_{vu} - 1$ & squared & $1$ & $R - \mathbf{1}_V\mathbf{1}_V^\top$ & --- \\
Full softmax & $P(w{=}v \mid c{=}u)$ & entropy & $q_u$ & $\log R + \log q\,\mathbf{1}_V^\top$ & $\mathbf{1}_V g^\top$ \\
SGNS & $R_{vu}/(R_{vu}+\nu)$ & binary entropy & $q_vq_u(R_{vu}+\nu)$ & $\log R - \log \nu\,\mathbf{1}_V\mathbf{1}_V^\top$ & --- \\
GloVe & $\log M_{vu}$ & squared & $f(M_{vu})$ & $\log R$ & $b\mathbf{1}_V^\top + \mathbf{1}_V\tilde b^\top$ \\
\bottomrule
\end{tabular}
\caption{The four objectives as weighted low-rank fits. The response $t_{vu}$ and the generator $\phi$ determine the target $\phi'(t_{vu})$; the weights $\omega_{vu}$ do not. The generators squared, entropy and binary entropy give squared-error, Kullback--Leibler and Bernoulli losses. ``Free offsets'' are the row- or column-constant terms the objective leaves undetermined: the softmax column gauge $g$ and GloVe's per-word biases $b, \tilde b$.  $\nu \geq 1$ is the negative-sampling rate and $f$ is GloVe's weighting function. Row one attains its target exactly at rank $K-1$ (Proposition~\ref{thm:factorization}); the
other three attain theirs only where the rank constraint does not bind. For more detail, see Appendix \ref{app-app:log_approximation}.}
\label{tab:targets}
\end{table}

Since none of these targets is $R - \mathbf{1}_V\mathbf{1}_V^\top$, Theorem~\ref{thm:word_embedding_guarantee} does not apply to the embeddings that fit them. 
However, I still obtain the following result on how word embeddings vary with the loadings, analogous to Theorem~\ref{thm:word_embedding_guarantee} for the SGNS target.

\begin{assumption}[Word co-occurrence]\label{ass:R_pos}
$R_{vu} > 0$ for all $v, u \in [V]$.
\end{assumption}

Assumption~\ref{ass:R_pos} requires every pair of words to co-occur and guarantees that every entry of $\log R$ is finite. 
Write $R_{\min} := \min_{v,u} R_{vu}$, $R_{\max} := \max_{v,u} R_{vu}$ and $\kappa_R := R_{\max}/R_{\min}$, so that $\log\kappa_R$ is the range of entries in $\log R$. 

  \begin{proposition}[Recovery from the SGNS target]\label{prop:oe_general}
  Let Assumptions~\ref{ass:topic_reg} and~\ref{ass:R_pos} hold. Suppose $\hat W\hat{\tilde W}^\top = \log R -\log \nu\,\mathbf{1}_V\mathbf{1}_V^\top$, the SGNS target of Table~\ref{tab:targets}, with
  $\hat{\tilde W} \in \mathbb{R}^{V\times r}$ of full column rank. Then, for all $v, v' \in [V]$,
  \begin{align}\label{eq:oe_embed_equiv}
  \frac{\sigma_{\min}(\beta)}{R_{\max}\,\sigma_{\max}(\hat{\tilde W})}\,\|\beta_v - \beta_{v'}\|
  \;\le\; \|\hat w_v - \hat w_{v'}\|
  \;\le\; \frac{\sigma_{\max}(\beta)}{R_{\min}\,\sigma_{\min}(\hat{\tilde W})}\,\|\beta_v - \beta_{v'}\| ,
  \end{align}
where $\|\beta_v - \beta_{v'}\|^2 = (\tilde B_{v\bullet} - \tilde B_{v'\bullet})^\top \Sigma_\Theta\, (\tilde B_{v\bullet} - \tilde B_{v'\bullet})$ is the topic-loading distance of Theorem~\ref{thm:word_embedding_guarantee}.
 
In particular, $B_{v\bullet} = \lambda B_{v'\bullet}$ for some $\lambda > 0$ implies $\hat w_v = \hat w_{v'}$.
  \end{proposition}

  \begin{proof}
Since the shift $\log\nu\,\mathbf{1}_V\mathbf{1}_V^\top$ is constant and cancels in row differences, $(\log R)_{v\bullet} - (\log R)_{v'\bullet} = (\hat w_v - \hat w_{v'})\hat{\tilde W}^\top$. All entries of $R$ lie in $[R_{\min}, R_{\max}]$, so the mean value theorem applied to $\log$ gives $|R_{vu} - R_{v'u}|/R_{\max} \le |\log R_{vu} - \log R_{v'u}| \le |R_{vu} - R_{v'u}|/R_{\min}$ for each $u$.
Thus, for the entire vector:
\[
\frac{1}{R_{\max}}\,\|R_{v\bullet} - R_{v'\bullet}\|
\;\le\; \|(\hat w_v - \hat w_{v'}) \hat{\tilde W}^\top\|
\;\le\; \frac{1}{R_{\min}}\,\|R_{v\bullet} - R_{v'\bullet}\|.
\]

 By Lemma~\ref{thm:topic_primitive}, $\beta\beta^\top = R - \mathbf{1}_V\mathbf{1}_V^\top$. All rows of
  $\mathbf{1}_V\mathbf{1}_V^\top$ are equal, so they cancel in the row difference:
  $R_{v\bullet} - R_{v'\bullet} = (\beta_v - \beta_{v'})^\top\beta^\top$. Since $\beta$ has full column
  rank $K-1$ (Proposition~\ref{thm:factorization}), $\sigma_{\min}(\beta) > 0$ and
  \[
  \sigma_{\min}(\beta)\,\|\beta_v - \beta_{v'}\| \;\le\; \|R_{v\bullet} - R_{v'\bullet}\|
  \;\le\; \sigma_{\max}(\beta)\,\|\beta_v - \beta_{v'}\| .
  \]
Combining with the display above gives 
\begin{align*}
\frac{\sigma_{\min}(\beta)}{R_{\max}}\,\|\beta_v - \beta_{v'}\|
\;\le\; \|(\hat w_v - \hat w_{v'}) \hat{\tilde W}^\top\|
\;\le\; \frac{\sigma_{\max}(\beta)}{R_{\min}}\,\|\beta_v - \beta_{v'}\|.
\end{align*}
Finally, full column rank of $\hat{\tilde W}$ gives $\sigma_{\min}(\hat{\tilde W})\|x\| \le \|x\hat{\tilde W}^\top\| \le \sigma_{\max}(\hat{\tilde W})\|x\|$ for every $x$. Substituting into the display above gives \eqref{eq:oe_embed_equiv}. 

 The proportional case follows from $B_{v\bullet} = \lambda B_{v'\bullet}
  \Rightarrow \tilde B_{v\bullet} = \tilde B_{v'\bullet} \Rightarrow \beta_v = \beta_{v'}$, which makes the
  right-hand side of \eqref{eq:oe_embed_equiv} zero.
  \end{proof}

Proposition~\ref{prop:oe_general} is the analogue of Theorem~\ref{thm:word_embedding_guarantee} for the SGNS target: any $\hat W, \hat{\tilde W}$ that attain that target induce a metric equivalent to the one $\beta$ induces, reproducing the topic-loading geometry up to a distortion of at most $\kappa_R\,\mathrm{cond}(\beta)\,\mathrm{cond}(\hat{\tilde W})$. In particular, words with proportional topic loadings receive identical embeddings under any such factorization, exactly as in Theorem~\ref{thm:word_embedding_guarantee}. 

\begin{remark}[Target versus embedding]\label{rem:target_vs_embedding}
Proposition~\ref{prop:oe_general} is a statement about the SGNS target, not about the trained embedding. Since generically the SGNS target is full rank, the trained embedding can only approximate the target, and is instead the weighted low-rank fit of Table \ref{tab:targets}, with the weights determining which approximation \citep{levy2014neural}. The proposition's practical content therefore depends on the quality of this approximation, and hence on the spectrum of the target.\footnote{To bound the spectrum theoretically requires strong assumptions I deem unlikely to hold in practice. Two examples: i) Call $v$ an \emph{anchor word} for topic $k$ if $B_{vk} > 0$ and $B_{vk'} = 0$ for all $k' \neq k$ \citep{donoho2003does, arora2013practical}. If every word is an anchor, then $\log R$, and hence the target, has at most rank $K$; this is the population analogue of the block structure \citet{li2023transformers} obtain for transformers under disjoint topics.
ii) Write $\varrho := \|R - \mathbf{1}_V\mathbf{1}_V^\top\|_\infty$ for the maximum distance from independence. Then, $\varrho < 1$ gives $\log R = \beta\beta^\top + O(\varrho^2)$ entrywise, so the target is of rank at most $K$ up to an $O(V\varrho^2)$ perturbation. 
  }

I explore this further in Section~\ref{sec:simulations}, where I compute the spectrum exactly under our simulation designs and find the target approximately low rank. In our application in Section \ref{sec:application} I simply set $r=50$ to ensure $r>K$ and that the rank constraint is not too binding.
\end{remark}

\begin{remark}[The other log targets]\label{rem:other_targets}
Proposition~\ref{prop:oe_general} uses only that the SGNS target is $\log R$ up to a constant, so it covers any target that is $\log R$ up to an offset constant down each column. However, both full-softmax skip-gram and GloVe carry an offset constant along each \emph{row}, which does not cancel in a row difference. For example, for full softmax that offset is $\log q\,\mathbf{1}_V^\top$, so $R_{v\bullet} = R_{v'\bullet}$ gives
\[
(\hat w_v - \hat w_{v'})\hat{\tilde W}^\top \;=\; (\log q_v - \log q_{v'})\,\mathbf{1}_V^\top ,
\]
and the two embeddings coincide if and only if $q_v = q_{v'}$: same-loading words agree up to a single marginal direction. Lemma~\ref{app-lem:index} gives the general statement: for any entrywise transform of $R$ and any offsets constant along a row or a column, words with proportional loadings have target rows differing by the row-offset difference alone.
\end{remark}

CBOW \citep{mikolov2013efficient} predicts the target word from the \emph{average} of its context embeddings. With a one-word context the average is trivial: under negative sampling its target is again $\log R - \log\nu\,\mathbf{1}_V\mathbf{1}_V^\top$, so Proposition~\ref{prop:oe_general} applies verbatim. The averaging is what removes the theory: with more than one context word the score depends on the context only through the averaged embedding, so no closed-form $V\times V$ target exists. Given its prevalence in the literature, I nevertheless treat CBOW empirically in Sections~\ref{sec:simulations}--\ref{sec:application}.

\subsection{Document Embeddings}\label{sec:doc_embeddings}

Word embeddings can be aggregated to create document representations. A natural approach is to represent each document by the average of its word embeddings:
\begin{align}\label{eq:doc_embedding}
\bar{h}_d = \frac{1}{N_d} \sum_{j=1}^{N_d} h_{w_{dj}},
\end{align}
where $h_w$ is the embedding assigned to word $w$---for instance, the $\hat\beta_w \in \mathbb{R}^{K-1}$ of Theorem~\ref{thm:word_embedding_guarantee}, with explicit representative $\beta_w$ given by Lemma~\ref{thm:topic_primitive}---and $w_{dj}$ is the $j$-th word in document $d$.

First, note that averaging is the natural aggregator under the model of Section~\ref{sec:generative}, since the tokens $w_{d,1},\dots,w_{d,N_d}$ are i.i.d.\ draws from the document's word distribution $\Pi_{\bullet d}$. Averaging is also the operation behind a broad class of document- and sentence-embedding methods that are used in practice. Unweighted averages of static word vectors are a widely used baseline \citep{wieting2016towards, iyyer2015deep}; \citet{arora2017simple} reweight the average by \emph{smooth inverse frequency} (down-weighting frequent words), and \citet{joulin2017bag} average word- and $n$-gram vectors for text classification. The same pooling extends to contextual models, where transformer sentence encoders mean-pool token embeddings into a fixed-length vector \citep{reimers2019sentence}.\footnote{The main alternative learns document vectors directly rather than by averaging \citep[Paragraph Vector, or doc2vec;][]{le2014distributed}, which I do not treat here.}

Intuitively, under our topic model, documents with similar topic mixture $\theta_{d}$ should have similar document embeddings, since they draw words from similar distributions over the vocabulary. I now formalize this intuition.
Define the \emph{expected document embedding} for a document with topic mixture $\theta_d$ as:
\begin{align}\label{eq:expected_doc_embedding}
\mu_d = \E[\bar{h}_d | \theta_d] = \sum_{v=1}^V P(w = v | d) h_v = \sum_{v=1}^V (B \Theta_{\bullet d})_v  h_v.
\end{align}
Further, because the words within document $d$, $w_{d,1}, \ldots, w_{d,N_d}$ are i.i.d.\ conditional on $\theta_d$, clearly $\bar h_d = \frac{1}{N_d}\sum_j h_{w_{dj}} \xrightarrow{p} \mu_d$ as $N_d \to \infty$ by the law of large numbers.
Finally, define the \emph{topic centroid embeddings}
\begin{align}\label{eq:topic_centroids}
c_k \;:=\; \sum_{v=1}^V B_{vk}\, h_v, \qquad k=1,\ldots, K,
\end{align}
and let $C \equiv [c_1\;\cdots\;c_K]$ denote the matrix with the centroids as columns. I obtain the following result.

\begin{proposition}[Linearity in the topic mixture]\label{prop:doc_embedding_general_freq}
Suppose $\Pi = B\Theta$ and the topic model is identified.\footnote{The anchor-word condition --- every topic has at least one word exclusive to it --- is one sufficient condition for identification (\cite{donoho2003does, arora2013practical,fu2019nonnegative}), also see the discussion in \cite{freyaldenhoven2025testability}. However, it is not necessary; see  the recent work of \cite{chen2022learning} that uses the \emph{sufficiently-scattered} condition in \cite{huang2013non} and \cite{huang2016anchor}.} Then, for every document $d$:
\begin{enumerate}[label=(\roman*)]
\item the expected document embedding is a linear function of the topic mixture,
\begin{align}\label{eq:mu_linear}
\mu_d \;=\; C\, \Theta_{\bullet d};
\end{align}
\item $\mu_d$ lies in the centroid simplex, $\mu_d \in \mathrm{conv}\{c_1,\ldots,c_K\}$.
\end{enumerate}
\end{proposition}

\begin{proof}
For (i), simply expanding the definition of $\mu_d$ in \eqref{eq:expected_doc_embedding} and rearranging the sums yields
\begin{align*}
\mu_d \;&=\; \sum_{v=1}^V (B\Theta_{\bullet d})_v\, h_v \;=\; \sum_{v=1}^V \sum_{k=1}^K B_{vk}\, \theta_{dk}\, h_v \\
\;&=\; \sum_{k=1}^K \theta_{dk}\, \sum_{v=1}^V B_{vk}\, h_v \;=\; \sum_{k=1}^K \theta_{dk}\, c_k = C\,\Theta_{\bullet d}.
\end{align*}
For (ii), the topic mixture satisfies $\theta_{dk} \geq 0$ and $\sum_k \theta_{dk} = 1$, which immediately implies $\mu_d \in \mathrm{conv}\{c_1,\ldots,c_K\}$.
\end{proof}

Note that, beyond the topic-model decomposition $\Pi = B\Theta$, Proposition~\ref{prop:doc_embedding_general_freq} imposes no further conditions on the word embeddings $h_v$: the linear decomposition $\mu_d = C\,\Theta_{\bullet d}$ and the simplex containment follow from linearity of averaging and hold for arbitrary---even random---embeddings. %
Proposition~\ref{prop:doc_embedding_general_freq} characterizes the \emph{space} in which document embeddings live: a low-dimensional simplex spanned by the topic centroids, with each document's position encoding its topic mixture. Corollary~\ref{cor:doc_embedding_metric} below characterizes the \emph{metric} on that space---the pairwise distances between document embeddings when the word embeddings satisfy the factorization of Theorem~\ref{thm:word_embedding_guarantee}.\footnote{Note that linearity of $\mu_d = C\,\Theta_{\bullet d}$ already implies a Lipschitz bound between expected document embeddings and topic mixtures: for any two documents,
\begin{align}\label{eq:doc_lipschitz}
\|\mu_d - \mu_{d'}\|_2 \;\leq\; \|C\|_{\mathrm{op}}\,\|\Theta_{\bullet d} - \Theta_{\bullet d'}\|_2,
\end{align}
with Lipschitz constant equal to the operator norm of the centroid matrix. Corollary~\ref{cor:doc_embedding_metric} sharpens this bound once the word embeddings satisfy $HH^\top = R - \mathbf{1}_V\mathbf{1}_V^\top$.}

\begin{corollary}[Document-embedding metric]\label{cor:doc_embedding_metric}
Suppose Assumption~\ref{ass:topic_reg} holds, $\Pi = B\Theta$, and the topic model is identified. If the word embeddings  $H=[h_1, \ldots, h_V]^\top$  satisfy $HH^\top = R - \mathbf{1}_V\mathbf{1}_V^\top$ (such as the explicit construction in Lemma~\ref{thm:topic_primitive}), the pairwise distance between expected document embeddings is given by
\begin{align}\label{eq:doc_metric_observable}
\|\mu_d - \mu_{d'}\|_2^2 \;=\; (\Pi_{\bullet d} - \Pi_{\bullet d'})^\top \,(R - \mathbf{1}_V\mathbf{1}_V^\top)\, (\Pi_{\bullet d} - \Pi_{\bullet d'}).
\end{align}
Equivalently,
\begin{align}\label{eq:doc_metric_topic}
\|\mu_d - \mu_{d'}\|_2^2 \;=\; (\Theta_{\bullet d} - \Theta_{\bullet d'})^\top \,G_B\, \Sigma_\Theta\, G_B\, (\Theta_{\bullet d} - \Theta_{\bullet d'}),
\end{align}
where $G_B := B^\top D_q^{-1} B$. Moreover, this distance is strictly positive whenever $\Theta_{\bullet d} \neq \Theta_{\bullet d'}$.
\end{corollary}

\begin{proof}
Using $\mu_d = H^\top \Pi_{\bullet d}$, we have
\begin{align*}
\|\mu_d - \mu_{d'}\|_2^2 &=\; (H^\top\Pi_{\bullet d} - H^\top \Pi_{\bullet d'})^\top(H^\top\Pi_{\bullet d} - H^\top \Pi_{\bullet d'})\\
&=(\Pi_{\bullet d} - \Pi_{\bullet d'})^\top\, H H^\top\, (\Pi_{\bullet d} - \Pi_{\bullet d'}) \\
&=\; (\Pi_{\bullet d} - \Pi_{\bullet d'})^\top \,(R - \mathbf{1}_V\mathbf{1}_V^\top)\, (\Pi_{\bullet d} - \Pi_{\bullet d'}).
\end{align*}
Using \eqref{eq:centering_identity} to substitute $R - \mathbf{1}_V\mathbf{1}_V^\top = D_q^{-1} B\, \Sigma_\Theta\, B^\top D_q^{-1}$ and using $\Pi_{\bullet d} = B\,\Theta_{\bullet d}$ yields
\begin{align*}
\|\mu_d - \mu_{d'}\|_2^2 &=\; (\Pi_{\bullet d} - \Pi_{\bullet d'})^\top \,(R - \mathbf{1}_V\mathbf{1}_V^\top)\, (\Pi_{\bullet d} - \Pi_{\bullet d'}) \\
&=\; (\Theta_{\bullet d} - \Theta_{\bullet d'})^\top \,B^\top D_q^{-1} B\, \Sigma_\Theta\, B^\top D_q^{-1} B\, (\Theta_{\bullet d} - \Theta_{\bullet d'}).
\end{align*}

For the final claim, write $x := \Theta_{\bullet d} - \Theta_{\bullet d'} \neq 0$, and note $\mathbf{1}_K^\top x = 0$ because both mixtures lie on the simplex. Since $\mathrm{rank}\,\Sigma_\Theta = K-1$ and $\Sigma_\Theta\mathbf{1}_K = 0$, we have $\ker\Sigma_\Theta = \mathrm{span}(\mathbf{1}_K)$, so \eqref{eq:doc_metric_topic} vanishes if and only if $G_B x = \lambda\mathbf{1}_K$ for some $\lambda \in \mathbb{R}$. In that case $x = \lambda G_B^{-1}\mathbf{1}_K$, and $\mathbf{1}_K^\top x = 0$ gives $\lambda\,\mathbf{1}_K^\top G_B^{-1}\mathbf{1}_K = 0$. Under Assumption~\ref{ass:topic_reg}, $G_B$ is positive definite, hence so is $G_B^{-1}$ and $\mathbf{1}_K^\top G_B^{-1}\mathbf{1}_K > 0$; therefore $\lambda = 0$ and $G_B x = 0$, so $x = 0$, a contradiction.
\end{proof}
Equation~\eqref{eq:doc_metric_observable} expresses the document-embedding distance through the centered probability ratio $R - \mathbf{1}_V\mathbf{1}_V^\top$ and the word-distribution profiles $\Pi_{\bullet d}$ with no reference to the topic-model primitives.\footnote{Note that, because $R - \mathbf{1}_V\mathbf{1}_V^\top$ has rank $K-1$ (Proposition~\ref{thm:factorization}), the distance depends only on the projection of $\Pi_{\bullet d} - \Pi_{\bullet d'}$ onto the $(K-1)$-dimensional topic subspace; word-frequency variation orthogonal to it leaves the embedding distance unchanged.}
The equivalent form \eqref{eq:doc_metric_topic} rewrites the same distance in latent topic coordinates as weighted topic-mixture differences. Those weights depend on the topic-mixture covariance $\Sigma_\Theta$ and the overlap among the topic loadings, $G_B$.
Taken together, Proposition~\ref{prop:doc_embedding_general_freq} and Corollary~\ref{cor:doc_embedding_metric} show that the document embedding is a linear---and, under Assumption~\ref{ass:topic_reg}, invertible---encoding of the latent topic mixture $\Theta_{\bullet d}$. In other words, $\mu_d$ is a sufficient statistic for $\Theta_{\bullet d}$, so a cluster of nearby embeddings is a set of documents with similar mixtures, and the map from an embedding-space centroid to a mixture is well defined. It requires no assumption about how documents are distributed over the simplex.

I next apply the results in this section to two uses of embeddings: one \emph{groups} units by latent type, where I show that $k$-means in embedding space\footnote{Throughout, I call the Euclidean space in which the word embeddings $h_v$, document embeddings $\mu_d$, and centroids $c_k$ live the embedding space — $\mathbb{R}^{K-1}$ under the exact factorization of Proposition~\ref{thm:factorization}.} recovers structure on the topic mixtures, Section~\ref{sec:clustering}.  The other \emph{conditions} on them, where I show that adjusting for the embedding is equivalent to adjusting for the topic mixture, Section~\ref{sec:conditioning}.

\subsection{Grouping: Clustering Documents}\label{sec:clustering}

I next ask when $k$-means clustering recovers meaningful structure. I begin with the simplest case---cluster centers at the topic centroids---and then generalize.

\begin{lemma}[Voronoi-cell membership for concentrated documents]\label{lem:voronoi_membership}
Suppose Assumption~\ref{ass:topic_reg} holds, the topic model is identified, and the word embeddings satisfy $HH^\top = R - \mathbf{1}_V\mathbf{1}_V^\top$ (such as the explicit construction in Lemma~\ref{thm:topic_primitive}). Let $\mathcal{V}_k$ denote the Voronoi cell of $c_k$ under the Euclidean metric in embedding space,\footnote{Given a simplex with vertices $c_1,\ldots,c_K$, the \emph{Voronoi cell} of $c_k$ is $\mathcal{V}_k = \{x : \|x - c_k\|_2 \leq \|x - c_{k'}\|_2 \text{ for all } k' \neq k\}$, the points at least as close to $c_k$ as to any other vertex. The cells $\{\mathcal{V}_1,\ldots,\mathcal{V}_K\}$ partition the simplex into convex polytopes meeting at the perpendicular bisectors of the vertex pairs; assigning each point to its nearest vertex is the nearest-centroid rule used by $k$-means with fixed centers.} and $k^*(d) = \arg\max_k \theta_{dk}$. Define
\begin{align}\label{eq:eta_def}
\eta \;:=\; \frac{1}{2} \frac{\min_{k \neq k'} \|c_k - c_{k'}\|_2}{\max_{k \neq k'} \|c_k - c_{k'}\|_2} \;\in\; (0, 1/2].
\end{align}
Then, for any document $d$ with dominant-topic weight $\theta_{d, k^*(d)} > 1 - \eta$, $\mu_d \in \mathcal{V}_{k^*(d)} \cap \mathrm{conv}\{c_1,\ldots,c_K\}$.
\end{lemma}

\begin{proof} %
The centroids are distinct: $c_k = C e_k$ is the embedding of a pure-topic document, and $e_k \neq e_{k'}$ for $k \neq k'$, so Corollary~\ref{cor:doc_embedding_metric} gives $\|c_k - c_{k'}\|_2 > 0$. Hence $\eta \in (0, 1/2]$ is well defined.

Next, fix document $d$ and let $k^* = k^*(d)$. Write $\Theta_{\bullet d} = (1-\delta) e_{k^*} + \delta z$ with $\delta = 1 - \theta_{d, k^*} \in [0, 1]$ and $z \in \Delta_{K-1}$. With  $\mu_d = C\,\Theta_{\bullet d}$ (Proposition~\ref{prop:doc_embedding_general_freq}),
\begin{align*}
\mu_d - c_{k^*} &= C \Theta_{\bullet d} - c_{k^*} =  C (1-\delta) e_{k^*} + C \delta z - c_{k^*} = (1-\delta) c_{k^*} + C \delta z - c_{k^*} \\
&= \delta(Cz - c_{k^*}).
\end{align*}
Since $z \mapsto \|Cz - c_{k^*}\|_2$ is convex on the simplex $\Delta_{K-1}$, its maximum is attained at a vertex; hence $\|Cz - c_{k^*}\|_2 \leq \max_{k'} \|c_{k'} - c_{k^*}\|_2 \leq C^{\max}$, where $C^{\max} := \max_{k \neq k'} \|c_k - c_{k'}\|_2$. Therefore, since $\delta = 1 - \theta_{d, k^*}$,
\begin{align}\label{eq:slack}
\|\mu_d - c_{k^*}\|_2 \;\leq\; \delta\, C^{\max} \;=\; (1-\theta_{d, k^*})\,C^{\max}.
\end{align}
For any $k' \neq k^*$, the triangle inequality gives
\begin{align*}
\|\mu_d - c_{k'}\|_2 \;\geq\; \|c_{k^*} - c_{k'}\|_2 - \|\mu_d - c_{k^*}\|_2 \;\geq\; \Delta^* - \delta C^{\max},
\end{align*}
where $\Delta^* := \min_{k \neq k'} \|c_k - c_{k'}\|_2 > 0$. 

It follows that, whenever $\delta C^{\max} < \Delta^* - \delta C^{\max}$, or equivalently, $\delta < \Delta^*/(2 C^{\max}) = \frac{\min_{k \neq k'} \|c_k - c_{k'}\|_2}{2 \max_{k \neq k'} \|c_k - c_{k'}\|_2} = \eta$, we have:
\begin{align*}
\|\mu_d - c_{k^*}\|_2 \le \delta C^{\max} < \Delta^* - \delta C^{\max} \le \|\mu_d - c_{k'}\|_2 \quad \text{for all } k' \neq k^*.
\end{align*}
Since $\delta<\eta$, and $\theta_{d, k^*} > 1 - \eta$ are equivalent, $\theta_{d, k^*} > 1 - \eta$ thus implies $\mu_d \in \mathcal{V}_{k^*}$. Finally, since, by Proposition~\ref{prop:doc_embedding_general_freq}, $\mu_d \in \mathrm{conv}\{c_1,\ldots,c_K\}$ for all $d$, this completes the proof.
\end{proof}

In words, Lemma~\ref{lem:voronoi_membership} states that documents whose mass is concentrated on a single topic vertex of $\Delta_{K-1}$ end up in the Voronoi cell of the corresponding centroid.
The constant $\eta$ in Lemma~\ref{lem:voronoi_membership} is defined as half the ratio of minimum to maximum pairwise centroid distance. $\eta = 1/2$ is achieved when the centroids are equidistant, i.e., when they form an equilateral $K$-simplex; For the other extreme, $\eta \to 0$ when some pair of centroids becomes much closer than the others.
I illustrate in Figure \ref{fig:voronoi_lemma}.
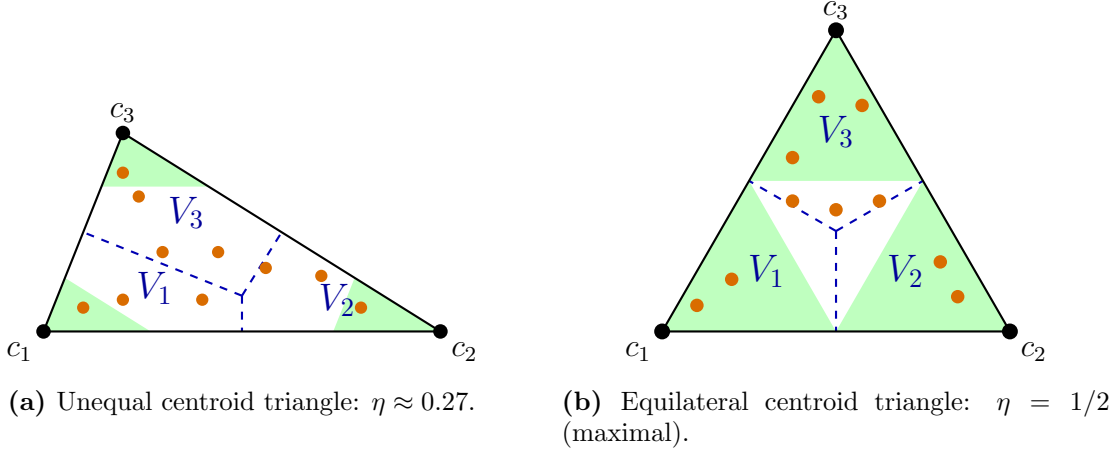
\begin{figure}[tb!]
\centering
\begin{subfigure}[t]{0.48\textwidth}
\centering
\begin{tikzpicture}[scale=1.05]
\coordinate (c1) at (0, 0);
\coordinate (c2) at (5, 0);
\coordinate (c3) at (1, 2.5);
\coordinate (m12) at (2.5, 0);
\coordinate (m13) at (0.5, 1.25);
\coordinate (m23) at (3.0, 1.25);
\coordinate (circ) at (2.5, 0.45);

\fill[green!25] (0, 0)    -- (1.35, 0)    -- (0.27, 0.675) -- cycle;
\fill[green!25] (5, 0)    -- (3.65, 0)    -- (3.92, 0.675) -- cycle;
\fill[green!25] (1, 2.5)  -- (0.73, 1.825) -- (2.08, 1.825) -- cycle;

\draw[blue!70!black, dashed, thick] (circ) -- (m12);
\draw[blue!70!black, dashed, thick] (circ) -- (m13);
\draw[blue!70!black, dashed, thick] (circ) -- (m23);

\draw[black, thick] (c1) -- (c2) -- (c3) -- cycle;

\node[blue!60!black] at (1.4, 0.55) {\large $V_1$};
\node[blue!60!black] at (3.7, 0.4)  {\large $V_2$};
\node[blue!60!black] at (1.8, 1.5)  {\large $V_3$};

\fill[orange!85!black] (0.5, 0.3) circle (2.2pt);
\fill[orange!85!black] (1.0, 0.4) circle (2.2pt);
\fill[orange!85!black] (2.0, 0.4) circle (2.2pt);
\fill[orange!85!black] (4.0, 0.3) circle (2.2pt);
\fill[orange!85!black] (3.5, 0.7) circle (2.2pt);
\fill[orange!85!black] (1.0, 2.0) circle (2.2pt);
\fill[orange!85!black] (1.2, 1.7) circle (2.2pt);
\fill[orange!85!black] (1.5, 1.0) circle (2.2pt);
\fill[orange!85!black] (2.2, 1.0) circle (2.2pt);
\fill[orange!85!black] (2.8, 0.8) circle (2.2pt);

\fill[black] (c1) circle (2.6pt);
\fill[black] (c2) circle (2.6pt);
\fill[black] (c3) circle (2.6pt);
\node[anchor=north east] at (c1) {$c_1$};
\node[anchor=north west] at (c2) {$c_2$};
\node[anchor=south]      at (c3) {$c_3$};
\end{tikzpicture}
\caption{Unequal centroid triangle: $\eta \approx 0.27$.}
\label{fig:voronoi_lemma_irregular}
\end{subfigure}
\hfill
\begin{subfigure}[t]{0.48\textwidth}
\centering
\begin{tikzpicture}[scale=1.15]
\coordinate (c1) at (0, 0);
\coordinate (c2) at (4, 0);
\coordinate (c3) at (2, 3.464);
\coordinate (m12) at (2, 0);
\coordinate (m13) at (1, 1.732);
\coordinate (m23) at (3, 1.732);
\coordinate (circ) at (2, 1.155);

\fill[green!25] (0, 0)      -- (2, 0)      -- (1, 1.732) -- cycle;
\fill[green!25] (4, 0)      -- (2, 0)      -- (3, 1.732) -- cycle;
\fill[green!25] (2, 3.464)  -- (1, 1.732)  -- (3, 1.732) -- cycle;

\draw[blue!70!black, dashed, thick] (circ) -- (m12);
\draw[blue!70!black, dashed, thick] (circ) -- (m13);
\draw[blue!70!black, dashed, thick] (circ) -- (m23);

\draw[black, thick] (c1) -- (c2) -- (c3) -- cycle;

\node[blue!60!black] at (1.2, 0.7) {\large $V_1$};
\node[blue!60!black] at (2.8, 0.7) {\large $V_2$};
\node[blue!60!black] at (2.0, 2.3) {\large $V_3$};

\fill[orange!85!black] (0.4, 0.3) circle (2.2pt);
\fill[orange!85!black] (0.8, 0.6) circle (2.2pt);
\fill[orange!85!black] (3.4, 0.4) circle (2.2pt);
\fill[orange!85!black] (3.2, 0.8) circle (2.2pt);
\fill[orange!85!black] (1.8, 2.7) circle (2.2pt);
\fill[orange!85!black] (2.3, 2.6) circle (2.2pt);
\fill[orange!85!black] (2.0, 1.4) circle (2.2pt);
\fill[orange!85!black] (1.5, 1.5) circle (2.2pt);
\fill[orange!85!black] (2.5, 1.5) circle (2.2pt);
\fill[orange!85!black] (1.5, 2.0) circle (2.2pt);

\fill[black] (c1) circle (2.6pt);
\fill[black] (c2) circle (2.6pt);
\fill[black] (c3) circle (2.6pt);
\node[anchor=north east] at (c1) {$c_1$};
\node[anchor=north west] at (c2) {$c_2$};
\node[anchor=south]      at (c3) {$c_3$};
\end{tikzpicture}
\caption{Equilateral centroid triangle: $\eta = 1/2$ (maximal).}
\label{fig:voronoi_lemma_equilateral}
\end{subfigure}
\caption{Illustration of Lemma~\ref{lem:voronoi_membership} for $K=3$. In each panel, the topic centroids $c_1, c_2, c_3$ (black) span the convex hull $\mathrm{conv}\{c_1,c_2,c_3\}$ (solid triangle). The perpendicular bisectors of the centroid pairs (dashed) partition the triangle into three Voronoi cells $\mathcal{V}_1, \mathcal{V}_2, \mathcal{V}_3$, with $\mathcal{V}_k$ containing the points closer to $c_k$ than to any other centroid. Green shading marks the regions where the lemma guarantees correct assignment, $\theta_{d, k^*(d)} > 1-\eta$. Orange dots illustrate document embeddings $\mu_d = C\,\Theta_{\bullet d}$ for ten documents with varying topic mixtures.}
\label{fig:voronoi_lemma}
\end{figure}

Lemma~\ref{lem:voronoi_membership} is the vertex case: documents close to a vertex will be assigned to the corresponding Voronoi cell. But the logic behind this is not special to the vertices: Below, I establish the more general result that, for any set of ``archetypes'' $\mu^*_1, \ldots, \mu^*_L$, documents whose embeddings are sufficiently close to an archetype $\mu^*_\ell$ will be assigned to the Voronoi cell of $\mu^*_\ell$. Thus, the recovered clusters then reflect \emph{mixture-type} similarity rather than dominant-topic similarity, and the number of clusters $L$ need not equal the number of topics $K$. Figure~\ref{fig:voronoi_lemma_general} illustrates the archetype geometry. 
\begin{figure}[tb]
\centering
\begin{tikzpicture}[scale=1.15]
\coordinate (c1) at (0, 0);
\coordinate (c2) at (4, 0);
\coordinate (c3) at (2, 3.464);
\coordinate (muone) at (1, 0.7);
\coordinate (mutwo) at (2.5, 2.0);

\begin{scope}
\clip (c1) -- (c2) -- (c3) -- cycle;
\fill[green!25] (muone) circle (0.99);
\fill[green!25] (mutwo) circle (0.99);
\end{scope}

\draw[blue!70!black, dashed, thick] (1.168, 2.022) -- (2.92, 0);

\draw[black, thick] (c1) -- (c2) -- (c3) -- cycle;

\node[blue!60!black] at (1.2, 1.6) {\large $V_1$};
\node[blue!60!black] at (2.0, 3.0) {\large $V_2$};

\fill[orange!85!black] (0.5, 0.4) circle (2.2pt);
\fill[orange!85!black] (1.0, 1.2) circle (2.2pt);
\fill[orange!85!black] (1.5, 0.5) circle (2.2pt);
\fill[orange!85!black] (2.2, 0.3) circle (2.2pt);
\fill[orange!85!black] (2.4, 1.9) circle (2.2pt);
\fill[orange!85!black] (3.0, 1.5) circle (2.2pt);
\fill[orange!85!black] (2.0, 2.5) circle (2.2pt);
\fill[orange!85!black] (1.8, 1.5) circle (2.2pt);
\fill[orange!85!black] (3.0, 0.8) circle (2.2pt);
\fill[orange!85!black] (1.2, 1.0) circle (2.2pt);
\fill[orange!85!black] (2.5, 1.4) circle (2.2pt);
\fill[orange!85!black] (1, 0.6) circle (2.2pt);
\fill[orange!85!black] (2.5, 2.3) circle (2.2pt);

\node[anchor=north east] at (c1) {$c_1$};
\node[anchor=north west] at (c2) {$c_2$};
\node[anchor=south]      at (c3) {$c_3$};

\fill[black] (muone) circle (2.6pt);
\fill[black] (mutwo) circle (2.6pt);
\node[anchor=north east, xshift=3pt] at (muone) {$\mu^*_1$};
\node[anchor=south west, xshift=1pt]  at (mutwo) {$\mu^*_2$};
\end{tikzpicture}
\caption{Illustration of the nearest-archetype membership for $K=3$ and two archetypes ($L = 2$). The convex hull $\mathrm{conv}\{c_1,c_2,c_3\}$ (solid black) is partitioned by the perpendicular bisectors of archetype pairs (dashed) into Voronoi cells $V'_\ell$ of the archetype embeddings $\mu^*_\ell$ (black dots in the interior or on the hull boundary). Green shading marks the guaranteed regions $\{\mu_d : \|\mu_d - \mu^*_\ell\|_2 < \Delta'^*/2\}$ — disks of radius $\Delta'^*/2$ around each archetype, clipped to the hull. Orange dots illustrate document embeddings $\mu_d = C\,\Theta_{\bullet d}$. Documents inside a green disk are guaranteed by the triangle-inequality argument to lie in the corresponding Voronoi cell.}
\label{fig:voronoi_lemma_general}
\end{figure}
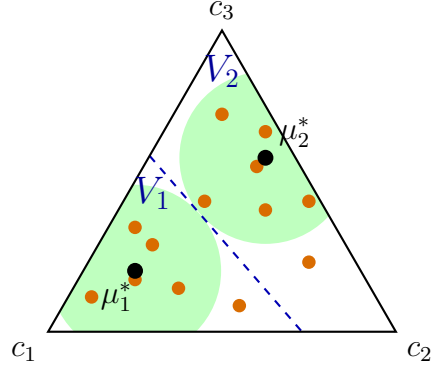
Formally, extending the argument from a single document to the full corpus, I obtain the following proposition.

\begin{proposition}[$k$-means identification under archetype concentration]\label{prop:archetype_identification}
Suppose Assumption~\ref{ass:topic_reg} holds. Let $\pi^*_1, \ldots, \pi^*_L \in \Delta_{K-1}$ be $L \geq 2$ distinct \emph{archetype mixtures}---points in $\Delta_{K-1}$ around which documents concentrate---with minimum separation
\[
\Delta'^* \;:=\; \min_{\ell \neq \ell'} \|C(\pi^*_\ell - \pi^*_{\ell'})\|_2 \;>\; 0,
\]
and let $\ell(d) := \arg\min_\ell \|C(\Theta_{\bullet d} - \pi^*_\ell)\|_2$ denote the nearest-archetype assignment in topic-mixture space. If every document satisfies
\[
\|C(\Theta_{\bullet d} - \pi^*_{\ell(d)})\|_2 \;<\; \Delta'^*/4,
\]
then the partition $\mathcal{P}^* := \{d : \ell(d) = \ell\}_{\ell=1}^L$ is a fixed point of $k$-means: assigning each document to the nearest cell mean of $\mathcal{P}^*$ returns $\mathcal{P}^*$, and recomputing the means returns the same means. Equivalently, $\mathcal{P}^*$ is simultaneously the Voronoi partition of the archetype embeddings $\{\mu^*_\ell\}$ and of its own cell means.
\end{proposition}

\begin{proof}
 Let $\mathcal{G}_\ell := \{d : \ell(d) = \ell\}$, $\mu^*_\ell := C\,\pi^*_\ell$ and $\mu_d = C\,\Theta_{\bullet d}$. Within-cluster pairwise distances then satisfy, for $d, d' \in \mathcal{G}_\ell$,
\begin{align*}
\|\mu_d - \mu_{d'}\|_2 \;\leq\; \|\mu_d - \mu^*_\ell\|_2 + \|\mu_{d'} - \mu^*_\ell\|_2 \;<\; \tfrac{\Delta'^*}{4} + \tfrac{\Delta'^*}{4} \;=\; \tfrac{\Delta'^*}{2},
\end{align*}
while for $d \in \mathcal{G}_\ell, d' \in \mathcal{G}_{\ell'}$ with $\ell \neq \ell'$,
\begin{align*}
\|\mu_d - \mu_{d'}\|_2 \;\geq\; \|\mu^*_\ell - \mu^*_{\ell'}\|_2 - \|\mu_d - \mu^*_\ell\|_2 - \|\mu_{d'} - \mu^*_{\ell'}\|_2 \;\geq\; \Delta'^* - 2\,\tfrac{\Delta'^*}{4} \;=\; \tfrac{\Delta'^*}{2}.
\end{align*}
so within-cluster pairwise distances are strictly smaller than cross-cluster ones.

For the fixed-point claim, let $m_\ell := \bar\mu_{\mathcal{G}_\ell}$ denote the cell means. Each $m_\ell$ is a convex combination of points within $\Delta'^*/4$ of $\mu^*_\ell$, so $\|m_\ell - \mu^*_\ell\|_2 < \Delta'^*/4$. For $d \in \mathcal{G}_\ell$,
\begin{align*}
\|\mu_d - m_\ell\|_2 \;\leq\; \|\mu_d - \mu^*_\ell\|_2 + \|\mu^*_\ell - m_\ell\|_2 \;<\; \tfrac{\Delta'^*}{2},
\end{align*}
while for $\ell' \neq \ell$,
\begin{align*}
\|\mu_d - m_{\ell'}\|_2 \;\geq\; \|\mu^*_\ell - \mu^*_{\ell'}\|_2 - \|\mu_d - \mu^*_\ell\|_2 - \|\mu^*_{\ell'} - m_{\ell'}\|_2 \;>\; \Delta'^* - \tfrac{\Delta'^*}{2} \;=\; \tfrac{\Delta'^*}{2}.
\end{align*}
Every document is therefore strictly closer to its own cell mean than to any other, so the assignment step reproduces $\mathcal{P}^*$; the update step then returns the same means. %
\end{proof}

At the simplex vertices, Proposition~\ref{prop:archetype_identification} specializes to the dominant-topic case, translating the embedding-distance condition back into a sufficient threshold on $\theta$. 
It is worth noting that Proposition~\ref{prop:archetype_identification} is a stability statement, not a recovery statement.\footnote{Note that the archetypes in Proposition~\ref{prop:archetype_identification} need not be fixed in advance. A natural choice is the $k$-means centers at convergence: by construction they lie in $\mathrm{conv}\{c_1,\ldots,c_K\}$ and equal $C\pi^*_\ell$ for some $\pi^*_\ell \in \Delta_{K-1}$. Proposition~\ref{prop:archetype_identification} then applies to these centers: if they are separated and every document lies within $\Delta'^*/4$ of its assigned center, the converged partition is a $k$-means fixed point.} 
Global optimality requires even stronger conditions.
However, even the weaker separation condition in Proposition~\ref{prop:archetype_identification} is not satisfied in our application. 
I therefore claim no recovery of a latent partition in our application. Instead, I rely on Corollary~\ref{cor:doc_embedding_metric}, which connects embedding distance to the difference of topic mixtures. A cluster is thus a set of documents with similar mixtures. %
This is what the interpretation of the clusters in Section~\ref{sec:application} rests on, and it requires no concentration assumption.

\subsection{Conditioning: Embeddings as Controls}\label{sec:conditioning}

The second use of document embeddings I consider is as a \emph{control}. A researcher who wants to control for a high-dimensional text $X_d$---a firm disclosure, a court filing, a job posting---cannot condition on the raw text and instead conditions on its embedding, implicitly assuming that adjustment for the embedding suffices. Proposition~\ref{prop:doc_embedding_general_freq}, and in particular the identity $\mu_d = C\,\Theta_{\bullet d}$, immediately makes the content of that assumption precise.

Formally, let $Z_d \in \{0,1\}$ be a treatment, $Y_d(0), Y_d(1)$ the potential outcomes for unit $d$, and $Y_d = Y_d(Z_d)$ the observed outcome.

\begin{assumption}[Topic unconfoundedness]\label{ass:topic_unconf}
$\{Y_d(0), Y_d(1)\} \perp Z_d \mid \Theta_{\bullet d}$, and $0 < P(Z_d = 1 \mid \Theta_{\bullet d}) < 1$.
\end{assumption}

Assumption~\ref{ass:topic_unconf} states that the document's topic mixture captures all confounding between treatment and potential outcomes: given $\Theta_{\bullet d}$, the realized words carry no further information about $(Z_d, Y_d(0), Y_d(1))$. This is the explicit content of the informal claim that the text is a sufficient control.

\begin{corollary}[Embedding adjustment]\label{cor:embedding_adjustment}
Suppose Assumption \ref{ass:topic_reg} holds, the topic model is identified, and $C$ is injective on $\Delta_{K-1}$.%
Then, controlling for document embedding is valid if and only if controlling for the topic mixture is. In particular, under Assumption~\ref{ass:topic_unconf}, the average treatment effect is identified by
\begin{align*}
\tau \;=\; \E\big[\, \E[Y_d \mid Z_d = 1, \mu_d] - \E[Y_d \mid Z_d = 0, \mu_d] \,\big].
\end{align*}
\end{corollary}

\begin{proof}
Since $\Pi = B\Theta$ and the topic model is identified, $\mu_d = C\,\Theta_{\bullet d}$ (Proposition~\ref{prop:doc_embedding_general_freq}). Since $C$ is injective on $\Delta_{K-1}$, $\mu_d$ and $\Theta_{\bullet d}$ are one-to-one functions of each other: conditioning on one is the same as conditioning on the other. Adjustment on the embedding is therefore valid exactly when adjustment on the topic mixture is. In particular, under Assumption~\ref{ass:topic_unconf}, $\{Y_d(0), Y_d(1)\} \perp Z_d \mid \mu_d$ and $0 < P(Z_d = 1 \mid \mu_d) < 1$, so $\tau$ is identified by adjustment on $\mu_d$ \citep{rosenbaum1983central}.
\end{proof}

Corollary~\ref{cor:embedding_adjustment} is an identification statement at the population level; it is exact only under three conditions. First, it uses the expected embedding $\mu_d$. With finite document length, both $\mu_d$ and $\Theta_{\bullet d}$ will be measured with error (e.g., $\bar h_d = \mu_d + O_p(N_d^{-1/2})$; see the discussion in Section~\ref{sec:doc_embeddings}). Also see \cite{battaglia2024inference} for a more general discussion of measurement error in AI/ML-generated variables. 
Second, it requires embedding dimension $r \ge K-1$: with $r < K-1$ the map $\mu_d = C,\Theta_{\bullet d}$ cannot be injective and conditions on only a projection of $\Theta_{\bullet d}$, leaving residual confounding regardless of $N_d$. However, for the high-dimensional embeddings used in practice I believe this bound is slack.
Third, $C$ must be injective on $\Delta_{K-1}$. For an embedding satisfying $H H^\top = R - \mathbf{1}_V\mathbf{1}_V^\top$ this holds by construction (Corollary~\ref{cor:doc_embedding_metric}). For Word2Vec, GloVe, or pretrained embeddings, it is a maintained assumption; see the discussion in Section~\ref{sec:other_embeddings}.

The point is not that embedding-based adjustment is valid, but that its validity reduces to a transparent condition on the topic mixture. This connects to a growing literature on causal inference with text \citep{roberts2020adjusting, egami2022how}; relative to that work, I derive the sufficiency condition from the generative model rather than assert it.

Although Corollary~\ref{cor:embedding_adjustment} is stated for the topic model, its two structural requirements are general. Suppose the confounding is driven by a latent summary of the text of dimension $s$ (here $s = K-1$, the topic mixture). Embedding-based control then requires, first, an embedding of dimension $r \ge s$---otherwise it encodes only a projection of that summary and leaves residual confounding; and second, that the embedding actually \emph{recover} the summary, so that its low-dimensional geometry coincides with the one the confounding lives in. %
Dimension alone thus does not suffice---an embedding can be low-dimensional in a geometry that has little to do with the text's, and adjusting for it then controls for the wrong object.

\section{Numerical Illustration}\label{sec:simulations}

I next simulate data from the model described in Section~\ref{sec:generative} and compare two embedding constructions: (i) the closed-form SVD-$\beta$ embedding of Proposition~\ref{thm:factorization}, and (ii) skip-gram with negative sampling (SGNS) \citep{mikolov2013distributed} as implemented in gensim \citep{rehurek2010gensim} (cf. Table \ref{tab:targets}). Theorem~\ref{thm:word_embedding_guarantee} predicts that SVD-$\beta$ inherits the topic-loading geometry, and Proposition~\ref{prop:oe_general} carries the prediction over to the SGNS target.%
\footnote{Appendix \ref{app-sec:CBOW} repeats the exercise with continuous bag-of-words (CBOW). The results are qualitatively similar to SGNS.}

Both embeddings use dimension $r=K$. By Proposition~\ref{thm:factorization}, $r = K-1$ is the minimum sufficient embedding dimension for SVD-$\beta$. The extra dimension should have population eigenvalue zero and be genuine slack, and I confirm this in the figures below. The SGNS target is the approximately rank-$(K-1)$ $\log R$ shifted by the constant $-\log\nu$. We quantify that approximation in the last paragraph of this section. Because the shift is constant, it raises the rank of the target without spreading words apart, so it does not appear in the centered projections plotted below.

To illustrate our theoretical results, consider two simulation designs on a common $V=16$ vocabulary of fruits and vegetables: a two-topic case (Design 1) and a three-topic extension (Design 2). Documents are generated with topic mixtures drawn from a symmetric Dirichlet $(\alpha=1)$ prior on the simplex, with $D=363$ documents and $N_d=487$ words per document, matching the dimensions of the application in Section~\ref{sec:application}, and the full document as context.\footnote{Appendix~\ref{app-sec:simulations} repeats this exercise with two further specifications: i) a larger corpus, $D=1{,}000$ and $N_d=10{,}000$ (Appendix~\ref{app-sec:large_corpus}); and ii) a more concentrated topic-mixture distribution ($\alpha=0.01$) that puts more mass near the simplex vertices (Appendix~\ref{app-app:sim_sparse}).}%

\textbf{Design 1: Two Topics ($K=2$).} The topic-word matrix is
\begin{align}
B = \left(\begin{array}{cc}
\mathbf{0.08} & 0.00 \\
\mathbf{0.08} & 0.01 \\
\mathbf{0.08} & 0.00 \\
\mathbf{0.12} & 0.00 \\
\mathbf{0.06} & 0.04 \\
\mathbf{0.07} & 0.00 \\
\mathbf{0.10} & 0.02 \\
\mathbf{0.12} & 0.00 \\
\mathbf{0.10} & 0.01 \\
\mathbf{0.13} & 0.00 \\
0.01 & \mathbf{0.10} \\
0.02 & \mathbf{0.18} \\
0.00 & \mathbf{0.30} \\
0.01 & \mathbf{0.05} \\
0.01 & \mathbf{0.12} \\
0.01 & \mathbf{0.17}
\end{array}\right) \quad \text{with rows corresponding to:}
\begin{array}{l}
\textcolor{red}{\text{Apple}} \\
\textcolor{red}{\text{Banana}} \\
\textcolor{red}{\text{Pear}} \\
\textcolor{red}{\text{Mango}} \\
\textcolor{red}{\text{Peach}} \\
\textcolor{red}{\text{Cherry}} \\
\textcolor{red}{\text{Orange}} \\
\textcolor{red}{\text{Lemon}} \\
\textcolor{red}{\text{Mandarin}} \\
\textcolor{red}{\text{Grapefruit}} \\
\textcolor{green!60!black}{\text{Cucumber}} \\
\textcolor{green!60!black}{\text{Tomato}} \\
\textcolor{green!60!black}{\text{Kale}} \\
\textcolor{green!60!black}{\text{Pepper}} \\
\textcolor{green!60!black}{\text{Carrot}} \\
\textcolor{green!60!black}{\text{Onion}}
\end{array}.
\end{align}
Apple, Pear, Mango, Cherry, Lemon, and Grapefruit are anchor words for topic 1 (fruits) with varying prevalence, and Kale is an anchor word for topic 2 (vegetables). The remaining words load on both topics with varying intensity.

Figure~\ref{fig:embeddings_large} shows the resulting word embeddings.
\begin{figure}[tb]
\centering
\begin{subfigure}[b]{0.48\textwidth}
\centering
\includegraphics[width=\textwidth]{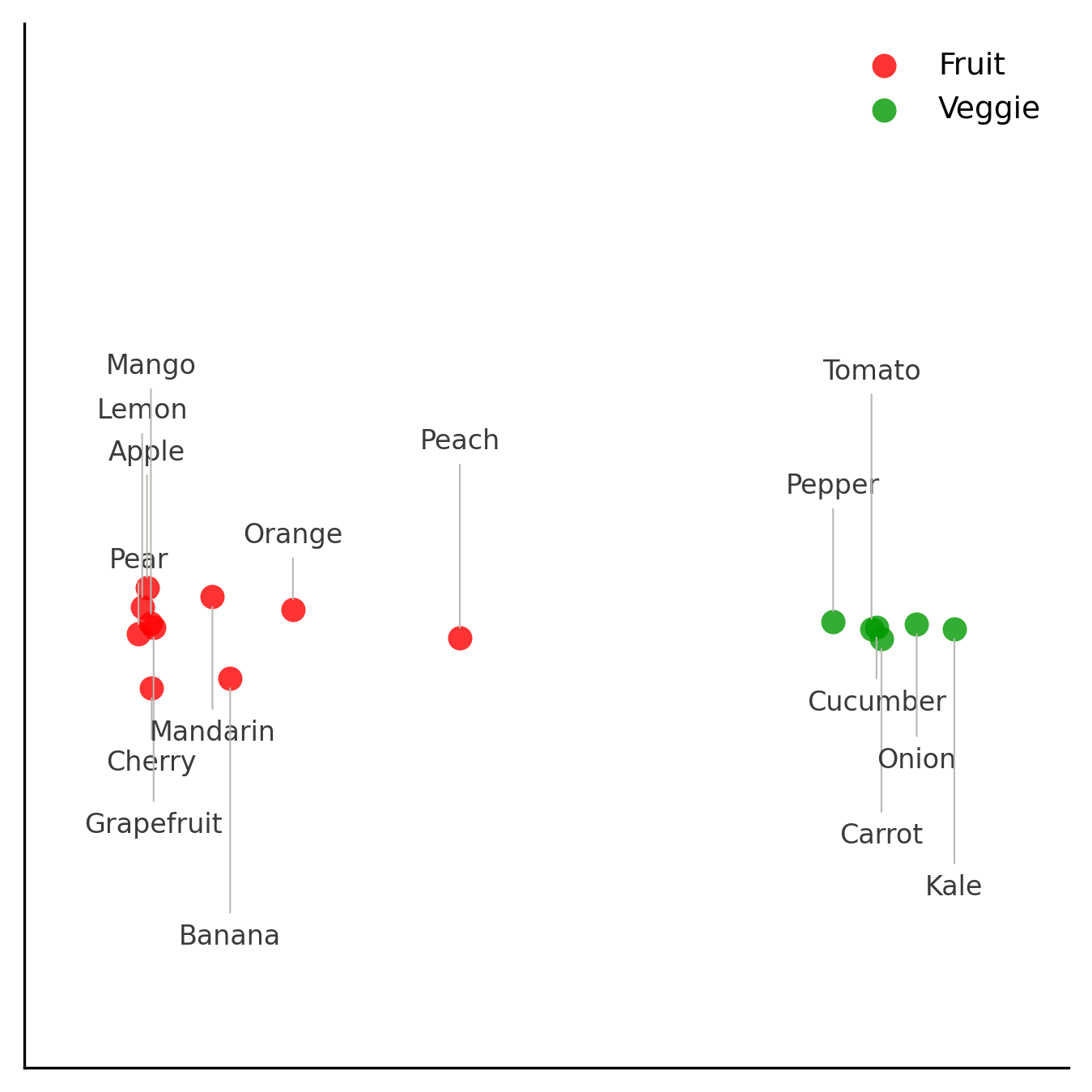}
\caption{SVD-$\beta$}
\label{fig:2d-svd_embeddings}
\end{subfigure}
\hfill
\begin{subfigure}[b]{0.48\textwidth}
\centering
\includegraphics[width=\textwidth]{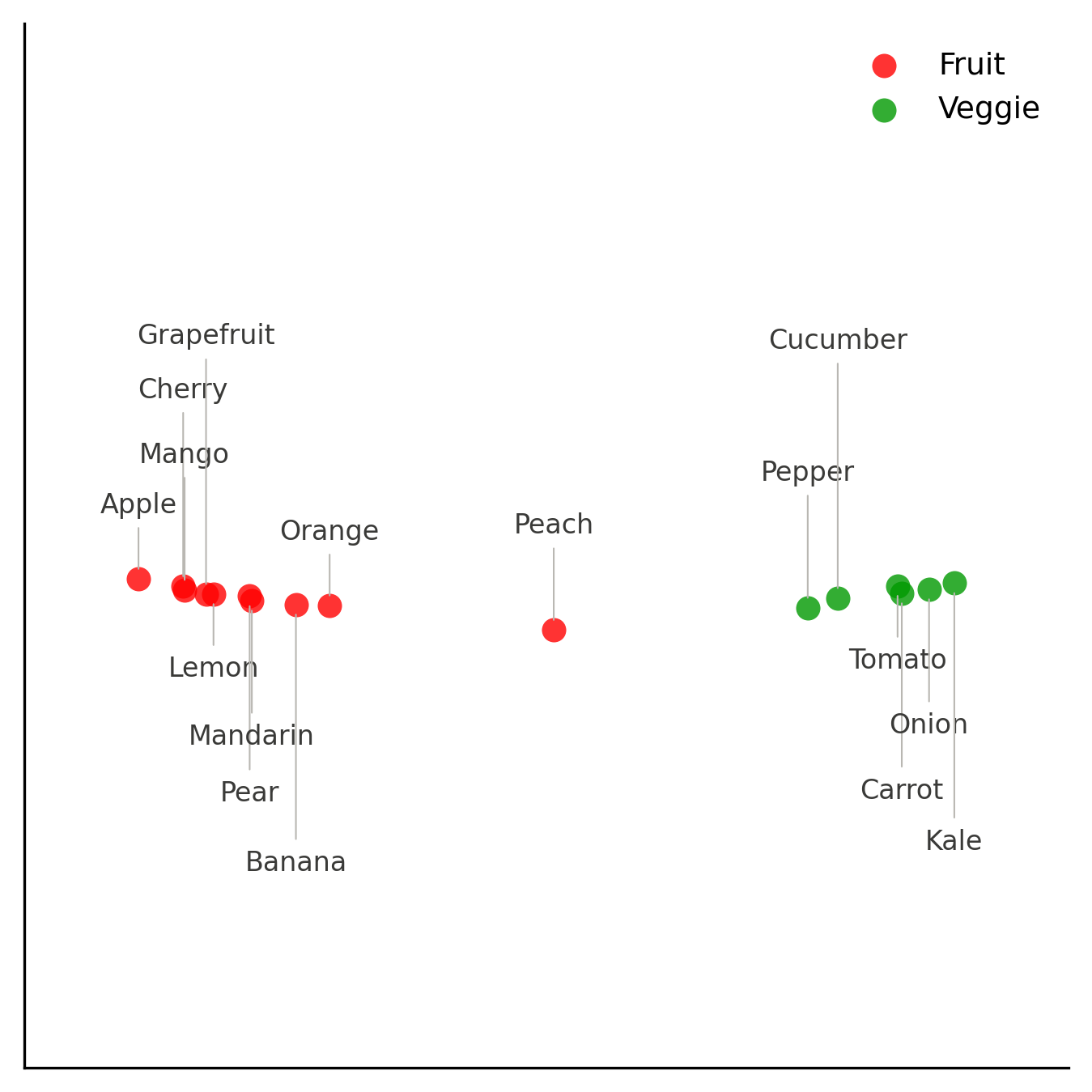}
\caption{SGNS}
\end{subfigure}
\caption{Word embeddings for Design 1 ($K=2$): SGNS vs.\ the closed-form SVD-$\beta$ construction of Proposition~\ref{thm:factorization}, both computed on the same corpus ($D=363$, $N_d=487$). Colors mark the dominant topic from the true $B$ matrix.}
\label{fig:embeddings_large}
\end{figure}
Both methods recover the same qualitative structure: words separate cleanly along their dominant topic, with fruits (red) and vegetables (green) forming distinct clusters. The SGNS and SVD-$\beta$ geometries are nearly identical up to rotation and scaling. Both clouds are one-dimensional, which is what the theory implies: SVD-$\beta$ has population rank $K-1=1$, and SGNS's second direction is the common offset, which the centered projection removes. Proposition~\ref{prop:oe_general} predicts only that the two induce equivalent metrics on the words; here the distortion it permits is evidently small. Within each cluster, the relative positions of words reflect their loadings: anchor words (e.g., Mango, Grapefruit, Kale) lie at the extremes, while mixed words (e.g., Peach, Orange) sit closer to the boundary between clusters, with their positions reflecting their relative loadings. This pattern holds exactly for SVD-$\beta$, and approximately for SGNS. On the other hand, the SVD-$\beta$ embedding appears somewhat noisier (further from rank 1), perhaps reflecting the difficulty in directly matrix-factorizing the noisy sample co-occurrence structure.

\textbf{Design 2: Three Topics ($K=3$).} I next split the "fruit" category into a "citrus" topic (orange) and a "non-citrus" topic (red), while keeping the same vegetable topic (green). The topic-word matrix is
\begin{align}
B = \left(\begin{array}{ccc}
\mathbf{0.10} & 0.05 & 0.00 \\
\mathbf{0.15} & 0.02 & 0.01 \\
\mathbf{0.13} & 0.04 & 0.00 \\
\mathbf{0.19} & 0.05 & 0.00 \\
\mathbf{0.10} & 0.03 & 0.04 \\
\mathbf{0.10} & 0.04 & 0.00 \\
0.05 & \mathbf{0.15} & 0.02 \\
0.05 & \mathbf{0.18} & 0.00 \\
0.03 & \mathbf{0.16} & 0.01 \\
0.04 & \mathbf{0.22} & 0.00 \\
0.02 & 0.00 & \mathbf{0.10} \\
0.02 & 0.02 & \mathbf{0.18} \\
0.00 & 0.00 & \mathbf{0.30} \\
0.00 & 0.02 & \mathbf{0.05} \\
0.01 & 0.01 & \mathbf{0.12} \\
0.01 & 0.01 & \mathbf{0.17}
\end{array}\right) \quad \text{with rows corresponding to:}
\begin{array}{l}
\textcolor{red}{\text{Apple}} \\
\textcolor{red}{\text{Banana}} \\
\textcolor{red}{\text{Pear}} \\
\textcolor{red}{\text{Mango}} \\
\textcolor{red}{\text{Peach}} \\
\textcolor{red}{\text{Cherry}} \\
\textcolor{orange}{\text{Orange}} \\
\textcolor{orange}{\text{Lemon}} \\
\textcolor{orange}{\text{Mandarin}} \\
\textcolor{orange}{\text{Grapefruit}} \\
\textcolor{green!60!black}{\text{Cucumber}} \\
\textcolor{green!60!black}{\text{Tomato}} \\
\textcolor{green!60!black}{\text{Kale}} \\
\textcolor{green!60!black}{\text{Pepper}} \\
\textcolor{green!60!black}{\text{Carrot}} \\
\textcolor{green!60!black}{\text{Onion}}
\end{array}
\end{align}
Kale is an anchor word for topic 3 (vegetables), but topics 1 (non-citrus fruits) and 2 (citrus) have no anchor words, though they differ in their relative loadings on the same set of words.

Figure~\ref{fig:embeddings_noanchor_k3} shows the resulting word embeddings.
With $r=K=3$, the embeddings live in $\mathbb{R}^3$, and I project to two dimensions by simply retaining the first two coordinates. 
For SVD-$\beta$, $\beta_{\bullet 1}$ and $\beta_{\bullet 2}$ span the two largest-variance directions and the third column corresponds to a zero eigenvalue in population.
For SGNS, the coordinate basis is chosen by the gradient-based optimizer and is arbitrary.  Dropping the third coordinate projects out one mixture of the three coordinates. %
\begin{figure}[tb]
\centering
\begin{subfigure}[b]{0.48\textwidth}
\centering
\includegraphics[width=\textwidth]{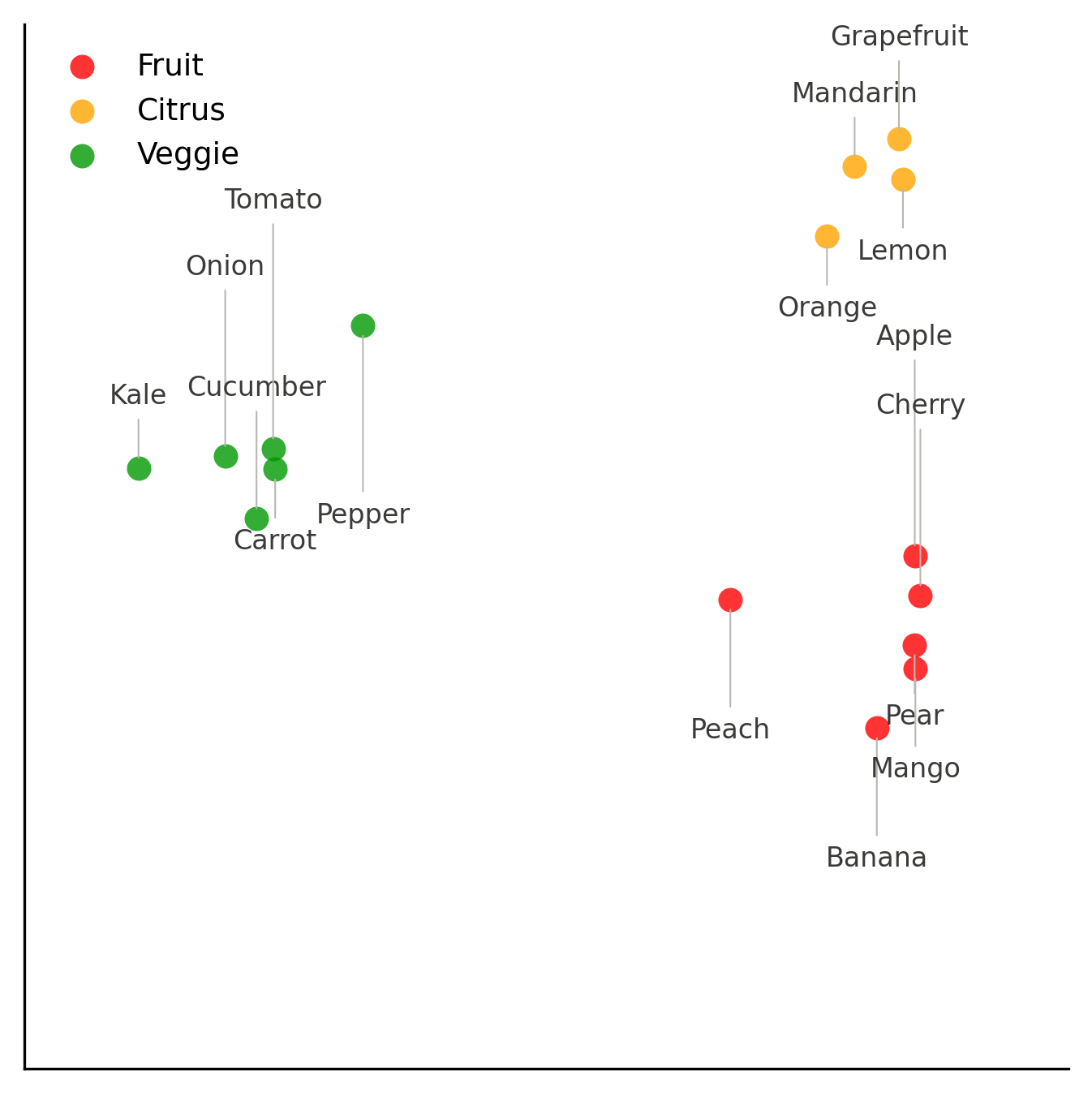}
\caption{SVD-$\beta$}
\end{subfigure}
\hfill
\begin{subfigure}[b]{0.48\textwidth}
\centering
\includegraphics[width=\textwidth]{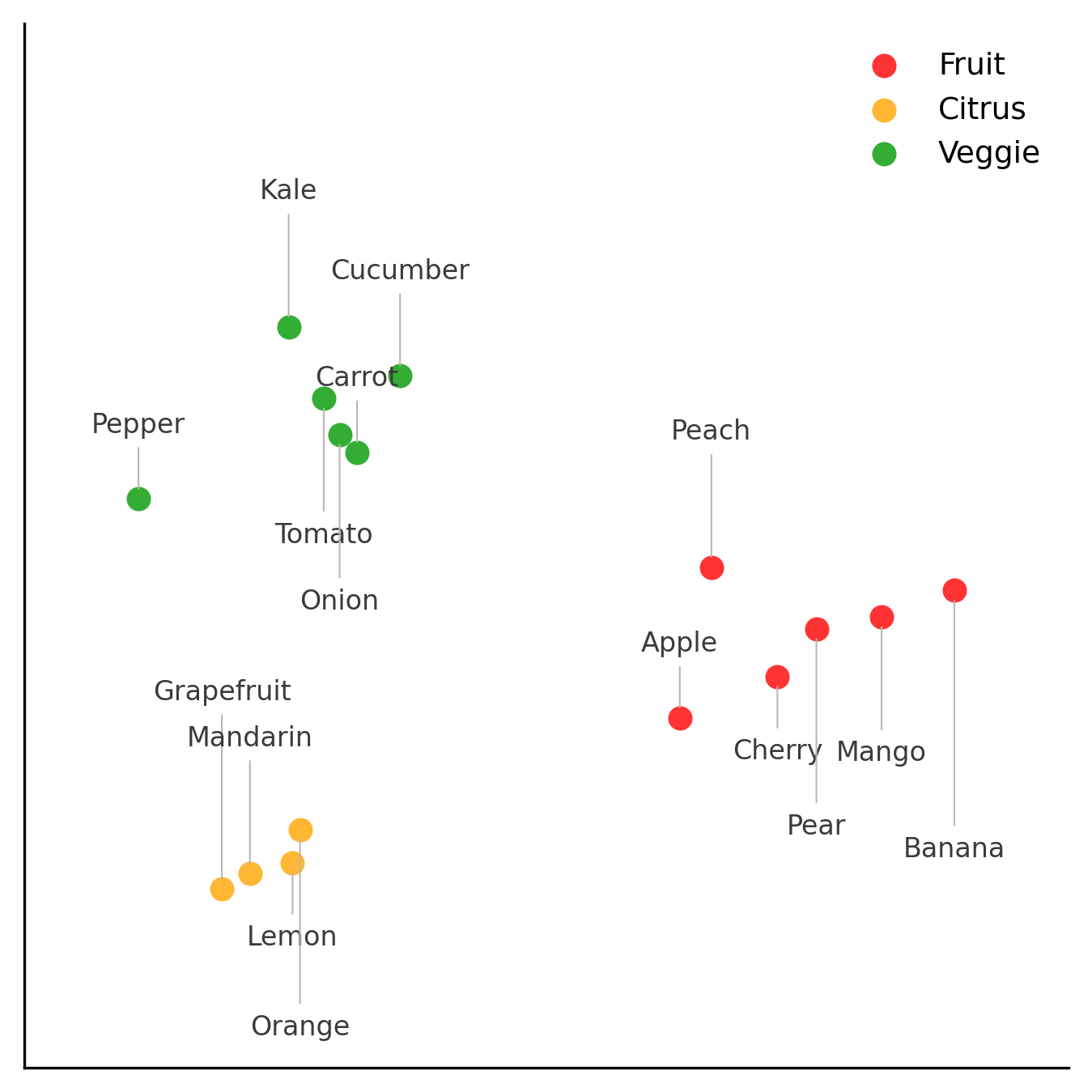}
\caption{SGNS}
\end{subfigure}
\caption{Word embeddings for Design 2 ($K=3$), projected into two dimensions. Both embeddings recover topic structure even though topic~3 lacks an anchor word.}
\label{fig:embeddings_noanchor_k3}
\end{figure}

Both methods produce embeddings in which the three topics are clearly separated, with words organized into three distinguishable groups corresponding to non-citrus fruits (red), citrus fruits (orange), and vegetables (green). The anchor word for vegetables (Kale) sits on the edge of the convex hull, consistent with the theory: anchor words receive embeddings proportional to a single topic centroid, while non-anchor words are pulled toward mixtures of centroids in proportion to their topic loadings. Notably, the absence of anchor words for topics~1 and~2 does not prevent the embeddings from separating these two groups---the differing relative loadings of shared words across the two topics are sufficient to recover the distinction. The SGNS and SVD-$\beta$ embeddings again yield qualitatively similar geometries, providing empirical support for the theoretical connections established in Section~\ref{sec:other_embeddings}.

\paragraph{Document Embeddings.} Next, I compute document embeddings as the average of word embeddings within each document (cf. Equation~\eqref{eq:doc_embedding}). By Proposition~\ref{prop:doc_embedding_general_freq}, documents with similar topic mixtures $\Theta_{\bullet d}$ should have similar embeddings, and document embeddings should lie in the $(K-1)$-dimensional simplex spanned by the topic centroids.
Figure~\ref{fig:doc_emb_noanchor_k3_alt} shows the document embeddings for Design 2 ($K=3$), colored by dominant topic (projected into 2-d following the same step as in Figure \ref{fig:embeddings_noanchor_k3}).\footnote{Alternatively, Appendix~\ref{app-app:sim_umap} presents a UMAP-projected version of Figure~\ref{fig:embeddings_noanchor_k3}. UMAP's nonlinear transformation distorts the geomerty and ``warps'' the simplex, but the results are qualitatively similar. I use UMAP projections in our empirical application in Section~\ref{sec:application}, where the dimensionality is much larger.}
\begin{figure}[htbp]
\centering
\begin{subfigure}[b]{0.48\textwidth}
\centering
\includegraphics[width=\textwidth]{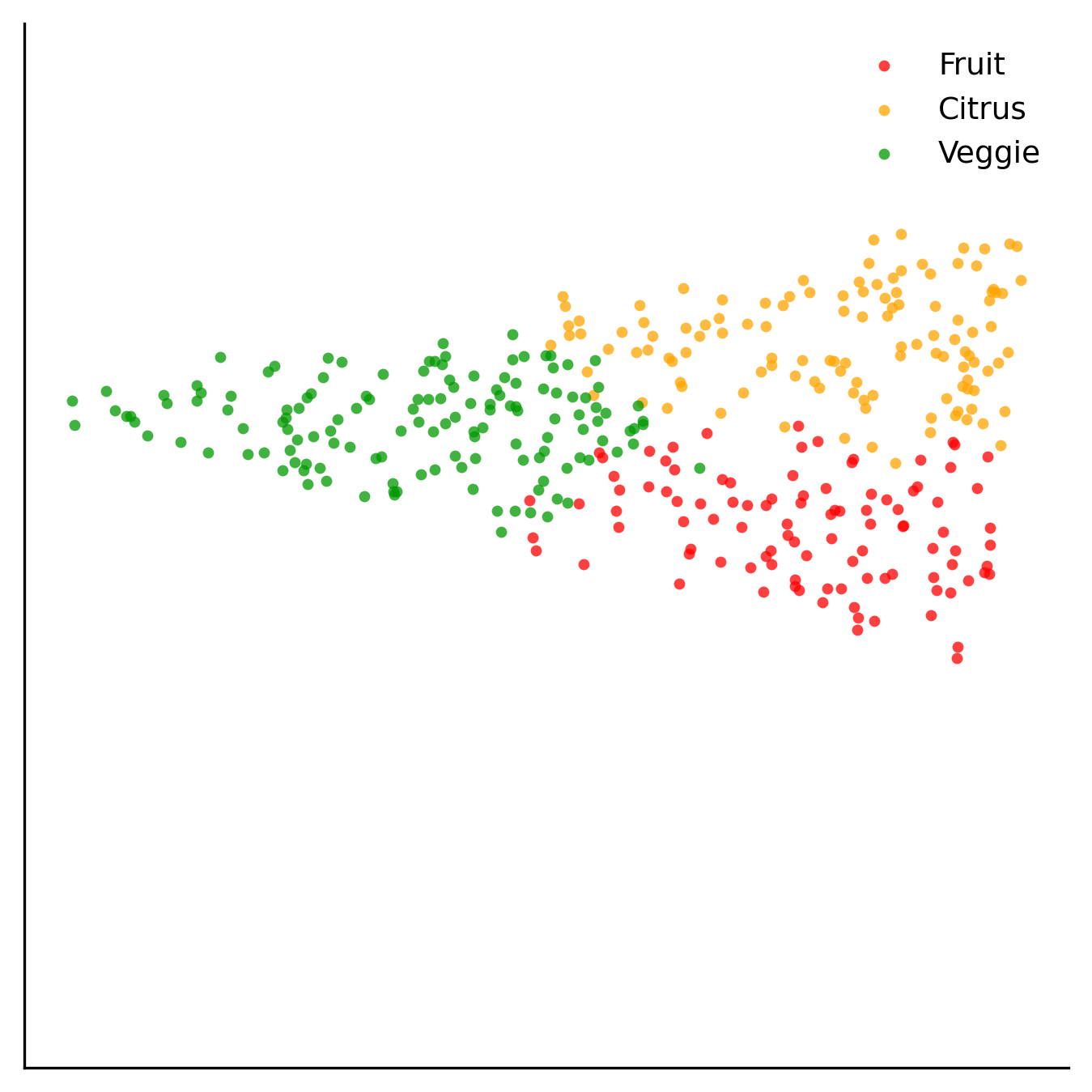}
\caption{SVD-$\beta$}
\end{subfigure}
\hfill
\begin{subfigure}[b]{0.48\textwidth}
\centering
\includegraphics[width=\textwidth]{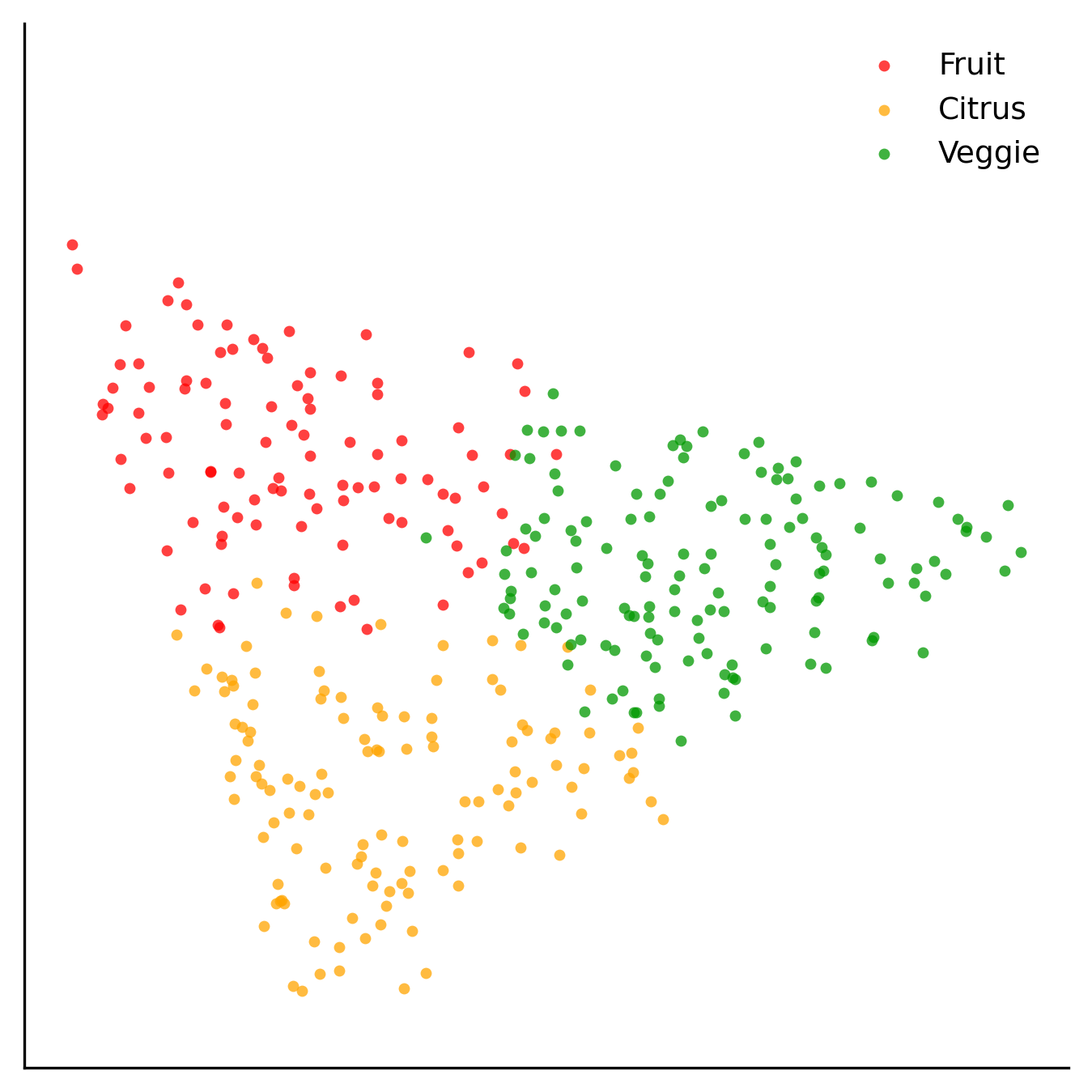}
\caption{SGNS}
\end{subfigure}
\caption{Document embeddings for Design 2 ($K=3$), projected in two dimensions. Color indicates the dominant topic.}
\label{fig:doc_emb_noanchor_k3_alt}
\end{figure}
With Dirichlet concentration parameter $\alpha=1$, the topic mixtures are well spread across the simplex, and the document embeddings fill out the simplex spanned by the three topic centroids. The two embedding methods yield similar document-embedding geometries. Appendix~\ref{app-app:sim_archetype} shows additional document-embedding figures for a version that replaces the Dirichlet prior on $\Theta_{\bullet d}$ with an archetype mixture in the interior of the topic simplex, with documents concentrated around the two archetypes.

I note that all of the figures above are based on a single corpus. However, I found the results to be representative across multiple realizations.

\textbf{The rank of the SGNS target.} Remark~\ref{rem:target_vs_embedding} reduces the gap between the SGNS target and the trained embedding to a property of the spectrum of $T = \log R - \log\nu\,\mathbf{1}_V\mathbf{1}_V^\top$. The shift is a known constant and cancels in the row differences, so the question is how much of $\log R$ the remaining $K-1$ directions capture. I measure that by \begin{align}\label{eq:nuclear_share}
\pi_{K-1} \;=\; \frac{\sum_{i \le K-1} |\lambda_i(\log R)|}{\sum_i |\lambda_i(\log R)|},
\end{align}
the share of spectral mass in the leading $K-1$ eigenvalues.
Under the two designs above the spectrum is available exactly --- $R = \tilde B G \tilde B^\top$ from the design $B$ and the closed-form Dirichlet moments --- and gives $\pi_{K-1} = \piRankDesignOne$ for Design 1 and $\piRankDesignTwo$ for Design 2. %
In these designs we can therefore think of Proposition~\ref{prop:oe_general} approximately describing the trained embedding and not only its target, in line with the empirical results above.
Table~\ref{tab:target_rank_alpha1} reports $\pi_{K-1}$ on a grid of vocabulary sizes and topic counts under a generic data-generating process: each column of $B$ and each topic mixture $\Theta_{\bullet d}$ is an independent draw from a symmetric Dirichlet$(\alpha=1)$, uniform on its simplex. Across the grid, a rank-$K$ approximation remains close to the SGNS target.\footnote{It is worth noting that the quality of the approximation degrades for smaller values of $\alpha$ (cf. Appendix \ref{app-app:target_rank_sparse}.}

\begin{table}[tb!]
\centering
\footnotesize
\setlength{\tabcolsep}{6pt}
\begin{tabular}{lrrrrrr}
\toprule
$V$ & $K=2$ & $K=3$ & $K=5$ & $K=10$ & $K=25$ & $K=50$ \\
\midrule
100 & $0.909$ & $0.926$ & $0.946$ & $0.968$ & $0.987$ & $0.995$ \\
250 & $0.905$ & $0.924$ & $0.943$ & $0.966$ & $0.985$ & $0.993$ \\
500 & $0.905$ & $0.923$ & $0.943$ & $0.964$ & $0.984$ & $0.992$ \\
1,000 & $0.905$ & $0.924$ & $0.943$ & $0.964$ & $0.983$ & $0.991$ \\
2,500 & $0.905$ & $0.923$ & $0.943$ & $0.963$ & $0.983$ & $0.991$ \\
\bottomrule
\end{tabular}
\caption[Rank of the SGNS target across $V$ and $K$]{Share of the spectral
mass of $\log R$ carried by its leading $K-1$ eigenvalues,
$\pi_{K-1} = \sum_{i \le K-1}|\lambda_i| / \sum_i |\lambda_i|$, at
$\alpha = 1$. Each column of $B$ and each topic mixture
$\Theta_{\bullet d}$ is drawn from a symmetric Dirichlet$(\alpha)$; $R$ is then
formed in population from $B$ and $G$, so no corpus is sampled. Means over 8
draws; the largest standard deviation in any cell is $0.004$. The
negative-sampling shift contributes one further eigenvalue, of size $V\log\nu$,
and is excluded.}
\label{tab:target_rank_alpha1}
\end{table}

\section{Application: Clustering Metropolitan Economies}\label{sec:application}

I now apply the embedding-then-cluster pipeline to 363 U.S.\ Core-Based Statistical Areas (CBSAs). The empirical setup is straightforward: generate a textual economic description of each CBSA, embed the descriptions, and apply $k$-means.

The theory in Section~\ref{sec:theory} provides a principled reading of what this pipeline does. Corollary~\ref{cor:doc_embedding_metric} gives the pairwise document-embedding distance in closed form:
\begin{align*}
\|\mu_d - \mu_{d'}\|_2^2 \;=\; (\Pi_{\bullet d} - \Pi_{\bullet d'})^\top \,(R - \mathbf{1}_V\mathbf{1}_V^\top)\, (\Pi_{\bullet d} - \Pi_{\bullet d'}).
\end{align*}
Running $k$-means on these embeddings groups CBSAs whose descriptions imply similar mixtures. The clusters therefore reflect \emph{narrative similarity}, and each centroid corresponds to a well-defined mixture.

Our main empirical analysis considers five embedding models: (i) the closed-form SVD-$\beta$ embedding of Proposition \ref{thm:factorization}; (ii) corpus-trained skip-gram with negative sampling (SGNS; Proposition \ref{prop:oe_general}); (iii) corpus-trained continuous bag-of-words (CBOW)\footnote{Both corpus-trained arms and SVD-$\beta$ use the \emph{full document} as context. Note that the model from Section \ref{sec:generative} suggests this choice: since the words of document $d$ are drawn from the same mixture $\Pi_{\bullet d}$, positions within a document are exchangeable, so $R$ does not depend on the window at all in population and every within-document pair is another draw on the same object. Widening the window therefore reduces estimation variance without moving the target, and taking all pairs is the efficient choice. I show in Appendix~\ref{app-app:app_context_window} that results are robust to shorter context windows, although SVD-$\beta$ tends to deteriorate for very short windows $J$.};  (iv) pretrained Word2Vec averaging, using the Google News vectors; and (v) a transformer-based variant using OpenAI's \texttt{text-embedding-3-large}.\footnote{Alternatively, one could ask an LLM to form the clusters and assign labels directly, skipping the embedding step. This did not work at this scale: prompted with all 363 descriptions, the model returned partitions covering only part of the sample and including locations absent from the input list.} Detailed method descriptions are collected in Appendix~\ref{app-app:app_methods}.

\begin{remark}[Theoretical coverage of the five methods]\label{rem:methods}
The five differ in how much of Section~\ref{sec:theory} applies to them. Method (i) is the theory's own construction: it computes the rank-$(K-1)$ factorization of Proposition~\ref{thm:factorization} directly, and averaging within a document returns exactly the $\mu_d$ of Proposition~\ref{prop:doc_embedding_general_freq}. Method (ii) is the corpus-trained arm the theory covers: the SGNS target is $\log R$ up to a constant, so by Proposition~\ref{prop:oe_general} it recovers the topic-loading geometry up to a distortion bounded by $\kappa_R$, the exponentiated spread of $\log R$ (Section~\ref{sec:other_embeddings}). Method (iii), CBOW, is not covered: it has no closed-form $V \times V$ target once the context holds more than one word (Section~\ref{sec:other_embeddings}), and the context here is the full document. Method (iv) runs the same architecture on a different corpus, so the geometry it carries is that of the corpus it was trained on---the object of interest only if that corpus and ours share their topic structure. Method (v) falls outside the framework altogether: it is a contextual transformer rather than an average of word vectors. I include it as a benchmark for a modern state-of-the-art approach to embeddings.
\end{remark}

For each Core-Based Statistical Area (CBSA), a 500-word economic description is generated by Claude Opus 4.8. I provide the exact prompt in Appendix~\ref{app-app:app_prompts}. I then apply $k$-means with $L=5$ clusters to group the 363 CBSAs. 
Table~\ref{tab:cluster_comparison_363} summarizes the five approaches.
\begin{table}[htb!]
\centering
\small
\begin{tabular}{lcccc}
\hline
\textbf{Method} & \textbf{L} & \textbf{Cluster sizes} & \textbf{Max share} & \textbf{Within-cluster share} \\
\hline
Corpus CBOW & 5 & 57/62/68/85/91 & 25\% & 0.90 \\
Corpus SGNS & 5 & 60/71/109/34/89 & 30\% & 0.78 \\
SVD-$\beta$ ($r=50$) & 5 & 97/154/75/21/16 & 42\% & 0.77 \\
Pretrained (Google News) & 5 & 89/33/109/64/68 & 30\% & 0.83 \\
LLM Embeddings & 5 & 25/114/59/101/64 & 31\% & 0.90 \\
\hline
\end{tabular}
\caption[Comparison of clustering approaches]{Comparison of clustering approaches (363 CBSAs). Cluster sizes are member counts by $k$-means label index. ``Within-cluster share'' is the fraction of embedding variance within clusters.}
\label{tab:cluster_comparison_363}
\end{table}

All five identify recognizable archetypes: energy dependence, university and government anchoring, deindustrialization, and diversified growth. I give per-cluster interpretations for all methods in Appendix~\ref{app-app:clusters}.
The number of clusters is chosen to balance interpretability and flexibility. I consider alternative values for $L$ in the Appendix (e.g., Online Appendix Figures \ref{app-fig:app_wcss_all} and \ref{app-fig:sensitivity_L}). The qualitative finding that these groupings separate economic narratives (and employment dynamics, cf Section \ref{sec:ar_panel}) holds across different choices of $L$.

\subsection{Local Employment Dynamics}\label{sec:ar_panel}

I return to our motivating example in the introduction: modelling local employment dynamics. 
In order to do so, I estimate autoregressive models for log employment across 363 CBSAs and examine whether the text-based clusters we derived above capture meaningful heterogeneity in employment dynamics.
Specifically, for each of the five embedding-based clustering methods (SVD-$\beta$, corpus-trained SGNS, corpus-trained CBOW, pretrained Word2Vec averaging, and LLM embeddings), I estimate:
\begin{equation}\label{eq:ar_grouped}
    y_{it} = \alpha_i + \rho_1^{(g_i)} y_{i,t-1} + \rho_2^{(g_i)} y_{i,t-2} + \varepsilon_{it},
\end{equation}
where $y_{it}$ is log total employment in year $t$ at the CBSA level, and $g_i \in \{1,\ldots,L\}$ denotes the cluster membership of CBSA $i$.\footnote{I construct a panel of geographically-harmonized metropolitan areas for which I can measure employment at an annual frequency starting in 1969. I use 2019 as the last year of analysis to avoid complications from the COVID-19 pandemic and recession. I focus on wage and salary employment, which is constructed by the Bureau of Economic Analysis (BEA) using administrative data and is less sensitive than self employment to changes over time in the tax code. The covariates used for the observables benchmark in Section~\ref{sec:ar_panel} are 1970 cross-sections from the same source.}
 I take $p=2$ as the primary specification based on our discussion in Appendix~\ref{app-app:app_p_sweep}. I thus allow the AR(2) slope coefficients in \eqref{eq:ar_intro} to vary across groups, but remain common \emph{within} each group---using our text-based cluster assignments.
The hope is that similarity in economic conditions between CBSAs translates into similar local employment dynamics. Under the topic model of Section~\ref{sec:theory}, I can make this condition more precise: that similarity in topic shares between CBSAs translates into similar local employment dynamics. As a concrete example, this might entail that descriptions with a relatively large share of words about manufacturing decline (treating this as a ``topic''), are associated with similar employment dynamics.

Figure~\ref{fig:ar_grouped}  depicts the group-specific implied impulse response functions (IRF) over a 10-year horizon under each clustering method. 
\begin{figure}[htb!]
\centering
\begin{minipage}{0.48\textwidth}
    \centering
    \includegraphics[width=\textwidth]{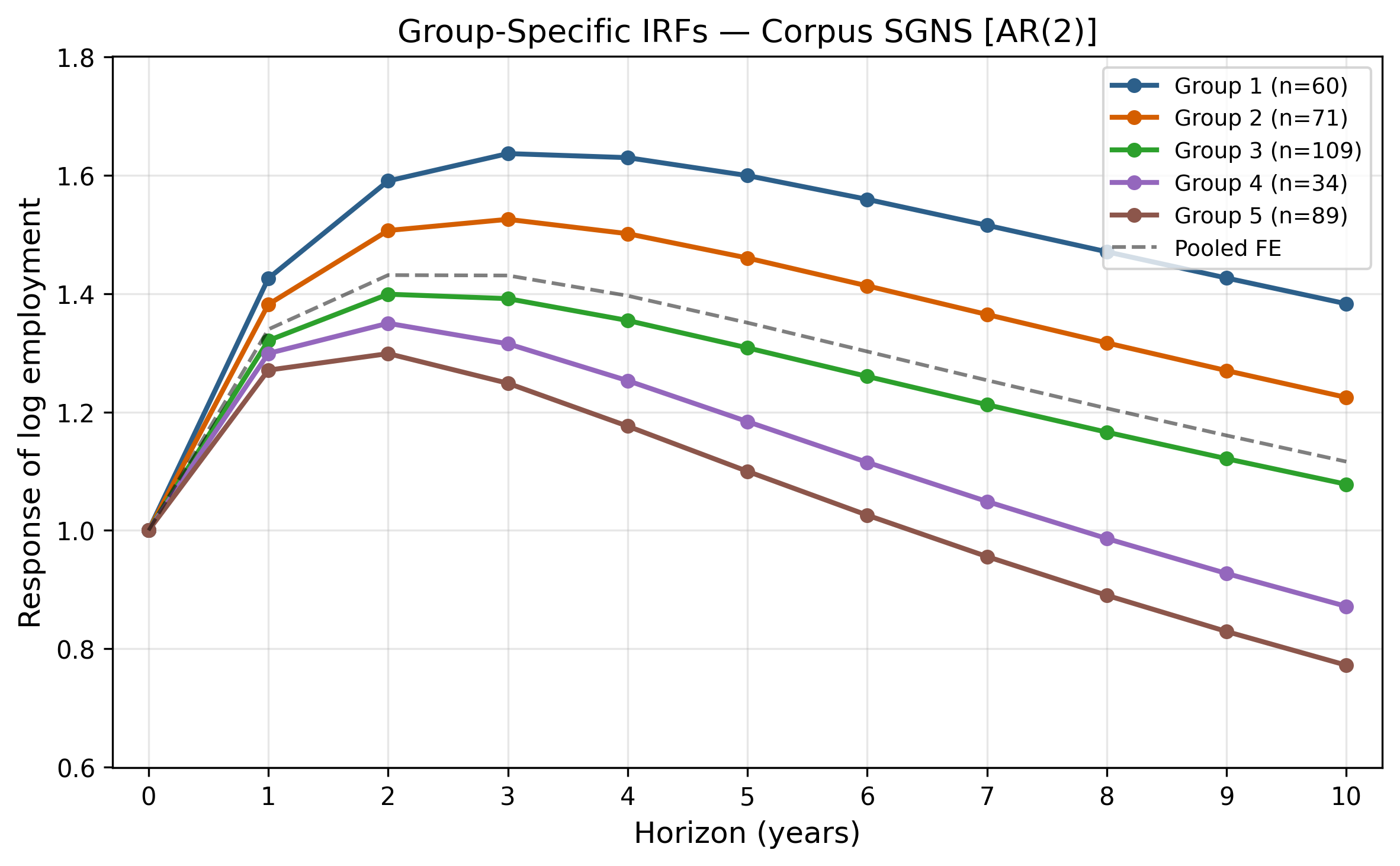}
\end{minipage}\hfill
\begin{minipage}{0.48\textwidth}
    \centering
    \includegraphics[width=\textwidth]{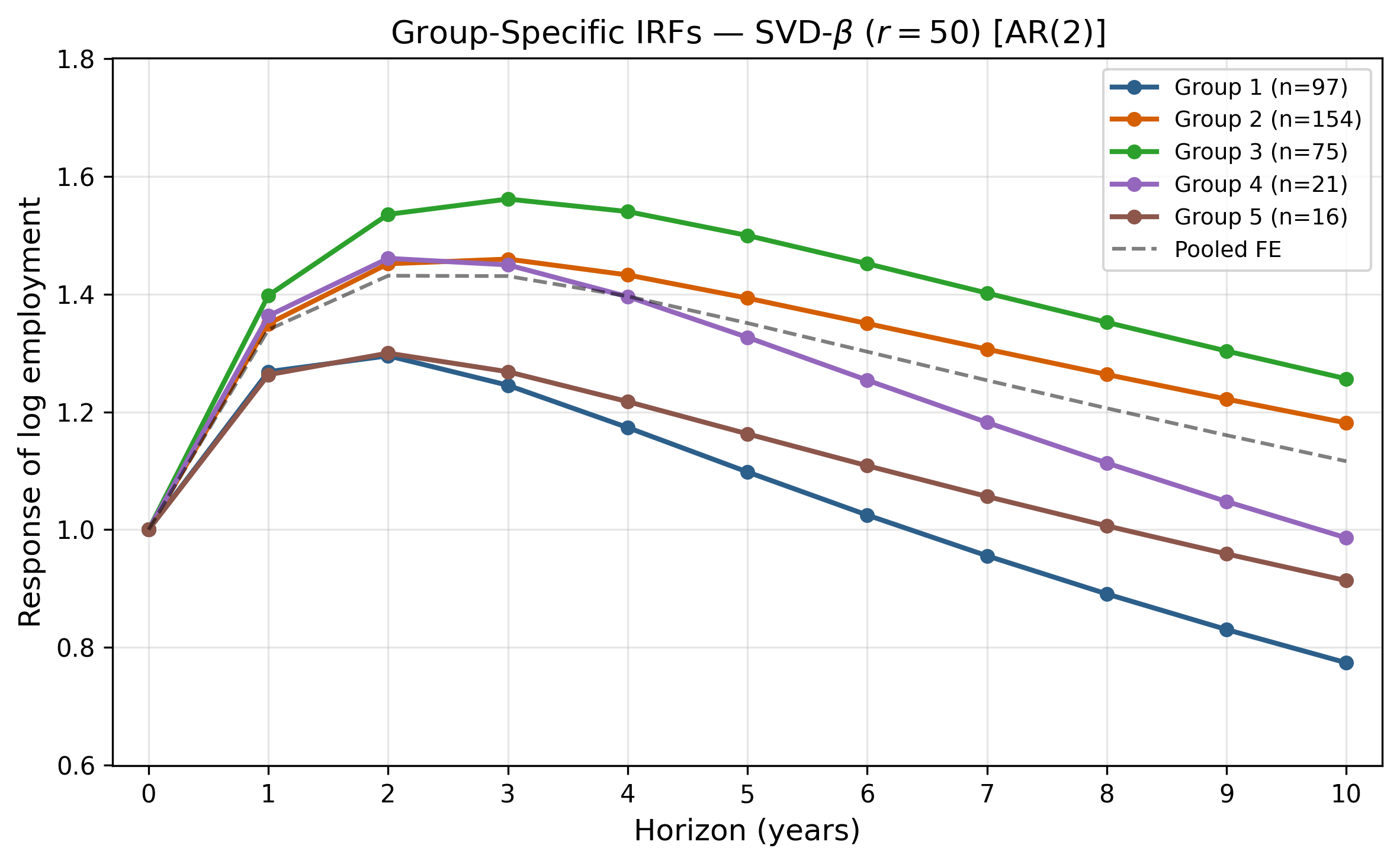}
\end{minipage}

\vspace{0.5em}

\begin{minipage}{0.48\textwidth}
    \centering
    \includegraphics[width=\textwidth]{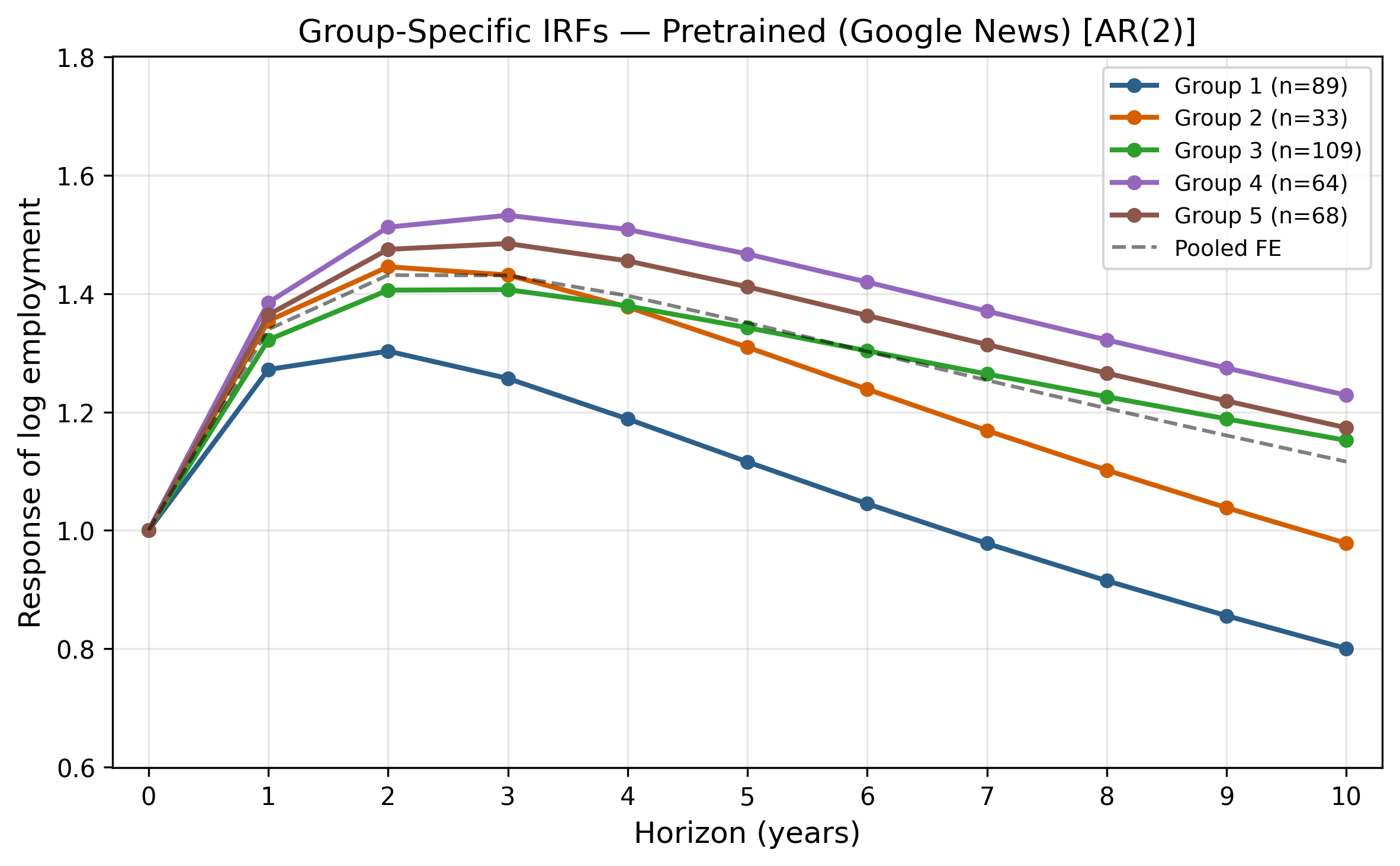}
\end{minipage}\hfill
\begin{minipage}{0.48\textwidth}
    \centering
    \includegraphics[width=\textwidth]{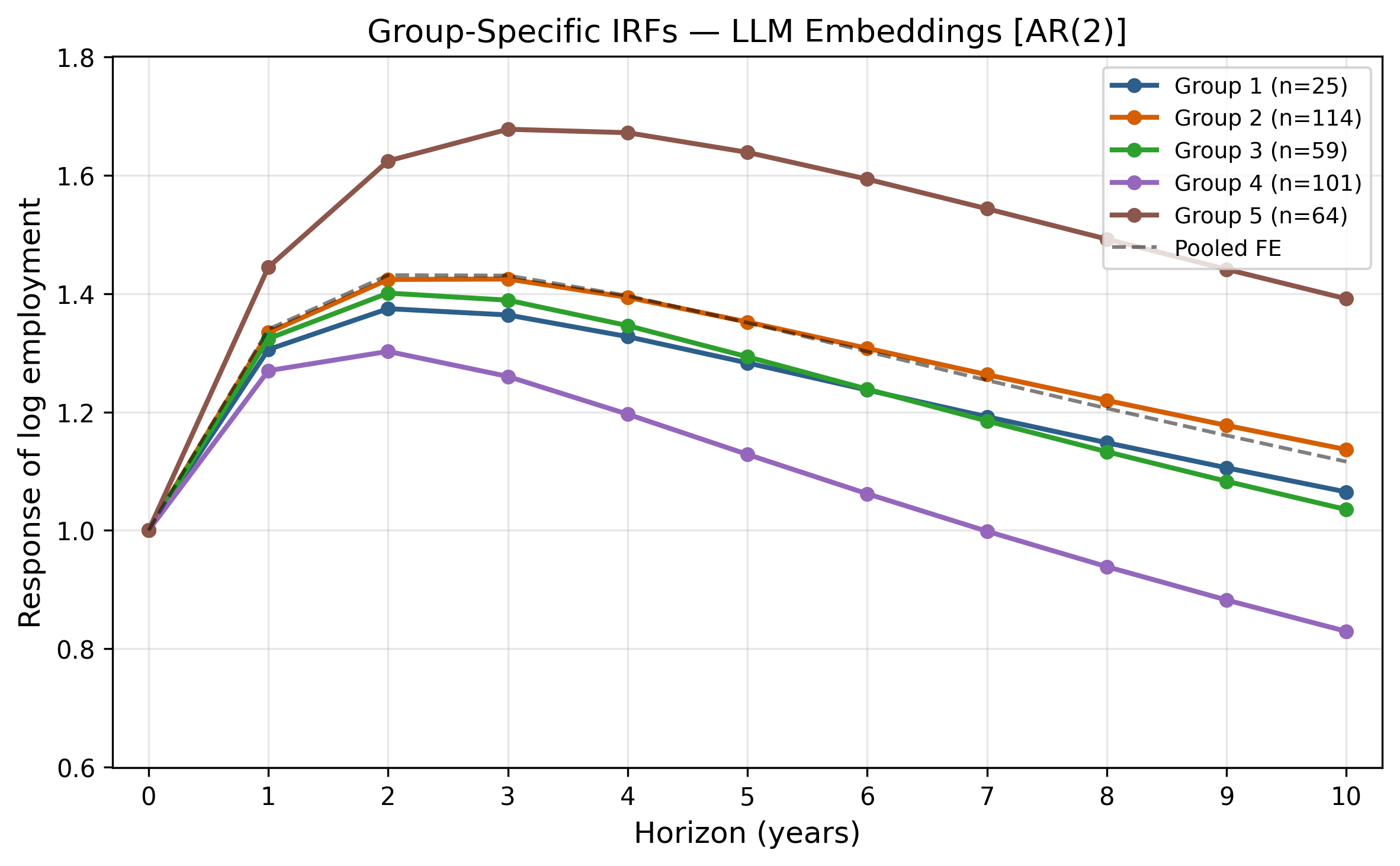}
\end{minipage}

\vspace{0.5em}

\begin{minipage}{0.48\textwidth}
    \centering
    \includegraphics[width=\textwidth]{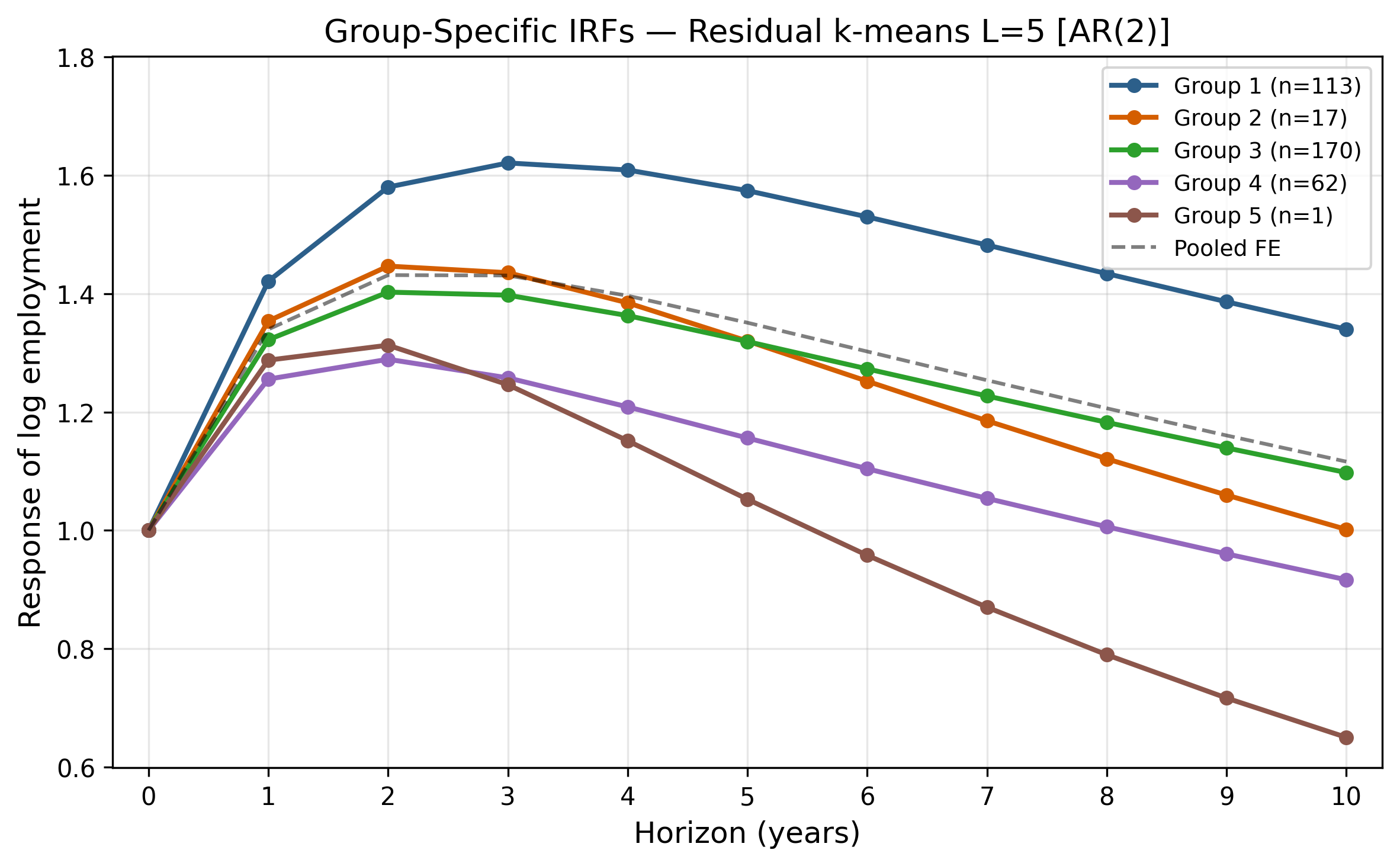}
\end{minipage}\hfill
\begin{minipage}{0.48\textwidth}
    \centering
    \includegraphics[width=\textwidth]{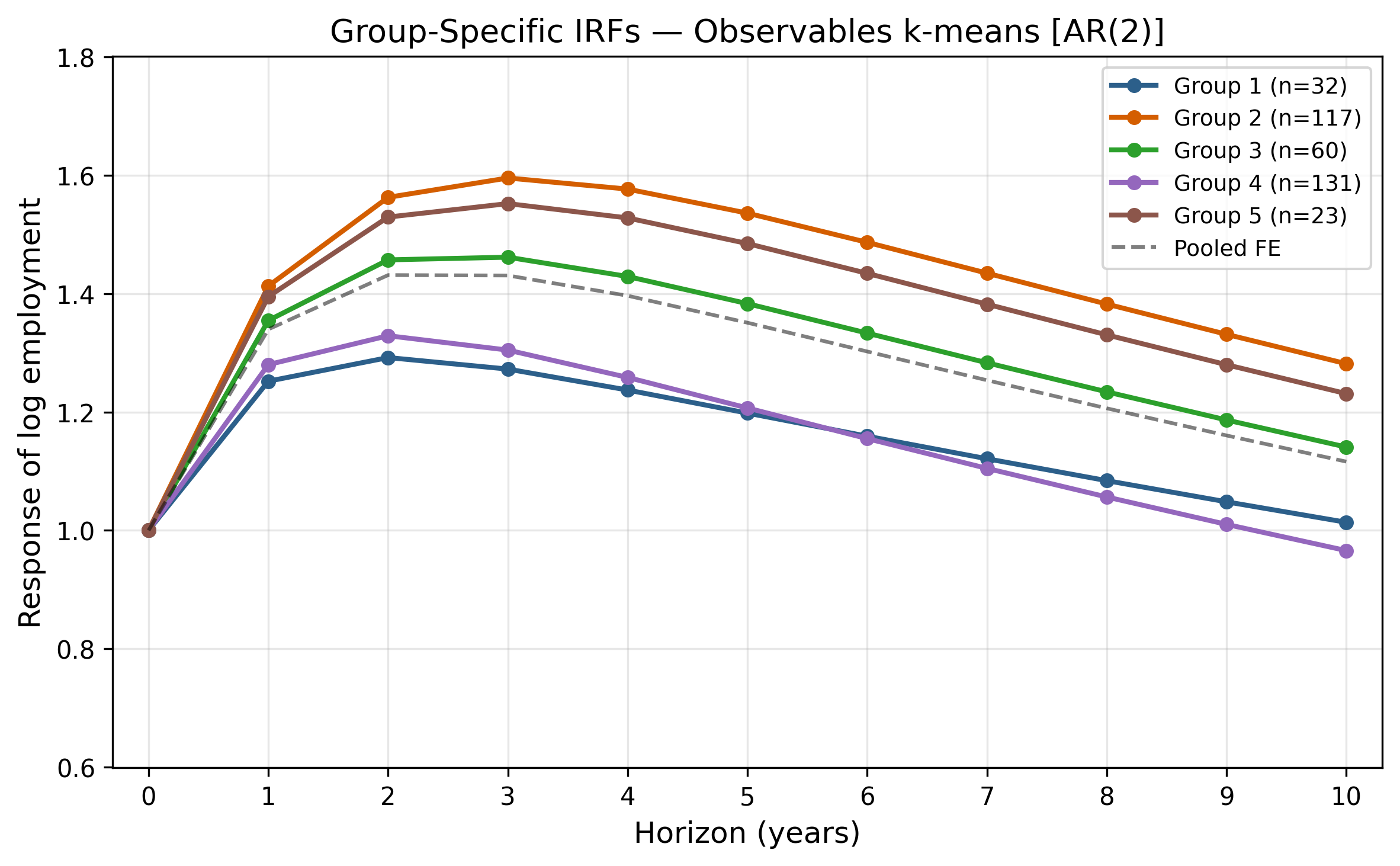}
\end{minipage}
\caption{Group-specific AR(2) impulse responses under each clustering ($L=5$), with the pooled FE estimate (dashed) for reference. The top four panels are embedding-based groupings; corpus CBOW is omitted here for space and appears in Table~\ref{tab:irf_ratio_main}. The bottom row gives the two groupings that do not use the text---residual $k$-means and observables $k$-means, on the 1970 covariates.}
\label{fig:ar_grouped}
\end{figure}
The top four panels use embedding-based groupings. The bottom row gives the two groupings that do not use the text. The bottom left uses residual $k$-means, which is formed by applying $k$-means clustering to the residuals from a pooled AR(2) specification---analogous to the first step of \cite{bonhomme2022discretizing}. The bottom right is observables $k$-means, applied to 1970 covariates (standardized industry mix, government and military employment shares, education, foreign-born and minority shares and population density; see Appendix~\ref{app-app:app_covariate_vintage} for more detail). 
I use $L=5$ throughout to match the number of clusters used by the text-based methods and report alternative specifications in Appendix~\ref{app-app:app_kmeans_sensitivity}.

To quantify how well each clustering method captures heterogeneity in employment dynamics, I compare the between-group variance of the 363 unit-specific IRFs across clustering methods in Table~\ref{tab:irf_ratio_main}. I do this in two ways: first, I compute the ratio of between-group variance to total variance, which is a standard measure of how well a grouping explains variation in a variable of interest (column 1). Since the total variance includes estimation error in the unit-specific responses, which is substantial in our application, I also include a version that normalizes the between-group variance relative to a permutation floor, which is the mean between-group variance over 2{,}000 random relabelings that hold the realized group sizes fixed (column 2). This is the level of between-group variance a random partition of the same sizes would deliver and provides a reference point for evaluating whether a given partition captures real heterogeneity in employment dynamics.
The last two columns count how frequently the grouping exceeds the two non-text based benchmarks when I vary the number of clusters $L\in\{2,3,4,5,6,8,10\}$.

\begin{table}[htb!]
\centering
\small
\setlength{\tabcolsep}{3pt}
\begin{tabular}{lcccc}
\hline
Grouping & \% Between & Ratio to floor & \multicolumn{2}{c}{Exceeds benchmark, of 7 $L$ values} \\
\cline{4-5}
 & & & Residual $k$-means & Observables $k$-means \\
\hline
Corpus SGNS & 28.3\% & 25.7$\times$ & 7 & 6 \\
LLM Embeddings & 27.1\% & 24.7$\times$ & 7 & 6 \\
SVD-$\beta$ ($r=50$) & 24.3\% & 21.8$\times$ & 7 & 6 \\
Corpus CBOW & 22.5\% & 20.4$\times$ & 7 & 5 \\
Pretrained (Google News) & 18.8\% & 16.8$\times$ & 6 & 5 \\
\hline
Observables $k$-means & 16.7\% & 15.5$\times$ & 7 & --- \\
Residual $k$-means & 13.9\% & 12.6$\times$ & --- & 0 \\
\hline
\end{tabular}
\caption[Between-group variance against the permutation floor]{Grouping-level summary at $L=5$, AR(2). ``\% Between'' is the share of total cross-city IRF variance explained by between-group differences in group-mean IRFs. ``Ratio to floor'' is the between-group variance of the 363 unit-specific impulse responses divided by that grouping's own size-matched permutation floor. The last two columns count how frequently the grouping exceeds each benchmark out of 7 cluster counts we consider. The upper block is the four embedding-based groupings, the lower the two that do not use the text.}
\label{tab:irf_ratio_main}
\end{table}

At $L=5$,  every grouping separates employment dynamics far more sharply than chance, and the five text-based groupings tend to separate it more sharply than the two benchmarks.\footnote{I date the covariates to 1970, so that they are effectively predetermined for a panel beginning in 1969. Repeating the exercise at later vintages raises its ratio monotonically, from $\ratioObsSeventy\times$ at 1970 to $\ratioObsMillennium\times$ at 2000. %
Nevertheless four of the five text methods exceed the benchmark at every vintage. I acknowledge that the same objection (not being predetermined) applies to the textual descriptions, which span 1979--2019 and are thus concurrent with the outcome window. This is because it is challenging to obtain a good textual description of the economy in the 1960s for some of the smaller CBSAs in the sample. Perhaps the modest increase in performance from using later vintages of observables, well short of the margin by which text exceeds the two benchmarks, is reassuring in that regard.} %
Varying $L$ does not change the main finding: The textual descriptions seem to carry real information about employment dynamics, more so than the two benchmarks I consider. Among the text-based methods, pretrained Word2Vec averaging tends to perform worst, while SVD-$\beta$, corpus-trained SGNS, and the LLM embeddings tend to perform best, trading top spots at different values of $L$ (cf. Appendix~\ref{app-app:app_kmeans_sensitivity}).
It is perhaps worth noting that a simple closed-form factorization computed from the corpus in seconds matches a commercial transformer embedding on this task.

\paragraph{In-sample vs Out-of-sample.} Our exercise in this section is entirely in-sample.%
I read the evidence as establishing that narrative text encodes economically meaningful structure about local labor markets---more of it than either benchmark recovers, including the curated covariate set---but not as establishing that text-based peer groups are the preferred way to model this panel for forecasting purposes.

\section{Conclusion}\label{sec:conclusion}

Embeddings are ubiquitous in NLP and increasingly used in economics. 
This paper asks what replacing a text with its embedding preserves, under a generative model in which documents are mixtures of $K$ latent topics.

At the word level, an embedding built from how often words appear together reflects the topic geometry: words that play the same role across topics end up close together. Extra dimensions do no harm, so a researcher need not know how many topics a corpus contains.
At the document level, proximity in embedding space is proximity in topic mixture. This means that a cluster of document embeddings corresponds to a set of documents with similar mixtures, and that adjusting for the embedding is equivalent to adjusting for the topic mixture.

Another lesson is that dimension alone is not the relevant property. 
Two alignments matter: the embedding must recover the text's low-dimensional structure, and that structure must be the one the analysis requires. My theory delivers the first. 
At the same time, connecting embeddings to the parameters of an underlying topic model makes the second an assumption about the economics rather than the algorithm. Take the control use: the informal premise that ``the embedding is a sufficient control'' reduces to the transparent condition that the topic mixture captures the confounding---an assumption with economic content that a practitioner can defend. Directly assuming that $d$-dimensional embedding vectors in Euclidean space are a sufficient control is much harder to justify on economic grounds.

In my application, I find that clustering LLM-generated descriptions of 363 U.S.\ metropolitan areas yields interpretable economic types, and those types separate local employment dynamics more sharply than a curated set of industry and demographic covariates. The comparisons line up with the theory: SVD-$\beta$ and corpus-trained SGNS, which target closely related population objects, perform similarly, while pretrained averaging, which carries the geometry of a different corpus, trails at almost every number of clusters.
Remarkably, our closed-form embedding matches a commercial transformer embedding on this task, at a fraction of the cost and with a guarantee attached.

Several extensions merit investigation. Transformer embeddings currently fall outside our theoretical framework. Characterizing what transformer architectures learn about topic structure---building on \cite{li2023transformers}---would bring the embeddings practitioners actually use inside the theory. 
A second extension is dynamic topic models \citep{blei2006dynamic}: characterizing how embeddings track changes in word meaning over time.

Finally, the pipeline studied here---generate text with a language model, embed it, then use the embedding downstream---retains the representational power of a large model while leaving every downstream step inspectable and reproducible. That division of labour seems a reasonable way to incorporate such models into empirical work.

\clearpage

\onehalfspacing
\allowdisplaybreaks

\newgeometry{margin=24mm}

\captionsetup[table]{name=Online Appendix Table}
\captionsetup[figure]{name=Online Appendix Figure}

\renewcommand{\theequation}{OA.\arabic{equation}}
\renewcommand{\thecorollary}{OA.\arabic{corollary}}
\renewcommand{\theproposition}{OA.\arabic{proposition}}
\renewcommand{\thedefn}{OA.\arabic{defn}}
\renewcommand{\theassumption}{OA.\arabic{assumption}}
\renewcommand{\thefigure}{OA.\arabic{figure}}
\renewcommand{\thetable}{OA.\arabic{table}}
\renewcommand{\thelemma}{OA.\arabic{lemma}}
\renewcommand{\theremark}{OA.\arabic{remark}}

\setcounter{equation}{0}
\setcounter{corollary}{0}
\setcounter{proposition}{0}
\setcounter{defn}{0}
\setcounter{assumption}{0}
\setcounter{figure}{0}
\setcounter{table}{0}
\setcounter{lemma}{0}
\setcounter{remark}{0}

\begin{center}
\textbf{\LARGE Online Appendix}
\end{center}
\bigskip

\appendix

\section{Other Word-Embedding Algorithms under the Topic Model}\label{app-app:log_approximation}

Section~\ref{sec:other_embeddings} of the main paper states what other standard embedding algorithms
target and what those targets recover. This appendix supplies the background. Section~\ref{app-app:fo_working}
writes the four objectives in a common form.
Sections~\ref{app-app:fullsoftmax_lemma}--\ref{app-app:cbow_remark} derive the specific target of this algorithm by algorithm.
Section~\ref{app-app:log_geometry} then shows what the embedding inherits from the topic model, using only that every pair of words co-occurs.

\subsection{The algorithms as weighted low-rank fits}\label{app-app:fo_working}

In the word embedding literature, $log (R)$ is called the pointwise mutual information $\mathrm{PMI}(v,u) = \log R_{vu}$ \citep{levy2014neural}, and I follow this convention here. Set
\begin{align}\label{app-eq:kappaR}
R_{\min} := \min_{v,u} R_{vu}, \qquad R_{\max} := \max_{v,u} R_{vu}, \qquad \kappa_R := R_{\max}/R_{\min}.
\end{align}

Each of the four algorithms scores the pair $(v,u)$ by an inner product $X_{vu} = w_v^\top \tilde w_u$ and fits that score to the data by minimizing a loss. Three ingredients pin the loss down: a \emph{response} $t_{vu}$ read off the corpus, a convex generator $\phi$, and weights $\omega_{vu} > 0$ on the pairs. Response and score are  different kinds of object. The response is data, fixed once the corpus is given; the score is the parameter, constrained to rank $r$. %
Lemma~\ref{app-lem:link} then shows that the weights drop out of the population optimum, so what the score converges to is determined by $t_{vu}$ and $\phi$ alone.

Take the generator first. Fix a strictly convex, differentiable $\phi$ on an interval $\mathcal{I}$ and set
\begin{align}\label{app-eq:bregman}
D_\phi(t, y) \;:=\; \phi(t) - \phi(y) - \phi'(y)(t-y), \qquad t, y \in \mathcal{I},
\end{align}
the vertical gap at $t$ between $\phi$ and its tangent line at $y$ --- the \emph{Bregman divergence} generated
by $\phi$ \citep{bregman1967relaxation} (Also see \cite{collins2001generalization, udell2016generalized}). Convexity makes $D_\phi \geq 0$ and strict
convexity makes it vanish only at $t = y$, so $D_\phi(t, \cdot)$ measures discrepancy from $t$; it is not
symmetric and is not a metric. Every loss below is a $D_\phi$ for one of three generators,
\begin{align*}
\phi(y) = \tfrac12 y^2, \qquad
\phi(y) = \textstyle\sum_v y_v \log y_v, \qquad
\phi(y) = y\log y + (1-y)\log(1-y),
\end{align*}
which give squared error, the Kullback--Leibler divergence, and the Bernoulli log-loss respectively.

The previous paragraph gave us a discrepancy, the Bregman divergence. But there is a scale mismatch: the score $X_{vu}$ lives in $\mathbb{R}$. On the other hand, the domain of the response $t_{vu}$, $\mathcal{I}$, depends on the algorithm, but generally $\mathbb{R} \ne \mathcal{I}$ (cf. Table \ref{tab:targets}). The derivative $\phi'$ is the \emph{link} that maps $\mathbb{R}$ onto the response scale:
Strict convexity makes $\phi'$ strictly increasing and hence invertible, so $(\phi')^{-1}(X_{vu})$ is the score expressed on the response's scale --- the fitted value for $t_{vu}$.
A loss of the form $D_\phi\bigl(t_{vu},\, (\phi')^{-1}(X_{vu})\bigr)$ therefore compares response with fitted value, and I call the matrix with entries $\phi'(t_{vu})$ the algorithm's \emph{target}.

\begin{lemma}[Weights do not move the target]\label{app-lem:link}
Let $\omega_{vu} > 0$ and $t_{vu} \in \mathrm{int}\,\mathcal{I}$ for all $v,u$. Then, the unconstrained minimizer of
$\sum_{v,u}\omega_{vu} D_\phi\bigl(t_{vu},\, (\phi')^{-1}(X_{vu})\bigr)$ over $X \in \mathbb{R}^{V\times V}$ is
$X_{vu} = \phi'(t_{vu})$, entrywise and independently of the weights.
\end{lemma}
\begin{proof}
Each term is nonnegative and vanishes iff $(\phi')^{-1}(X_{vu}) = t_{vu}$; $\phi'$ is
strictly increasing, hence invertible, so this holds iff $X_{vu} = \phi'(t_{vu})$. The weights are positive, so
minimizing the sum minimizes each term.
\end{proof}

Lemma~\ref{app-lem:link} states that the score equals the target whenever the rank constraint is non-binding. %

\subsection{Full-softmax skip-gram}\label{app-app:fullsoftmax_lemma}

The population full-softmax skip-gram objective is
\begin{align}\label{app-eq:fullsoftmax_pop_app}
\mathcal{L}_{\mathrm{full}}(W, \tilde W) \;:=\; \sum_{v, u \in [V]} M_{vu}\, \log\!\frac{\exp(w_v^\top \tilde w_u)}{\sum_{v' \in [V]} \exp(w_{v'}^\top \tilde w_u)},
\end{align}
where $w_v$ is the $v$-th row of $W$ and $\tilde w_u$ is the $u$-th row of $\tilde W$, so that $X := W\tilde W^\top$ collects the
scores $X_{vu} = w_v^\top \tilde w_u$. %
Equation~\eqref{app-eq:fullsoftmax_pop_app} is the population counterpart of the skip-gram model of \citet[eqs.\ 1--2]{mikolov2013distributed}, whose objective averages$\log P(w \mid c)$ over corpus positions within a context window; here that average is replaced by the
population co-occurrence $M$. 

\begin{lemma}[Full-softmax skip-gram target]\label{app-lem:fullsoftmax}
Under Assumptions~\ref{ass:topic_reg} and~\ref{ass:R_pos}, the unconstrained maximizers of
\eqref{app-eq:fullsoftmax_pop_app} satisfy $X^* = \log R + \log q\,\mathbf{1}^\top - \mathbf{1}g^\top$
for some $g \in \mathbb{R}^V$.
\end{lemma}
\begin{proof}
Lemma~\ref{app-lem:link} does not apply directly: the softmax in \eqref{app-eq:fullsoftmax_pop_app} normalizes over the
whole column $X_{\bullet u}$, so the loss is not separable across $v$ within a column, and we need to consider entire columns instead.
Writing $M_{vu} = q_u P(w{=}v \mid c{=}u)$ and $s_{vu} := \mathrm{softmax}(X_{\bullet u})_v$,
\begin{align}\label{app-eq:softmax_kl}
\mathcal{L}_{\mathrm{full}} \;=\; -\sum_u q_u\,\mathrm{KL}\bigl(P(w \mid c{=}u)\,\big\|\, s_{\bullet u}\bigr)
\;-\; \sum_u q_u H\bigl(P(w \mid c{=}u)\bigr),
\end{align}
whose second term does not involve $X$. So the objective is again a weighted Bregman fit, with response
$t_{vu} = P(w{=}v \mid c{=}u)$, the entropy generator, and weight $\omega_{vu} = q_u$; what differs from
Lemma~\ref{app-lem:link} is only that the divergence is taken one column at a time rather than one entry at a
time. It is maximized at
$s_{\bullet u} = P(w \mid c{=}u)$ for every $u$, that is at $X_{vu} = \log P(w{=}v \mid c{=}u) + g_u$, the
column gauge $g$ being free because $\mathrm{softmax}$ is invariant to adding a constant to a column. Finally
$\log P(w{=}v \mid c{=}u) = \log R_{vu} + \log q_v$ by definition of $R$.
\end{proof}

\subsection{SGNS}\label{app-app:sgns_lemma}

Computing the population full-softmax skip-gram objective is computationally difficult. \cite{mikolov2013distributed} therefore introduce negative sampling as an approximation used in practice, which I consider next.

Summing the pair-specific objective of \citet[eq.\ 5]{levy2014neural} over $(v,u)$ and normalizing counts by
corpus size gives the population SGNS objective
\begin{align}\label{app-eq:sgns_pop_app}
\mathcal{L}_{\mathrm{pop}}(W, \tilde W) \;:=\; \sum_{v,u \in [V]}\!\Bigl\{ M_{vu}\, \log \sigma(w_v^\top \tilde w_u) \;+\; \nu\, q_v q_u\, \log \sigma(-w_v^\top \tilde w_u) \Bigr\},
\end{align}
with $\sigma(x) := 1/(1+e^{-x})$, $\nu \geq 1$ the negative-sampling rate, and $w_v, \tilde w_u$ as above.\footnote{Implementations draw negatives from a smoothed unigram distribution rather than from $q$: \texttt{gensim} uses $P_{0.75}(u) \propto q_u^{3/4}$ \citep{mikolov2013distributed, levy2015improving}. This replaces $q_vq_u$ by $q_vP_{0.75}(u)$ in \eqref{app-eq:sgns_pop_app} and shifts the target to $\log R_{vu} + \log\bigl(q_u/P_{0.75}(u)\bigr) - \log\nu$. The added term is constant down each column, so it cancels in the row differences that Proposition~\ref{prop:oe_general} uses (cf.\ Remark~\ref{rem:other_targets}); I therefore work with the unsmoothed objective.}

\begin{lemma}[SGNS target]\label{app-lem:lg_restate}
 Under Assumptions~\ref{ass:topic_reg} and~\ref{ass:R_pos}, \eqref{app-eq:sgns_pop_app} is a weighted
Bernoulli-likelihood fit with weights $\omega_{vu} = q_vq_u(R_{vu}+\nu)$ and responses
$t_{vu} = R_{vu}/(R_{vu}+\nu)$, and its unconstrained maximizers satisfy
$X^* = \log R - \log\nu\,\mathbf{1}_V\mathbf{1}_V^\top$.
\end{lemma}
\begin{proof}
With $\omega_{vu} := M_{vu} + \nu q_vq_u = q_vq_u(R_{vu}+\nu)$ and $t_{vu} := M_{vu}/\omega_{vu} = R_{vu}/(R_{vu}+\nu)$,
\begin{align*}
\mathcal{L}_{\mathrm{pop}}(W, \tilde W) \;&=\; \sum_{v,u}\omega_{vu}\Bigl\{t_{vu}\log\sigma(X_{vu}) + (1-t_{vu})\log\sigma(-X_{vu})\Bigr\} \\
&=\; -\sum_{v,u}\omega_{vu} D_\phi\bigl(t_{vu}, \sigma(X_{vu})\bigr) \;+\; \mathrm{const},
\end{align*}
where $X := W \tilde W^\top$ and $\phi(y) = y\log y + (1-y)\log(1-y)$, for which \eqref{app-eq:bregman} is the Bernoulli
Kullback--Leibler divergence and $\phi' = \mathrm{logit} = \sigma^{-1}$. Assumption~\ref{ass:R_pos}
puts $t_{vu}$ in $(0,1)$, so Lemma~\ref{app-lem:link} applies and gives
$X^*_{vu} = \mathrm{logit}(t_{vu}) = \log(R_{vu}/\nu)$.
\end{proof}

The target is \citet[eqs.\ 6--7]{levy2014neural}: they let the inner products vary freely, which is the step
Lemma~\ref{app-lem:link} formalizes, and solve for $\vec w \cdot \vec c = \mathrm{PMI}(w,c) - \log \nu$. The
weighted-Bernoulli form is what I add. \citet[Section 3.2]{levy2014neural} observe that the loss on a pair
depends on its co-occurrence count and its expected negative-sample count, and conclude that SGNS performs a
weighted matrix factorization; they do not identify the weight $\omega_{vu}$ and response $t_{vu}$ that make it
one, and those are what place SGNS in the same frame as the other three algorithms in
Table~\ref{tab:targets}.

\subsection{GloVe}\label{app-app:glove_lemma}

The GloVe model \citep{pennington2014glove} minimizes the population objective
\begin{align}\label{app-eq:glove_pop_app}
\mathcal{L}_{\mathrm{GloVe}}(W, \tilde W, b, \tilde b) \;:=\; \sum_{v, u \in [V]} f(M_{vu})\, \Bigl(w_v^\top \tilde w_u \;+\; b_v \;+\; \tilde b_u \;-\; \log M_{vu}\Bigr)^2,
\end{align}
where $f: [0, \infty) \to [0, \infty)$ is positive on the support of $M$ and $b_v, \tilde b_u \in \mathbb{R}$
are learned per-word biases.

\begin{lemma}[GloVe target]\label{app-lem:glove}
Under Assumptions~\ref{ass:topic_reg} and~\ref{ass:R_pos} and with $f$ positive on
$\mathrm{supp}(M)$, the unconstrained minimizers of \eqref{app-eq:glove_pop_app} satisfy
$X_{vu} + b_v + \tilde b_u = \log M_{vu}$ entrywise. With $b^*_v = -\log q_v$ and
$\tilde b^*_u = -\log q_u$ the implied score matrix is $X^* = \log R$.
\end{lemma}
\begin{proof}
Apply Lemma~\ref{app-lem:link} with $\phi(y) = y^2/2$, response $\log M_{vu}$, and weights
$f(M_{vu})$. Decomposing $\log M_{vu} = \log q_v + \log q_u + \log R_{vu}$ and absorbing the marginals into
the biases leaves $X_{vu} = \log R_{vu}$.
\end{proof}

The biases do here what the gauge $g$ does in Lemma~\ref{app-lem:fullsoftmax}: both absorb the marginal offsets that separate $\log M$ from $\log R$. 

\subsection{CBOW}\label{app-app:cbow_remark}

CBOW \citep{mikolov2013efficient} reverses the conditioning relative to skip-gram: it predicts the target word
from the average of its surrounding context embeddings. The population objective is
\begin{align}\label{app-eq:cbow_pop_app}
\mathcal{L}_{\mathrm{CBOW}}(W, \tilde W) \;:=\; \sum_{w, \mathbf{u}} P(w, \mathbf{u})\, \log\!\frac{\exp\!\bigl(h_w^\top \bar h_{\mathbf{u}}\bigr)}{\sum_{w' \in [V]} \exp\!\bigl(h_{w'}^\top \bar h_{\mathbf{u}}\bigr)},
\end{align}
where $\mathbf{u} = (u_1, \ldots, u_{2J})$ is the context window and $\bar h_{\mathbf{u}} := \frac{1}{2J}\sum_{j=1}^{2J} h_{u_j}$
is the average context embedding. At the population optimum the softmax form yields
$h_w^\top \bar h_{\mathbf{u}} = \log P(w \mid \mathbf{u}) + \kappa(\mathbf{u})$ with $\kappa(\mathbf{u})$ the
per-context log-normalizer.

With a single-word context this is Lemma~\ref{app-lem:fullsoftmax}.. The same relabeling applies under negative sampling: with a one-word context the CBOW negative-sampling objective is \eqref{app-eq:sgns_pop_app} with $W$ and $\tilde W$ exchanged, and $M$ is symmetric, so its target is again $\log R - \log\nu\,\mathbf{1}_V\mathbf{1}_V^\top$ and Proposition~\ref{prop:oe_general} applies verbatim.
With more than one context word it does not: $h_w^\top \bar h_{\mathbf{u}}$ is a function of the context bundle
rather than of pairs $(w, c)$, so there is no $V \times V$ target matrix analogous to those in Table~\ref{tab:targets},
and a PMI factorization for CBOW analogous to Lemma~\ref{app-lem:lg_restate} remains open in the literature. 

\subsection{What the targets inherit from the topic model}\label{app-app:log_geometry}

The proof of Proposition~\ref{prop:oe_general} bounds distances between rows of $\log R$ and then carries
them to the embedding. That proposition is stated for the SGNS target, whose only offset is a scalar. The
following lemma isolates its first step and states it more generally --- for an arbitrary entrywise transform
and arbitrary offsets constant along a row or a column. The row offsets are what it takes to cover the
full-softmax and GloVe targets: both carry an offset constant along a row, and a row offset does not cancel in
a row difference (Remark~\ref{rem:other_targets}).

\begin{lemma}[Index sufficiency]\label{app-lem:index}
Let Assumptions~\ref{ass:topic_reg} and~\ref{ass:R_pos} hold, let $\psi$ be any function on
$(0,\infty)$, let $a, b \in \mathbb{R}^V$, and set $T_{vu} := \psi(R_{vu}) + a_v + b_u$. Then, for all
$v, v' \in [V]$, $\tilde B_{v\bullet} = \tilde B_{v'\bullet}$ implies
$T_{v\bullet} - T_{v'\bullet} = (a_v - a_{v'})\mathbf{1}^\top$.
\end{lemma}
\begin{proof}
$R = \tilde B G \tilde B^\top$ \eqref{eq:uncentered_identity} gives
$R_{vu} = \tilde B_{v\bullet}^\top G\, \tilde B_{u\bullet}$, so $\tilde B_{v\bullet} = \tilde B_{v'\bullet}$
gives $R_{v\bullet} = R_{v'\bullet}$ and hence $\psi(R_{v\bullet}) = \psi(R_{v'\bullet})$ entrywise. The
column offsets $b_u$ cancel in the difference.
\end{proof}

The hypothesis is the one used throughout: since $\tilde B = D_q^{-1}B$, $\tilde B_{v\bullet} = \tilde B_{v'\bullet}$
holds exactly when $B_{v\bullet} = \lambda B_{v'\bullet}$ for some $\lambda > 0$, in which case $q_v = \lambda q_{v'}$.
Words with proportional loadings therefore have target rows that differ by the constant $a_v - a_{v'}$, in every
one of the targets and whatever transformation produced them. Whether that constant is zero is what separates
the algorithms: for SGNS $a = 0$, so the rows coincide and Proposition~\ref{prop:oe_general} carries this
to the embedding; for full-softmax and GloVe $a \neq 0$, and Remark~\ref{rem:other_targets} works out what
survives.

\clearpage

\section{Additional Application Results}\label{app-app:application}

\subsection{Prompts and Generation Settings}\label{app-app:app_prompts}

Three separate calls to a language model appear in the application, and only the first affects the
partitions. This subsection records all three, since the corpus is the one input to the exercise that a
reader cannot inspect directly.

\paragraph{Generating the descriptions.} Each of the 363 descriptions comes from a single call to
\texttt{claude-opus-4-8} with a 1{,}500-token cap. The user message is
\begin{quote}\small
\texttt{Summarize the economic development of \{city\} from 1979-2019 in around 500 words with a focus
on the labor market. Focus on broad trends rather than specific numbers, and on aspects that make it
unique relative to the US as a whole.}
\end{quote}
and the system message is
\begin{quote}\small
\texttt{Write a substantive, best-effort description using your general knowledge of the place. You do
not need verified statistics or sources, and you should not refuse or add disclaimers about lacking data
or access -- always produce a description, doing the best you can even when you are uncertain. Start
right away with the description, without an opening sentence like `Here's a 500-word summary of
Philadelphia's economic development:'.}
\end{quote}

Four features of this are deliberate and worth stating.

First, ``broad trends rather than specific numbers'' is there to discourage invented statistics. The exercise needs qualitative economic structure, not a list of numbers, and the model is known to hallucinate statistics. Further, it will tend to recall numbers for locations it knows well and revert to broad trends for those it knows less well anyways, which would systematically change the structure of the description across locations.

Second, and relatedly, the refusal suppression is what makes the corpus complete. Without it the model declines or hedges for the smaller CBSAs it knows least about.

Third, ``unique relative to the US as a whole'' pushes the descriptions toward cross-sectional differentiation. A prompt asking simply for an economic history would return text dominated by features common to all US metros---national recessions, the general decline of manufacturing---which contribute a near-constant component to every document embedding and so cannot separate cities.

Fourth, the instruction to start immediately, without a preamble, matters for the same reason. A boilerplate opening clause repeated across 363 documents is a constant vector added to every average of
word embeddings, which shrinks the relative distances the clustering depends on.

The resulting corpus is 363 descriptions of 476 to 510 words, median 495, so document length is close to constant. %

\paragraph{Labelling the clusters.} The cluster names in Section~\ref{app-app:clusters} come from a second
call, made once per cluster per method, to \texttt{claude-opus-5} at low reasoning effort with a
2{,}000-token cap. It is given the member city names and their descriptions and asked
\begin{quote}\small
\texttt{Your response should include less than 60 words. Start right away to list the common themes.
Cities: \{cities\}. Based on the following documents, what do these cities have in common? Identify the
common themes: \{documents\}}
\end{quote}
These labels are purely \emph{ex post}. They are produced after the partitions are fixed, play no part in forming them, and enter no statistic reported anywhere in the paper---they exist only to interpret the clusters. %

\paragraph{The direct-LLM alternative.} The comparison in
Section~\ref{sec:application}'s footnote asks \texttt{claude-opus-5}, with a 4{,}000-token cap, to
partition the cities itself rather than embed them: it receives all 363 descriptions at once and is
asked to return exactly $L$ clusters as JSON, each with a name, a member list and a short explanation,
with the instruction that every city belong to exactly one cluster. It did not comply at this scale. The
returned partition covered only part of the sample and included place names absent from the input list,
which is why the paper routes through embeddings instead.

\subsection{Application Embedding Implementation}\label{app-app:app_methods}

This section specifies the five embedding methods compared in Section~\ref{sec:application}. %

\paragraph{Common to all five.} Every method starts from the same 363 descriptions and differs only
in the map from a description to a vector. Text is lower-cased and tokenized on the regular expression
\texttt{[A-Za-z]+}\allowbreak\texttt{(?:'[A-Za-z]+)?}, which keeps letter strings with internal
apostrophes and drops digits and punctuation. The result is $183{,}633$ tokens and $7{,}411$ distinct word types, with
documents between $485$ and $525$ tokens long (median $506$), so the averages below are taken over a
near-constant number of terms. The four word-level methods embed words and then take the unweighted
within-document average, which is the $\mu_d$ of
Proposition~\ref{prop:doc_embedding_general_freq}; the fifth embeds the document directly. Each
of the five is then clustered by $k$-means with $L = 5$, \texttt{k-means++} seeding
\citep{arthur2007kmeans} and $500$ restarts, keeping the best fit. 

\subsubsection{Closed-Form SVD-$\beta$ Embedding}\label{app-sec:app_svd_beta}

We directly estimate $\widehat R$ from the corpus word-context co-occurrence matrix (taking the full document as context), center to $\widehat R - \mathbf{1}_V\mathbf{1}_V^\top$, and take the top-$r$ eigendecomposition. The resulting $\widehat\beta = \Phi\Lambda^{1/2}$ satisfies $\widehat\beta\widehat\beta^\top = \widehat R - \mathbf{1}_V\mathbf{1}_V^\top$ on the top-$r$ eigenspace. Document embeddings are the within-document average of $\widehat\beta_v$. The eigenpairs are computed by Lanczos iteration.

\paragraph{Vocabulary filter.} We restrict to words appearing in at least 5 documents, giving $V = 2{,}440$ types and $94.1\%$ token coverage. Rare words have very small marginal probabilities $P(w)$, inflating $1/P(w)$ in the centered ratio, making it computationally unstable, and dominating the resulting eigendecomposition.

\paragraph{Choice of dimension.} Proposition~\ref{thm:factorization} gives rank $K-1$, and $K$ is not observable for the CBSA corpus. By Remark~\ref{rem:dimension_slack} we do not need it: any $r \geq K-1$ recovers the same geometry in population, so the dimension only has to be chosen large enough. We set $r = 50$ as the primary specification. Section~\ref{app-app:app_d_sweep} reports the sweep over $r \in \{10, 25, 50, 100, 200\}$ in place of the elbow the spectrum does not supply.

\subsubsection{Corpus-Trained Word2Vec: SGNS}\label{app-sec:app_corpus_sgns}

Skip-gram with negative sampling using the \texttt{gensim}
implementation \citep{rehurek2010gensim}: $r = 50$ dimensions, \texttt{min\_count} $= 2$, $5$ epochs, and
the full document as context, implemented as window $J = 10{,}000$, which exceeds every document length. \texttt{gensim} draws the effective half-window uniformly on
$\{1, \ldots, \texttt{window}\}$ per token, so about $95\%$ of tokens see the entire document and the
rest a narrower one.  Section~\ref{app-app:app_context_window} reports the sweep over
$J \in \{5, 10, 20, 50\}$
The remaining hyperparameters are \texttt{gensim} defaults: $\nu = 5$ negative samples per update drawn
from the unigram distribution raised to the $0.75$ power, frequent-word downsampling at threshold
$10^{-3}$, and an initial learning rate of $0.025$ decayed linearly.

\paragraph{Epochs.} The $5$ epochs are set against the $200$ of the CBOW arm
(Section~\ref{app-sec:corpus_cbow}) because a common count would not be equal training. CBOW makes one
update per target token, while SGNS makes one per (target, context) pair, so at full-document
context SGNS performs roughly $500$ times as many updates per epoch. Matching CBOW's $200$ epochs
would give SGNS about $500$ times CBOW's total training rather than the same amount. At $5$ epochs SGNS
still performs about $12$ times CBOW's total updates, so it is not the under-trained arm. At the
narrower windows of Section~\ref{app-app:app_context_window} the ratio is small enough that both objectives
run $200$ epochs.

\paragraph{Coverage.} After \texttt{min\_count} $= 2$, the trained vocabulary contains $4{,}876$ word types covering $98.6\%$ of the $183{,}633$ tokens, against $86.3\%$ for the pretrained vectors of Section~\ref{app-sec:app_pretrained_w2v}, because the vocabulary is learned from the data rather than inherited. The trade-off is corpus size: $184$ thousand tokens is arguably small to train a Word2Vec model from scratch. 

\subsubsection{Corpus-Trained Word2Vec: CBOW}\label{app-sec:corpus_cbow}

\paragraph{Settings.} $r = 50$ dimensions, \texttt{min\_count} $= 2$, $200$ epochs, and the full
document as context. Full-document context is implemented as window  $J= 10{,}000$. The remaining hyperparameters are \texttt{gensim} defaults: $5$ negative
samples per update drawn from the unigram distribution raised to the $0.75$ power, frequent-word
downsampling at threshold $10^{-3}$, and an initial learning rate of $0.025$ decayed linearly.

\subsubsection{Pretrained Word2Vec: Google News}\label{app-sec:app_pretrained_w2v}

We average pretrained Google News Word2Vec vectors---$300$-dimensional,
trained on  part of a Google News corpus of about 100 billion words \citep{word2vec2013} and distributed through \texttt{gensim}---over each CBSA description, yielding 363 embeddings in
$\mathbb{R}^{300}$. These vectors cover $86.3\%$ of the corpus tokens; the remainder are dropped from
the average.

\subsubsection{LLM Embedding-Based Clustering}\label{app-sec:app_llm_embeddings}

Using OpenAI's \texttt{text-embedding-3-large}, I embed each of the 363 descriptions into $\mathbb{R}^{3072}$. This is a contextual transformer rather than an average of word vectors, so none of the results of Section~\ref{sec:theory} technically apply to it; it enters as a benchmark for current practice.

\subsection{Choice of $L$ for the Application Embeddings}\label{app-app:app_wcss}

For each embedding method in Section~\ref{sec:application}, I use $L=5$ clusters. Figure~\ref{app-fig:app_wcss_all} reports the within-cluster sum of squares (WCSS) as a function of $L$ for all five methods. In each case the elbow is gradual, but $L=5$ provides a reasonable trade-off. The absence of a sharp elbow is itself informative: it is what one expects when the document embeddings are spread over the simplex rather than concentrated at a small number of vertices, which is the regime Section~\ref{sec:theory} distinguishes from exact recovery.
The vertical axis is the share of embedding variance within clusters.

\begin{figure}[htb!]
\centering
\begin{subfigure}[b]{0.48\textwidth}
\centering
\includegraphics[width=\textwidth]{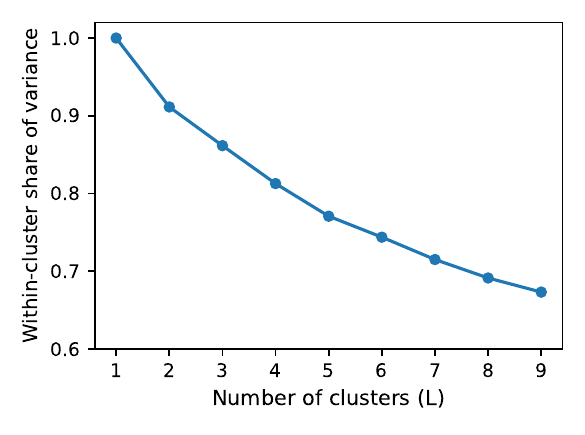}
\caption{SVD-$\beta$ ($r=50$)}
\label{app-fig:app_svd_beta_wcss}
\end{subfigure}
\hfill
\begin{subfigure}[b]{0.48\textwidth}
\centering
\includegraphics[width=\textwidth]{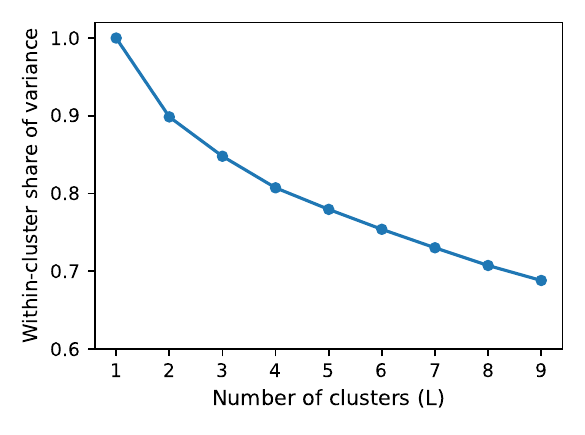}
\caption{Corpus SGNS}
\label{app-fig:app_sgns_wcss}
\end{subfigure}

\vspace{0.25em}

\begin{subfigure}[b]{0.48\textwidth}
\centering
\includegraphics[width=\textwidth]{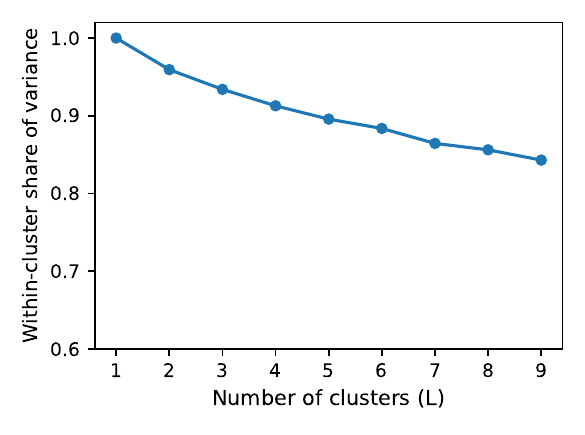}
\caption{Corpus CBOW}
\label{app-fig:app_trained_wcss}
\end{subfigure}
\hfill
\begin{subfigure}[b]{0.48\textwidth}
\centering
\includegraphics[width=\textwidth]{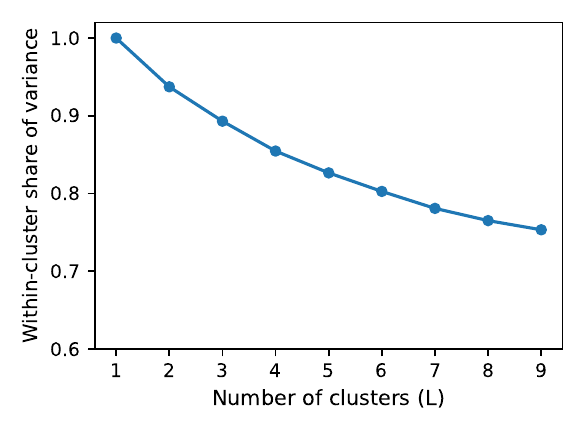}
\caption{Pretrained (Google News)}
\label{app-fig:app_w2v_wcss}
\end{subfigure}

\vspace{0.25em}

\begin{subfigure}[b]{0.48\textwidth}
\centering
\includegraphics[width=\textwidth]{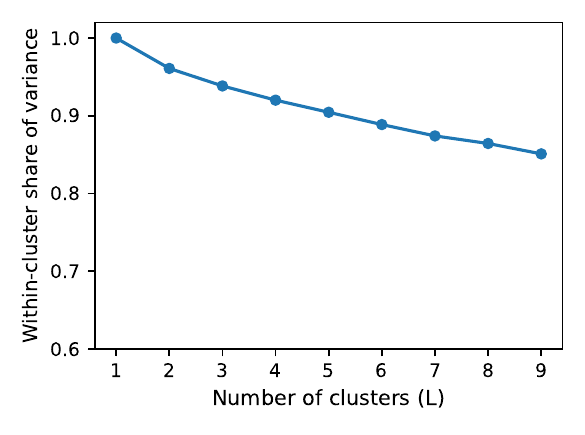}
\caption{LLM embeddings}
\label{app-fig:app_wcss}
\end{subfigure}
\caption[WCSS against the cluster count, by embedding method]{Within-cluster sum of squares as a function of the number of clusters $L$ for each of the five embedding methods (363 CBSAs), normalized by the total sum of squares. For $k$-means the centroids are the cluster means, so $\mathrm{TSS} = \mathrm{WCSS} + \mathrm{BCSS}$ holds exactly and the plotted quantity is the share of embedding variance still within clusters; one minus it is the share the partition explains.}
\label{app-fig:app_wcss_all}
\end{figure}

\clearpage

\subsection{Cluster Visualizations in UMAP Space}\label{app-app:app_umap}

\begin{figure}[htb!]
\centering
\begin{subfigure}[b]{0.48\textwidth}
\centering
\includegraphics[width=\textwidth]{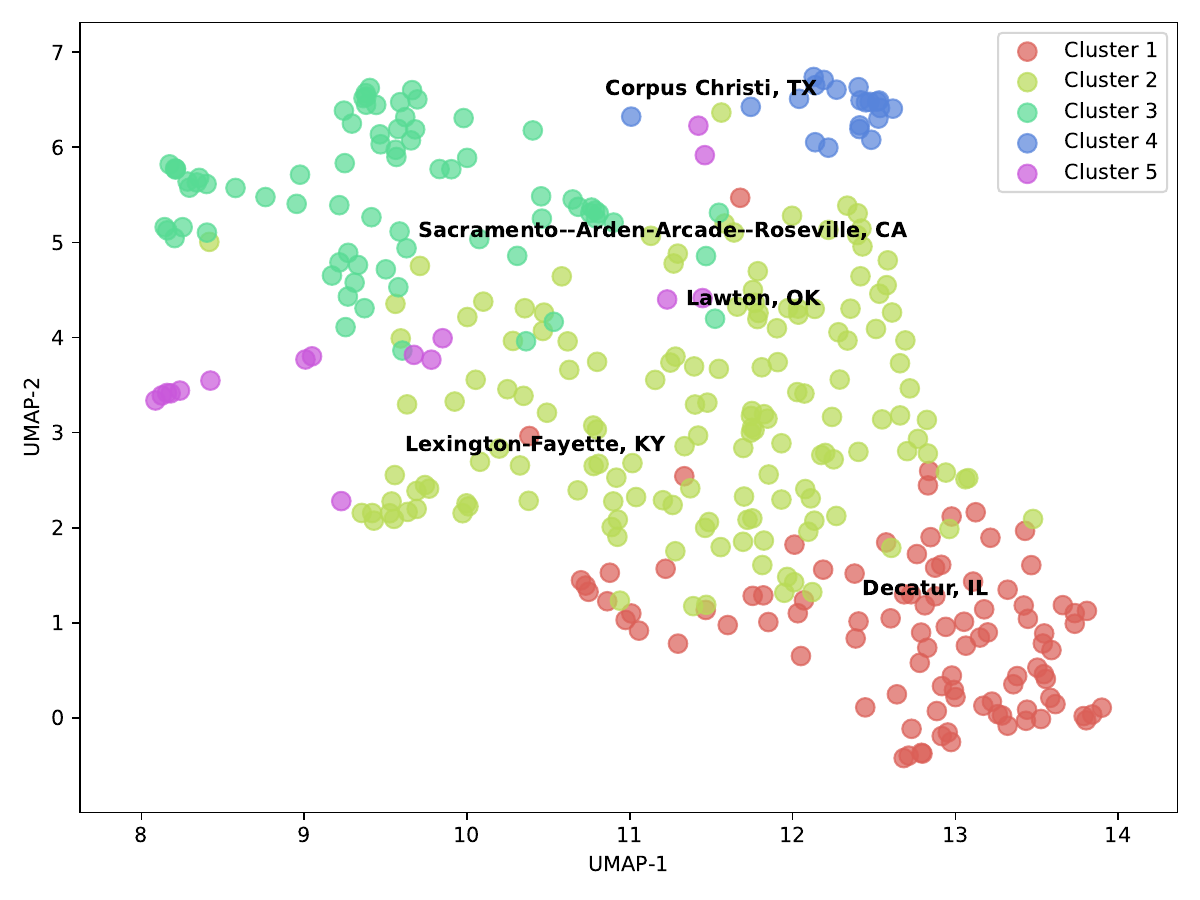}
\caption{SVD-$\beta$ ($r=50$)}
\label{app-fig:app_svd_beta_clusters}
\end{subfigure}
\hfill
\begin{subfigure}[b]{0.48\textwidth}
\centering
\includegraphics[width=\textwidth]{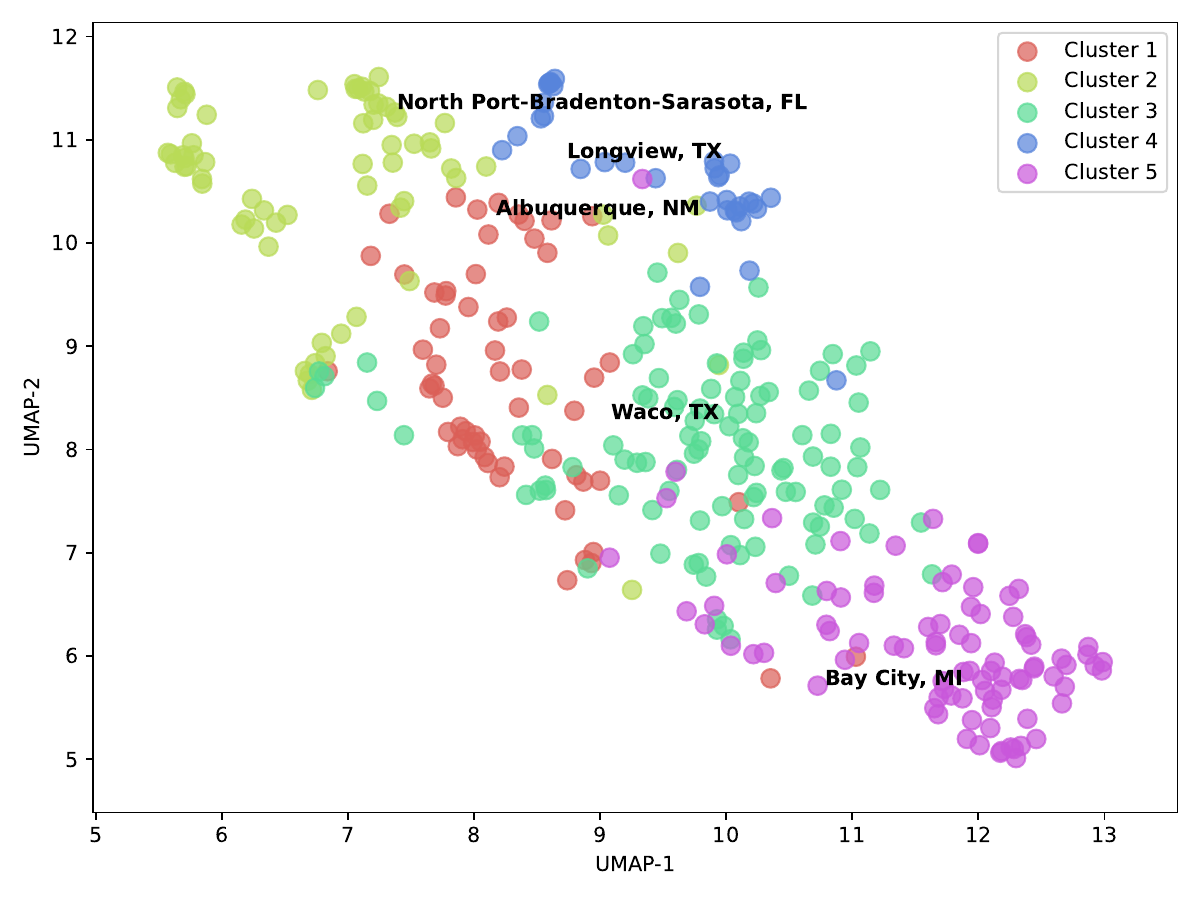}
\caption{Corpus SGNS}
\label{app-fig:app_sgns_clusters}
\end{subfigure}

\begin{subfigure}[b]{0.48\textwidth}
\centering
\includegraphics[width=\textwidth]{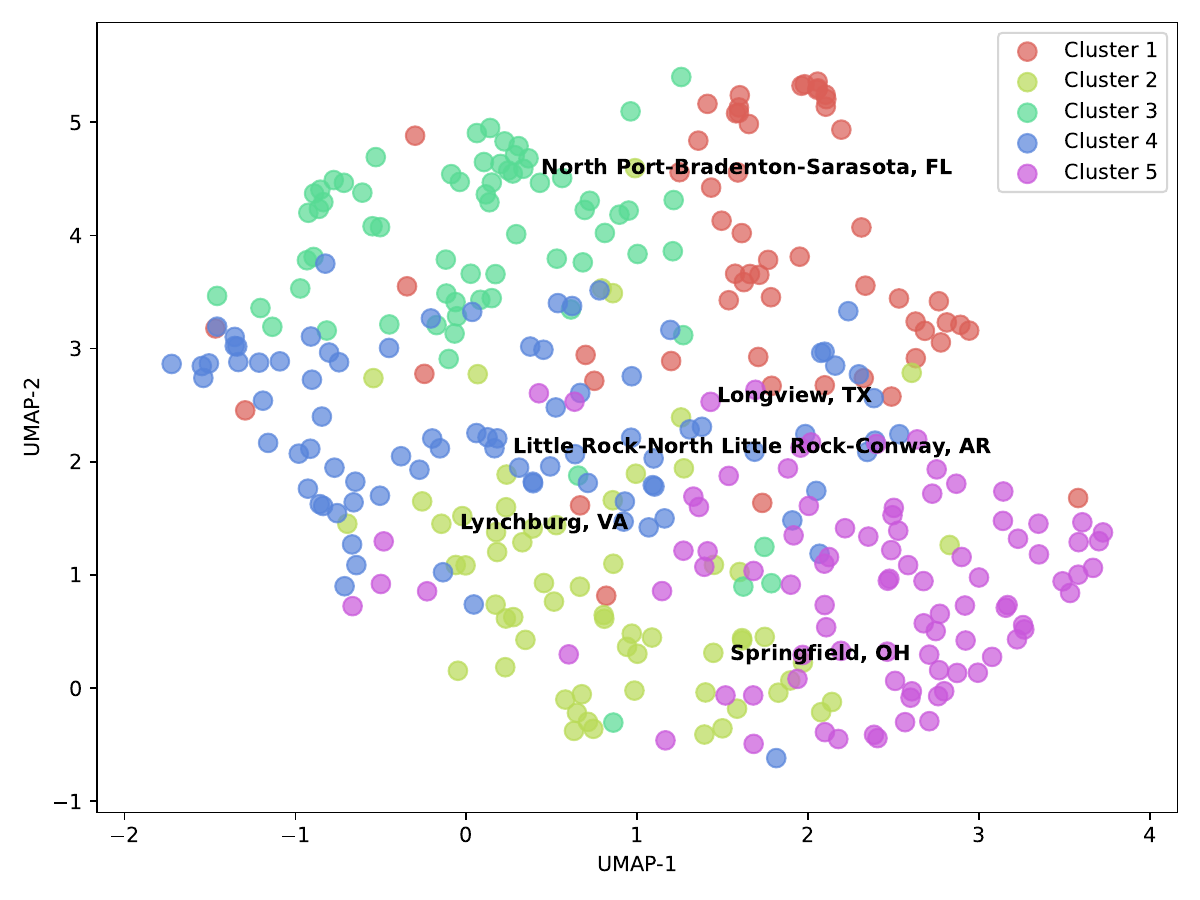}
\caption{Corpus CBOW}
\label{app-fig:app_trained_clusters}
\end{subfigure}
\hfill
\begin{subfigure}[b]{0.48\textwidth}
\centering
\includegraphics[width=\textwidth]{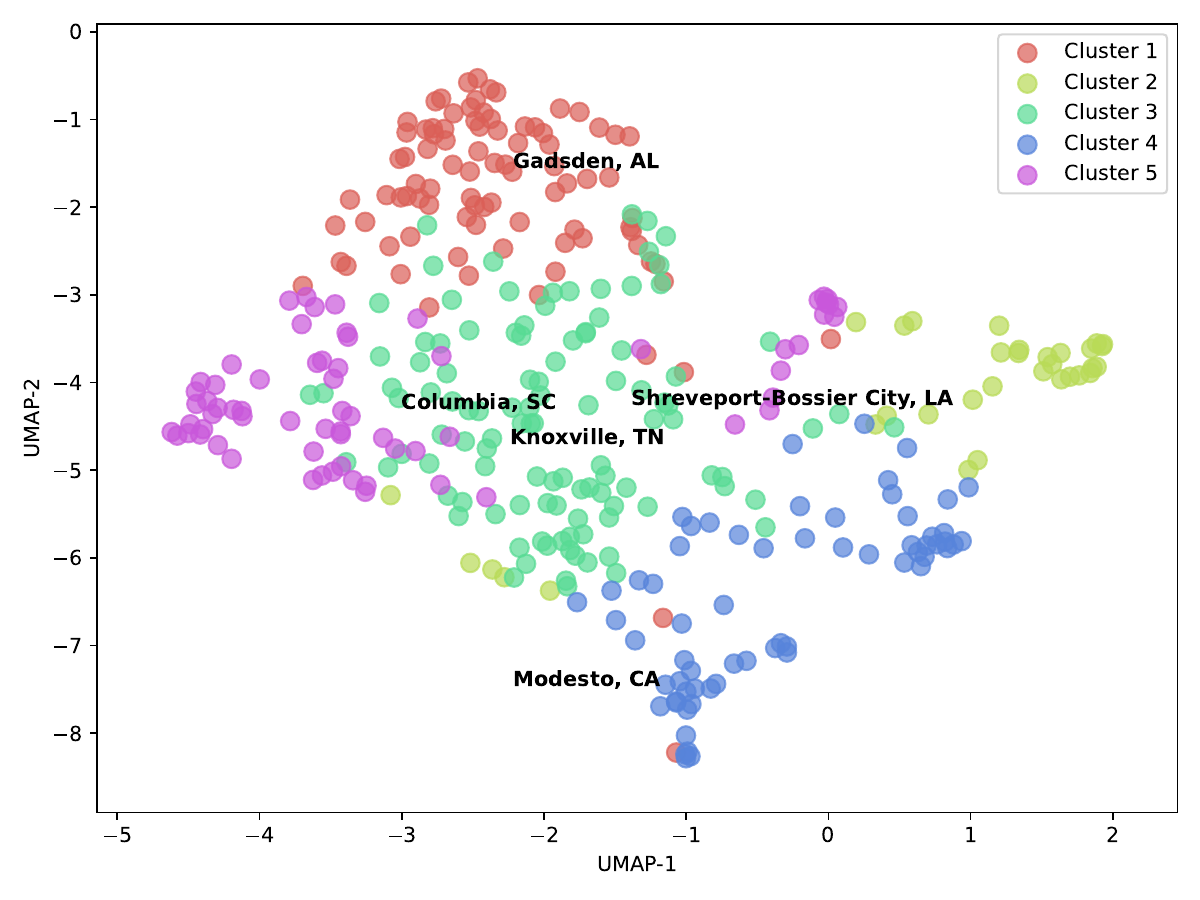}
\caption{Pretrained (Google News)}
\label{app-fig:app_w2v_clusters}
\end{subfigure}

\begin{subfigure}[b]{0.48\textwidth}
\centering
\includegraphics[width=\textwidth]{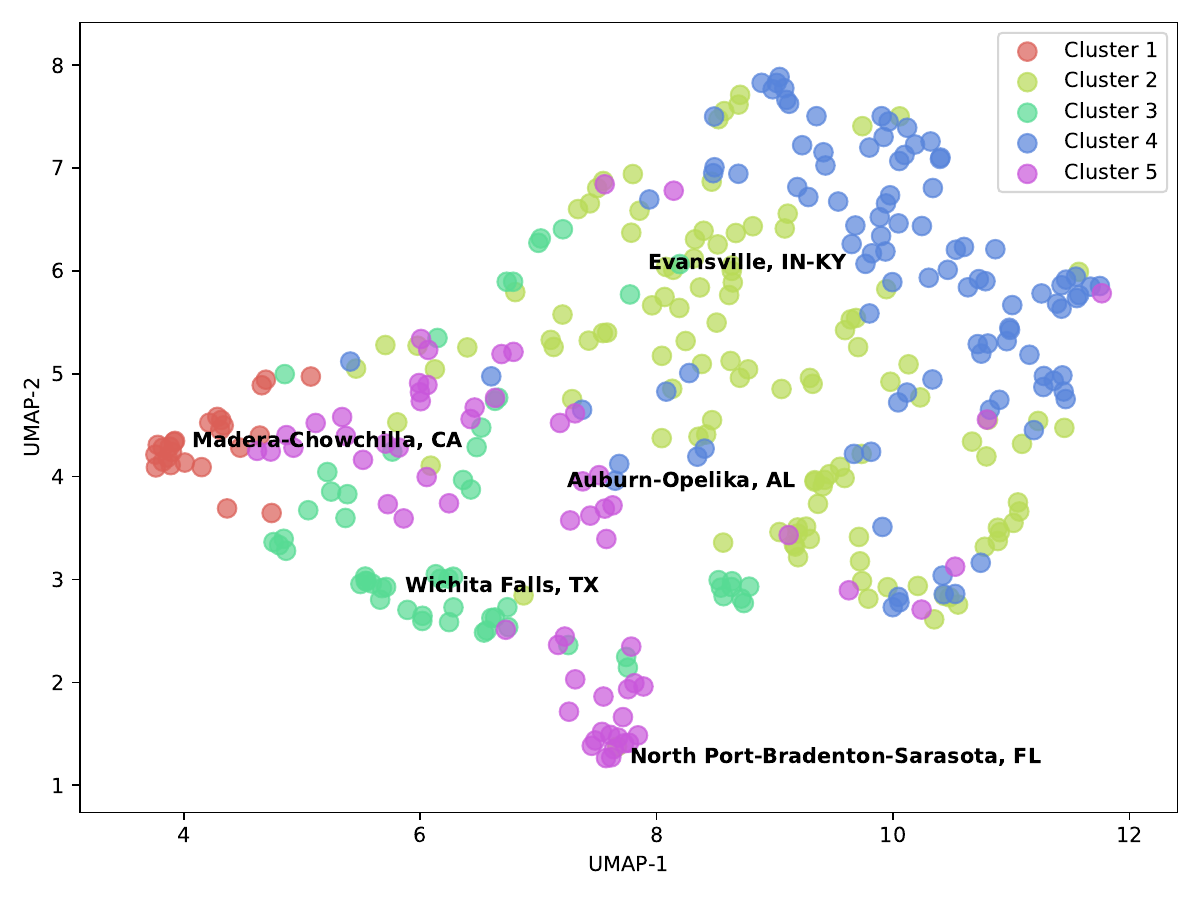}
\caption{LLM embeddings}
\label{app-fig:app_clusters}
\end{subfigure}
\caption[Cluster assignments in UMAP space, by embedding method]{Document embeddings projected to two dimensions with UMAP \citep{mcinnes2018umap}. Labels corresponds to the CBSA closest to that cluster's centroid in the full embedding space.}
\label{app-fig:app_clusters_all}
\end{figure}

\clearpage

\subsection{Sensitivity to the Number of Clusters}\label{app-app:app_kmeans_sensitivity}

Section~\ref{sec:application} fixes $L=5$. Figure~\ref{app-fig:sensitivity_L} sweeps $L \in \{2,3,4,5,6,8,10\}$ for each method, recomputing the permutation floor at every $L$ so the ratios are comparable down a column.

First, every method stays well above its floor at every $L$, so the finding that these
groupings separate employment dynamics more sharply than chance is not an artifact of the cluster count.

Second, a text method leads at every $L$: LLM embeddings at $L=2$ and $L=6$, SVD-$\beta$ at $L=3$, $4$, and $8$, corpus CBOW at $L=6$, corpus SGNS at $L=5$ and $10$. Residual $k$-means is below all five text methods at every $L$ except $L=8$, and below all three corpus-based methods across the grid. Observables $k$-means tends to fall between residual $k$-means and the text-based methods.

We conclude that the ordering \emph{among} the three corpus-based text methods is not stable across the grid. But the finding that
text-based groupings separate employment dynamics better than our benchmarks does not depend on the cluster count.

\begin{figure}[htb!]
\centering
\includegraphics[width=0.86\textwidth]{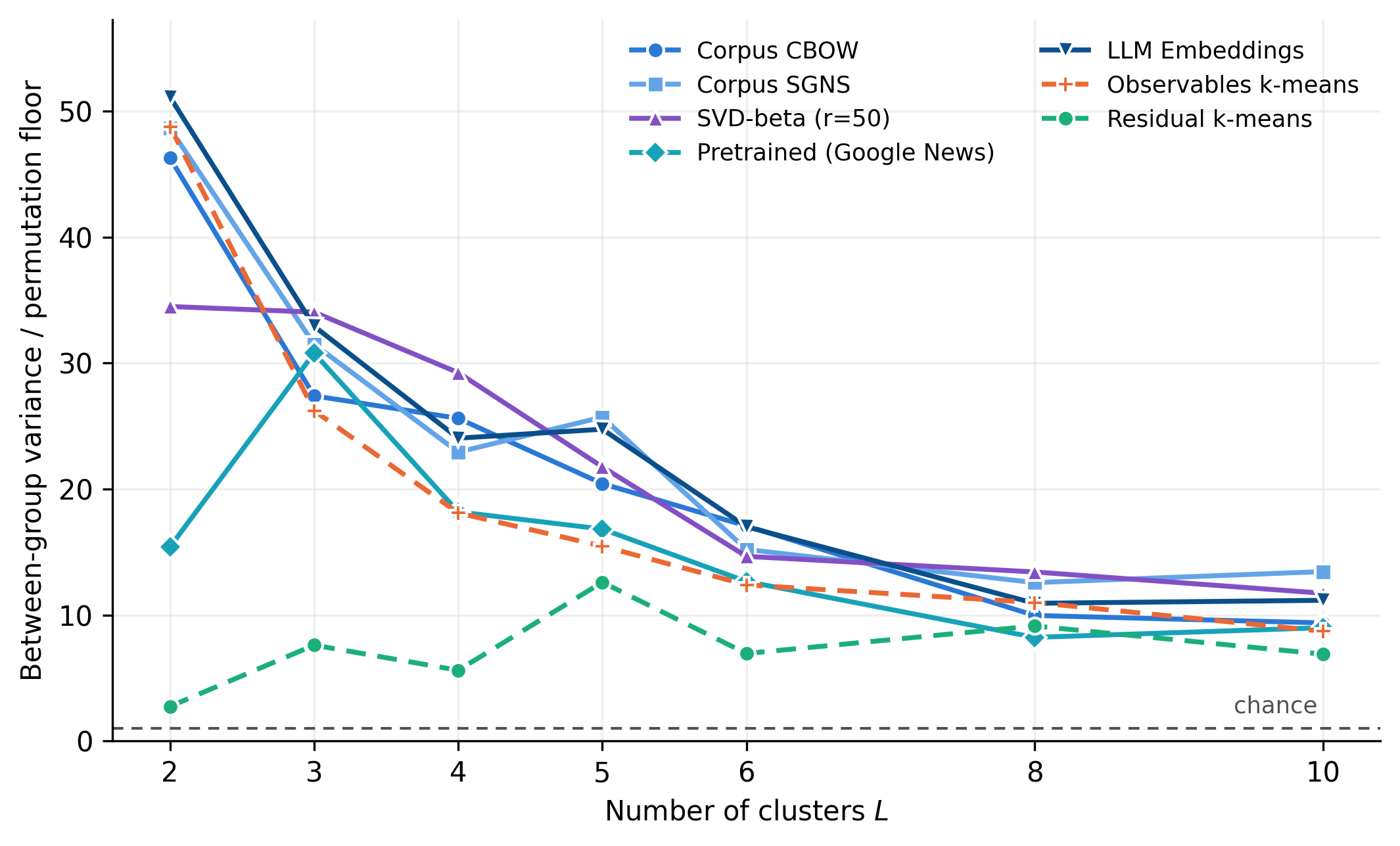}
\caption[Sensitivity to the number of clusters]{Between-group variance of unit-specific AR(2) impulse responses relative to the size-matched permutation floor, against the number of clusters $L$. The floor is recomputed at every $L$, so ratios are comparable along each line. Dashed lines represent the non-text benchmarks.}
\label{app-fig:sensitivity_L}
\end{figure}

\subsection{Sensitivity to the Embedding Dimension}\label{app-app:app_d_sweep}

Remark~\ref{rem:dimension_slack} shows that any $r \geq K-1$ recovers the same geometry in population, so the dimension needs only to be chosen large enough---which matters because $K$ is not observable for this corpus. Figure~\ref{app-fig:sensitivity_dimension} sweeps $r \in \{10,25,50,100,200\}$ for the three corpus-based methods at the full-document window and $L=5$.

\begin{figure}[htb!]
\centering
\includegraphics[width=\textwidth]{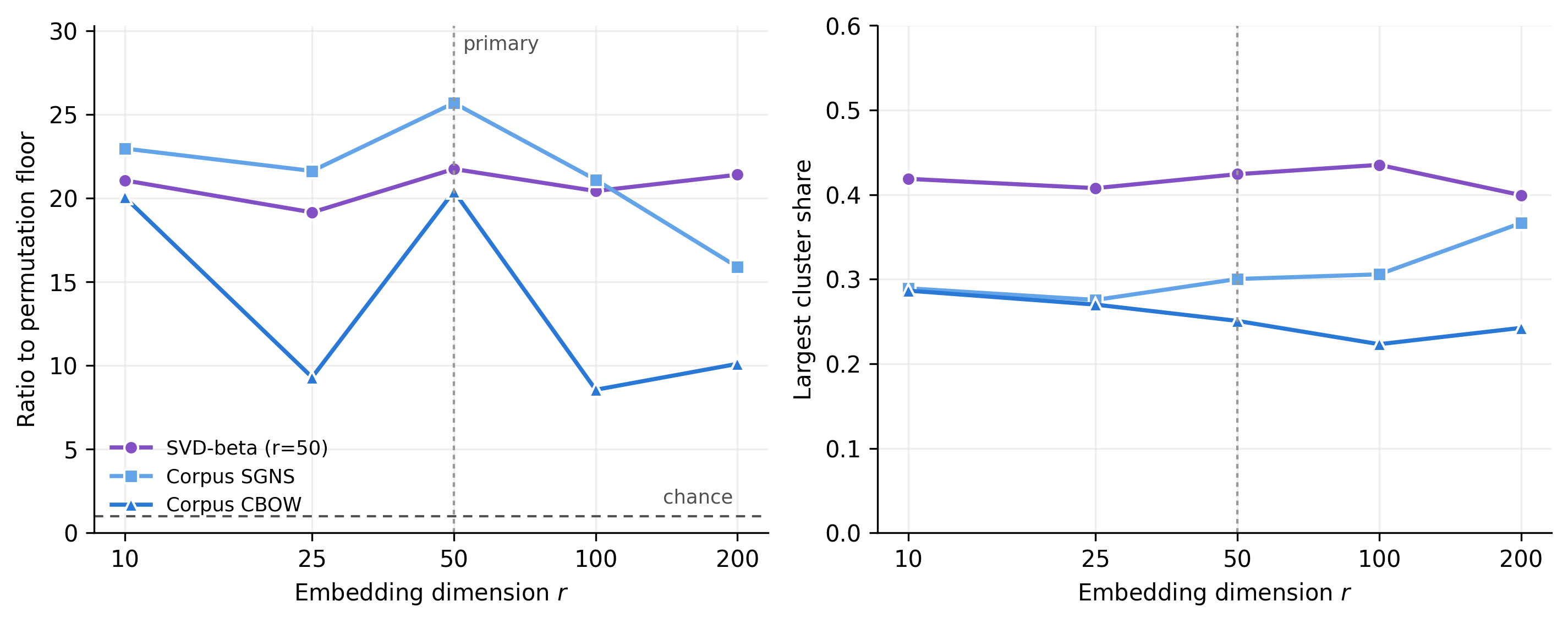}
\caption[Sensitivity to the embedding dimension]{Sensitivity to the embedding dimension $r$ at $L=5$, $p=2$ and the full-document window, for the three corpus-based methods. Left: between-group variance relative to the size-matched permutation floor. Right: share of CBSAs in the largest cluster. SVD-$\beta$ and corpus SGNS are flat from $r=10$ to $r=200$, which is what justifies the $r=50$ primary specification (dotted) without identifying $K$. Corpus CBOW alternates between two families of near-equivalent $k$-means partitions.}
\label{app-fig:sensitivity_dimension}
\end{figure}

SVD-$\beta$ and corpus SGNS are flat across the entire range: neither ratio trends in $r$, and their partitions stay comparably balanced throughout. This is what Remark~\ref{rem:dimension_slack} predicts once $\widehat R$ is well estimated---coordinates beyond $K-1$ contribute nothing in population and, at the full-document window, little in sample either. It justifies $r=50$ without an estimate of $K$: any choice across this range would serve.

Corpus CBOW is not flat, but in this case it reflects the instability of k-means. For the CBOW embeddings, the algorithm is choosing among near-equivalent optima, and which one it lands on moves the between-group ratio by a factor of two.

\subsection{Sensitivity to the Autoregressive Order}\label{app-app:app_p_sweep}

The AR order $p$ is a nuisance parameter of the same kind as the cluster number $L$, context window $J$ and embedding dimension $r$. In this section I explore the choice of $p$.

Figure~\ref{app-fig:sensitivity_p_noise} shows the share of cross-city IRF dispersion attributable to estimation error. 
To compute this, I hold the DGP fixed at a common-slope $AR(8)$ whose coefficients are the pooled FE estimates on the same panel, and whose innovations are each unit's own $AR(8)$ residuals resampled $i.i.d.$ The fitted model is then $AR(p)$, unit by unit, on the simulated paths, averaged over 20 replications. Cross-city IRF dispersion is defined as the mean squared deviation from the cross-unit mean response, horizon by horizon, and Figure~\ref{app-fig:sensitivity_p_noise} depicts the ratio between the simulated dipersion (under homogeneity) and the dispersion observed in the data.
\begin{figure}[htb!]
\centering
\includegraphics[width=0.86\textwidth]{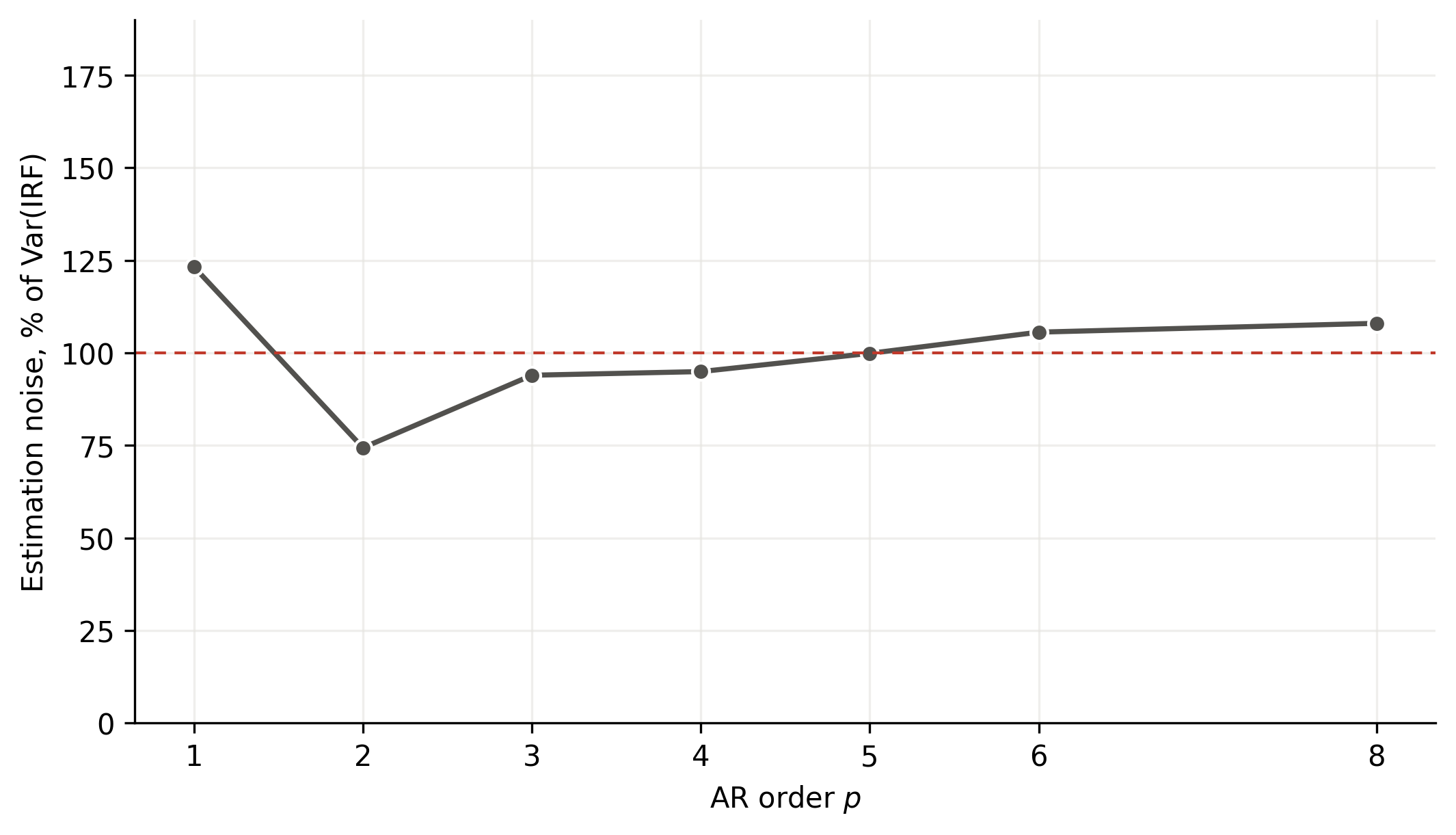}
\caption[Estimation noise against the autoregressive order]{Share of total cross-city IRF variance attributable to estimation error, against the autoregressive order $p$. The benchmark holds the DGP fixed at a common-slope AR(8) and fits AR($p$) to the simulated paths. The dashed line marks the point at which the simulated data generates as much dispersion as the observed data.}
\label{app-fig:sensitivity_p_noise}
\end{figure}

On that benchmark the share is minimised at $p=2$ (74\%) and is at least 94\% at every other order, reaching 100\% at $p=5$ and exceeding it from $p=6$. One note of caution: We redraw innovations independently across units; employment shocks are not (and likely positively correlated). That would bias all depicted ratios up and may be why the share exceeds 100\% from $p = 5$ on.

Figure~\ref{app-fig:sensitivity_p_ratio} depicts the ratio from Table \ref{tab:irf_ratio_main} for different values of $p$. Seven of the eight approaches peak at $p=2$, though they all remain above 1 throughout.
\begin{figure}[htb!]
\centering
\includegraphics[width=0.86\textwidth]{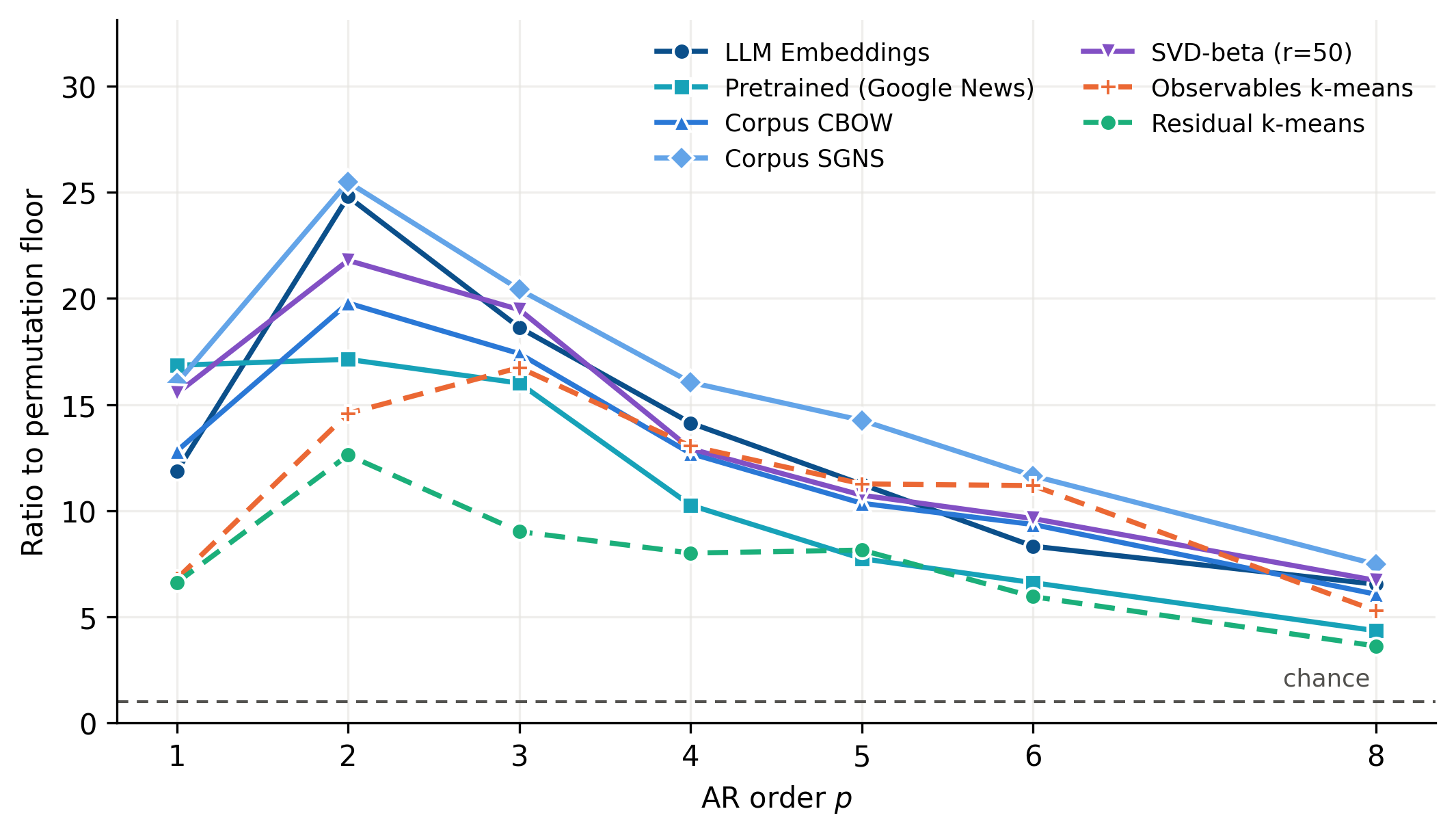}
\caption[Between-group variance against the autoregressive order]{Between-group variance of the unit-specific IRFs relative to the size-matched permutation floor, against the autoregressive order $p$, at $L=5$. Chance is at $1\times$. Dashed lines represent the two non-text benchmarks.}
\label{app-fig:sensitivity_p_ratio}
\end{figure}

We thus read $p=2$ as the best low-dimensional \emph{projection} of the dynamics for cross-sectional comparison rather than as a correctly specified model: the pooled panel clearly wants more lags, but fitting them spends per-unit degrees of freedom that a $T=51$ series does not have, and a parsimonious fit concentrates whatever genuine heterogeneity exists into few well-estimated coefficients instead of dispersing it across many noisy ones. 

One caveat applies to both figures. The Nickell bias in $\hat\rho$ is $O(1/T)$, but the horizon-$h$ impulse response is a compounding function of it, so the bias in the object plotted here is of order $h/T$; at $h=10$ and $T=51$ that is not negligible. Correcting it---by split-panel jackknife or an analytic adjustment---would plausibly shift the levels in both at the expense of even noisier estimates. We leave this for future work.

\subsection{Sensitivity to the Context Window}\label{app-app:app_context_window}

Section~\ref{sec:application} uses the full document as context for both corpus-trained arms and SVD-$\beta$, on the grounds that positions within a document are exchangeable under the model, so $R$ does not depend on the window in population. Figure~\ref{app-fig:sensitivity_context_window} tests that choice against $J \in \{5,10,20,50\}$, holding $L=5$.

Neither corpus CBOW nor corpus SGNS exhibit a clear trend. SVD-$\beta$ does trend. It holds $21.8\times$ at the full document and $21.7\times$ at $J=50$, but falls to $12.1\times$ at $J=5$---roughly half. This is in line with our findings from our numerical exercise: 
With $V$ in the low thousands, $\widehat R$ has on the order of six million cells, and the ordered word--context pairs the 363 descriptions supply fall from roughly eighty million at the full-document window to a few million at $J=5$. With that sample size the estimated SVD-$\beta$ embeddings start to become somewhat noisier, reflecting the difficulty in directly matrix-factorizing the noisy sample co-occurrence structure.

One implementation detail matters for reading the two Word2Vec lines: gensim's window parameter is a maximum, not a fixed width. For each token the effective window is drawn uniformly between one and that maximum, which is equivalent to weighting contexts by their distance from the focus word
\citep{levy2015improving}. Any direct comparison to SVD-$\beta$ based on Figure~\ref{app-fig:sensitivity_context_window} is thus to be treated with caution.

\begin{figure}[htb!]
\centering
\includegraphics[width=0.86\textwidth]{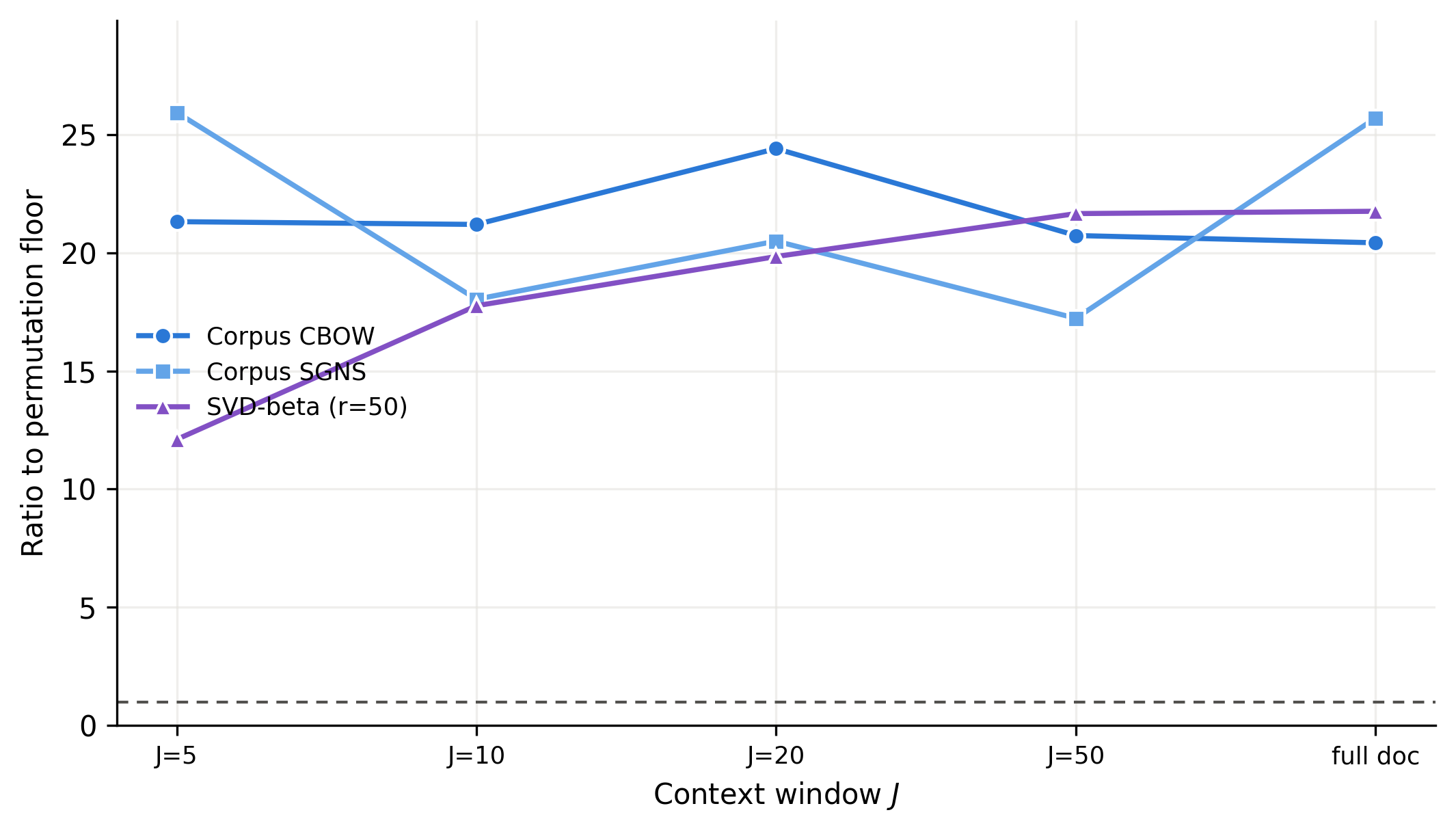}
\caption[Sensitivity to the context window]{Sensitivity to the context window $J$ at $L=5$ and $p=2$. Between-group variance of the unit-specific IRFs relative to the size-matched permutation floor, with chance at $1\times$. The SVD-$\beta$ ratio falls by about half as the window narrows to $J=5$; neither corpus-trained arm trends either way.}
\label{app-fig:sensitivity_context_window}
\end{figure}

\subsection{Dating the Observables: Covariate Vintage}\label{app-app:app_covariate_vintage}

The covariate set consists of ten variables, all at the CBSA level. Four are employment shares---manufacturing, total government, military, and federal civilian. Three are population shares---foreign-born, Black, and Hispanic. Two are education shares---the share with no high-school diploma and the share with a four-year college degree. The tenth covariate is population per square mile. We standardize each to mean zero and unit variance and then apply $k$-means with $L=5$ (\texttt{k-means++}, $500$ restarts, seed $42$).

The observables benchmark in Section~\ref{sec:ar_panel} clusters CBSAs on covariates measured at a single Census vintage. We take 1970 as the primary specification: it is the only vintage that is effectively predetermined, being measured one year into a fifty-one-year panel. In this section, I sweep the vintage to measure the contamination induced by using concurrent Census vintages. Table~\ref{app-tab:sensitivity_covariate_vintage} reports the sweep.Population density enters at 1970 throughout, being the only vintage available. The other nine (industry-mix, demographic and education shares) are avilable at 1970, 1980, 1990 and 2000 with complete coverage of the 363 CBSAs.
\begin{table}[htb!]
\centering
\small
\begin{tabular}{lrrrr}
\hline
Vintage & Years elapsed & Var(Between) & Ratio & Max share \\
\hline
1970 & 1 & 0.0209 & 15.5$\times$ & 36\% \\
1980 & 11 & 0.0238 & 17.2$\times$ & 46\% \\
1990 & 21 & 0.0243 & 17.4$\times$ & 44\% \\
2000 & 31 & 0.0248 & 18.6$\times$ & 41\% \\
\hline
\end{tabular}
\caption[Observables benchmark by covariate vintage]{Observables benchmark at each Census vintage, $L=5$ and $p=2$. ``Years elapsed'' is the number of years of the 1969--2019 outcome window already observed at the measurement date. ``Ratio'' is between-group variance over the size-matched permutation floor. All nine covariates are complete at every vintage; population density enters at 1970 throughout, being the only vintage available.}
\label{app-tab:sensitivity_covariate_vintage}
\end{table}

The ratio rises monotonically in the vintage, from $\ratioObsSeventy\times$ the permutation floor at 1970 to $\ratioObsMillennium\times$ at 2000---a gain of about $\vintageGain\%$ from dating the same nine covariates thirty years later.

We read this as a measurement of look-ahead: later measurement is itself partly an outcome of the dynamics being explained. On this metric the advantage is worth roughly a fifth of the benchmark's apparent performance. That matters beyond the choice of covariates, because the same objection applies to the text. The descriptions span 1979--2019 and so are concurrent with the outcome window rather than predetermined, a limitation Section~\ref{sec:application} states but cannot quantify. The vintage sweep gives an order of magnitude for what such concurrency can buy---roughly $\vintageGain\%$ here---which is well short of the margin by which the text-based groupings exceed the benchmark.

\subsection{Interpreting the Clusters}\label{app-app:clusters}

Table~\ref{app-tab:cluster_assignments} reports the cluster assignment for each of the 363 CBSAs under seven clustering approaches. Clusters are numbered 1--5 within each method; descriptions for the text-based methods follow below. Figure~\ref{app-fig:cluster_maps} shows the geographic distribution of cluster assignments across the continental United States.

\begin{figure}[htb!]
\centering
\includegraphics[width=.9\textwidth]{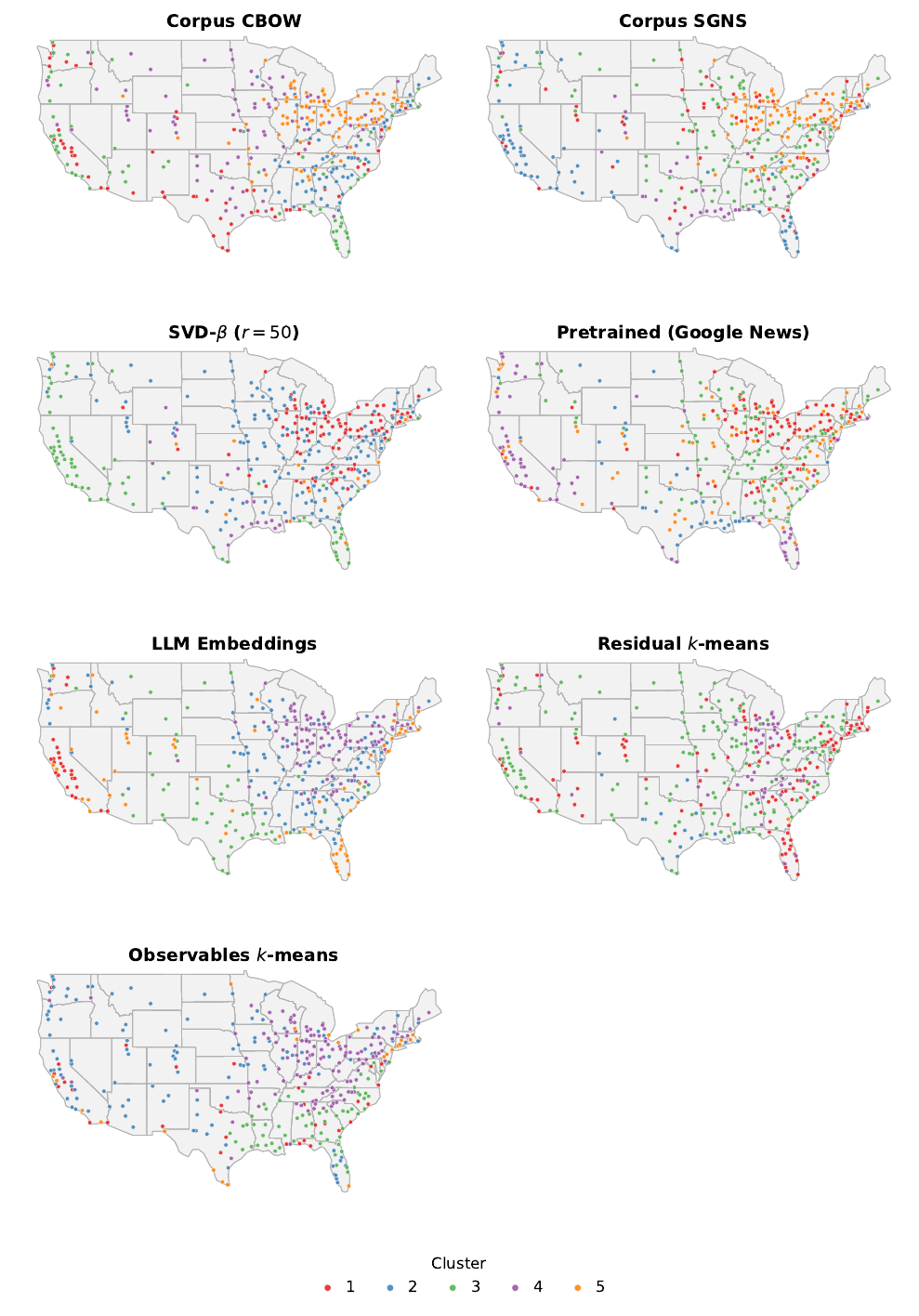}
\caption[Geographic distribution of cluster assignments]{Geographic distribution of cluster assignments for 363 CBSAs. Each dot is a CBSA, colored by its cluster. Cluster numbers are assigned within each method, so a color is comparable across CBSAs in one panel but not across panels.}
\label{app-fig:cluster_maps}
\end{figure}

{\small
\begin{longtable}{llccccccc}
\caption[Cluster assignments for all CBSAs, by method]{Cluster assignments for all 363 CBSAs: corpus CBOW (CB), corpus SGNS (CS), closed-form SVD-$\beta$ at $r=50$ (SB), pretrained Google News vectors (PW), LLM embedding (LE, $k$-means on OpenAI \texttt{text-embedding-3-large}), residual $k$-means (RK, $k$-means on pooled AR(2) residuals from Section~\ref{sec:ar_panel}), and observables $k$-means (OB, on the 1970 covariates of Section~\ref{sec:ar_panel}).Numeric labels correspond to the cluster descriptions below the table. }\label{app-tab:cluster_assignments} \\
\toprule
 & & \multicolumn{7}{c}{\textbf{Clustering Method}} \\
\cmidrule(lr){3-9}
\textbf{CBSA} & \textbf{State} & \textbf{CB} & \textbf{CS} & \textbf{SB} & \textbf{PW} & \textbf{LE} & \textbf{RK} & \textbf{OB} \\
\midrule
\endfirsthead
\multicolumn{9}{l}{\small\itshape Table~\ref{app-tab:cluster_assignments} continued} \\
\toprule
\textbf{CBSA} & \textbf{State} & \textbf{CB} & \textbf{CS} & \textbf{SB} & \textbf{PW} & \textbf{LE} & \textbf{RK} & \textbf{OB} \\
\midrule
\endhead
\midrule
\multicolumn{9}{r}{\textit{Continued on next page}} \\
\endfoot
\bottomrule
\endlastfoot
Abilene & TX & 4 & 3 & 2 & 3 & 3 & 2 & 2 \\
Akron & OH & 5 & 5 & 1 & 1 & 4 & 1 & 4 \\
Albany & GA & 2 & 3 & 2 & 3 & 2 & 1 & 3 \\
Albany-Schenectady-Troy & NY & 5 & 3 & 2 & 5 & 5 & 3 & 2 \\
Albuquerque & NM & 3 & 1 & 2 & 5 & 3 & 3 & 2 \\
Alexandria & LA & 4 & 3 & 2 & 5 & 3 & 3 & 3 \\
Allentown-Bethlehem-Easton & PA-NJ & 5 & 5 & 1 & 1 & 4 & 1 & 4 \\
Altoona & PA & 5 & 5 & 1 & 1 & 4 & 3 & 4 \\
Amarillo & TX & 4 & 3 & 2 & 3 & 3 & 3 & 2 \\
Ames & IA & 4 & 1 & 2 & 5 & 2 & 3 & 2 \\
Anderson & IN & 5 & 5 & 1 & 1 & 4 & 4 & 4 \\
Anderson & SC & 2 & 3 & 1 & 1 & 4 & 4 & 4 \\
Ann Arbor & MI & 4 & 1 & 2 & 5 & 2 & 4 & 2 \\
Anniston-Oxford & AL & 5 & 4 & 5 & 1 & 2 & 3 & 3 \\
Appleton & WI & 5 & 3 & 1 & 1 & 4 & 1 & 4 \\
Asheville & NC & 3 & 2 & 3 & 3 & 2 & 4 & 4 \\
Athens-Clarke County & GA & 2 & 3 & 2 & 5 & 2 & 1 & 3 \\
Atlanta-Sandy Springs-Marietta & GA & 2 & 1 & 2 & 3 & 5 & 1 & 3 \\
Atlantic City-Hammonton & NJ & 5 & 5 & 1 & 4 & 5 & 1 & 3 \\
Auburn-Opelika & AL & 2 & 3 & 2 & 3 & 2 & 1 & 3 \\
Augusta-Richmond County & GA-SC & 2 & 3 & 2 & 3 & 2 & 3 & 3 \\
Austin-Round Rock-San Marcos & TX & 4 & 1 & 2 & 5 & 5 & 1 & 2 \\
Bakersfield-Delano & CA & 1 & 2 & 3 & 4 & 1 & 3 & 2 \\
Baltimore-Towson & MD & 2 & 3 & 2 & 5 & 2 & 3 & 3 \\
Bangor & ME & 2 & 3 & 2 & 3 & 2 & 3 & 4 \\
Barnstable Town & MA & 3 & 2 & 3 & 4 & 5 & 1 & 2 \\
Baton Rouge & LA & 1 & 4 & 4 & 2 & 3 & 3 & 3 \\
Battle Creek & MI & 5 & 5 & 1 & 1 & 4 & 4 & 4 \\
Bay City & MI & 5 & 5 & 1 & 1 & 4 & 4 & 4 \\
Beaumont-Port Arthur & TX & 1 & 4 & 4 & 2 & 3 & 2 & 3 \\
Bellingham & WA & 3 & 2 & 3 & 4 & 2 & 3 & 2 \\
Bend & OR & 3 & 2 & 3 & 4 & 5 & 4 & 2 \\
Billings & MT & 4 & 3 & 2 & 2 & 3 & 3 & 2 \\
Binghamton & NY & 5 & 5 & 1 & 1 & 4 & 3 & 4 \\
Birmingham-Hoover & AL & 2 & 3 & 2 & 1 & 2 & 3 & 3 \\
Bismarck & ND & 4 & 3 & 2 & 2 & 3 & 3 & 2 \\
Blacksburg-Christiansburg-Radford & VA & 4 & 3 & 2 & 5 & 2 & 1 & 4 \\
Bloomington & IN & 4 & 1 & 2 & 5 & 2 & 3 & 2 \\
Bloomington-Normal & IL & 5 & 1 & 2 & 5 & 2 & 3 & 2 \\
Boise City-Nampa & ID & 4 & 1 & 2 & 3 & 5 & 3 & 2 \\
Boston-Cambridge-Quincy & MA-NH & 3 & 1 & 2 & 5 & 5 & 1 & 5 \\
Boulder & CO & 4 & 1 & 2 & 5 & 5 & 1 & 2 \\
Bowling Green & KY & 4 & 3 & 2 & 3 & 2 & 4 & 4 \\
Bremerton-Silverdale & WA & 1 & 4 & 5 & 5 & 3 & 3 & 1 \\
Bridgeport-Stamford-Norwalk & CT & 2 & 5 & 1 & 3 & 5 & 1 & 5 \\
Brownsville-Harlingen & TX & 1 & 2 & 3 & 4 & 3 & 3 & 5 \\
Brunswick & GA & 2 & 3 & 3 & 3 & 2 & 1 & 3 \\
Buffalo-Niagara Falls & NY & 5 & 5 & 1 & 1 & 4 & 3 & 5 \\
Burlington & NC & 2 & 5 & 1 & 3 & 2 & 1 & 4 \\
Burlington-South Burlington & VT & 4 & 1 & 2 & 5 & 5 & 1 & 2 \\
Canton-Massillon & OH & 5 & 5 & 1 & 1 & 4 & 1 & 4 \\
Cape Coral-Fort Myers & FL & 3 & 2 & 3 & 4 & 5 & 1 & 2 \\
Cape Girardeau-Jackson & MO-IL & 4 & 3 & 2 & 3 & 2 & 3 & 4 \\
Carson City & NV & 4 & 3 & 2 & 5 & 2 & 4 & 2 \\
Casper & WY & 4 & 4 & 4 & 2 & 3 & 2 & 2 \\
Cedar Rapids & IA & 4 & 3 & 2 & 3 & 2 & 3 & 4 \\
Champaign-Urbana & IL & 4 & 1 & 2 & 5 & 2 & 3 & 2 \\
Charleston & WV & 5 & 5 & 1 & 1 & 4 & 3 & 4 \\
Charleston-North Charleston-Summerville & SC & 3 & 1 & 2 & 3 & 5 & 3 & 1 \\
Charlotte-Gastonia-Rock Hill & NC-SC & 2 & 3 & 2 & 3 & 5 & 1 & 4 \\
Charlottesville & VA & 3 & 1 & 2 & 5 & 2 & 1 & 3 \\
Chattanooga & TN-GA & 2 & 5 & 1 & 3 & 4 & 4 & 4 \\
Cheyenne & WY & 1 & 4 & 2 & 2 & 3 & 3 & 2 \\
Chicago-Joliet-Naperville & IL-IN-WI & 2 & 5 & 1 & 3 & 4 & 1 & 5 \\
Chico & CA & 4 & 2 & 3 & 4 & 1 & 3 & 2 \\
Cincinnati-Middletown & OH-KY-IN & 2 & 3 & 2 & 3 & 4 & 1 & 4 \\
Clarksville & TN-KY & 1 & 4 & 5 & 3 & 3 & 4 & 1 \\
Cleveland & TN & 2 & 3 & 2 & 3 & 2 & 4 & 4 \\
Cleveland-Elyria-Mentor & OH & 5 & 5 & 1 & 1 & 4 & 1 & 5 \\
Coeur d'Alene & ID & 3 & 2 & 3 & 4 & 5 & 4 & 2 \\
College Station-Bryan & TX & 4 & 1 & 2 & 5 & 3 & 2 & 3 \\
Colorado Springs & CO & 4 & 4 & 5 & 5 & 3 & 1 & 1 \\
Columbia & MO & 4 & 3 & 2 & 5 & 2 & 3 & 2 \\
Columbia & SC & 2 & 3 & 2 & 5 & 2 & 1 & 3 \\
Columbus & GA-AL & 2 & 3 & 2 & 3 & 2 & 3 & 1 \\
Columbus & IN & 4 & 5 & 1 & 1 & 4 & 4 & 4 \\
Columbus & OH & 4 & 3 & 2 & 3 & 2 & 1 & 2 \\
Corpus Christi & TX & 1 & 4 & 4 & 2 & 3 & 2 & 2 \\
Corvallis & OR & 4 & 1 & 2 & 5 & 2 & 3 & 2 \\
Crestview-Fort Walton Beach-Destin & FL & 3 & 2 & 3 & 4 & 5 & 3 & 1 \\
Cumberland & MD-WV & 5 & 5 & 1 & 3 & 4 & 3 & 4 \\
Dallas-Fort Worth-Arlington & TX & 4 & 1 & 2 & 3 & 5 & 1 & 4 \\
Dalton & GA & 5 & 5 & 1 & 1 & 4 & 4 & 4 \\
Danville & IL & 5 & 5 & 1 & 1 & 4 & 4 & 4 \\
Danville & VA & 5 & 5 & 1 & 1 & 4 & 4 & 3 \\
Davenport-Moline-Rock Island & IA-IL & 5 & 5 & 1 & 1 & 4 & 3 & 4 \\
Dayton & OH & 5 & 5 & 1 & 1 & 4 & 1 & 4 \\
Decatur & AL & 5 & 5 & 1 & 1 & 4 & 1 & 4 \\
Decatur & IL & 5 & 5 & 1 & 1 & 4 & 3 & 4 \\
Deltona-Daytona Beach-Ormond Beach & FL & 3 & 2 & 3 & 4 & 5 & 1 & 2 \\
Denver-Aurora-Broomfield & CO & 4 & 1 & 4 & 2 & 5 & 1 & 2 \\
Des Moines-West Des Moines & IA & 4 & 1 & 2 & 3 & 5 & 3 & 2 \\
Detroit-Warren-Livonia & MI & 5 & 5 & 1 & 1 & 4 & 4 & 5 \\
Dothan & AL & 2 & 3 & 2 & 3 & 2 & 4 & 3 \\
Dover & DE & 1 & 3 & 2 & 5 & 2 & 3 & 1 \\
Dubuque & IA & 4 & 3 & 2 & 1 & 4 & 3 & 4 \\
Duluth & MN-WI & 5 & 5 & 1 & 1 & 4 & 3 & 2 \\
Durham-Chapel Hill & NC & 4 & 1 & 2 & 5 & 2 & 3 & 3 \\
Eau Claire & WI & 4 & 3 & 2 & 3 & 2 & 3 & 4 \\
El Centro & CA & 1 & 2 & 3 & 4 & 1 & 3 & 5 \\
Elizabethtown & KY & 1 & 3 & 2 & 3 & 2 & 3 & 1 \\
Elkhart-Goshen & IN & 4 & 2 & 2 & 1 & 4 & 4 & 4 \\
Elmira & NY & 5 & 5 & 1 & 1 & 4 & 1 & 4 \\
El Paso & TX & 1 & 2 & 2 & 1 & 3 & 3 & 5 \\
Erie & PA & 5 & 5 & 1 & 1 & 4 & 3 & 4 \\
Eugene-Springfield & OR & 4 & 3 & 3 & 4 & 2 & 4 & 2 \\
Evansville & IN-KY & 2 & 5 & 1 & 3 & 4 & 3 & 4 \\
Fargo & ND-MN & 4 & 1 & 2 & 3 & 2 & 3 & 2 \\
Farmington & NM & 1 & 4 & 4 & 2 & 3 & 2 & 2 \\
Fayetteville & NC & 1 & 4 & 5 & 5 & 3 & 3 & 1 \\
Fayetteville-Springdale-Rogers & AR-MO & 4 & 1 & 2 & 3 & 2 & 1 & 4 \\
Flagstaff & AZ & 3 & 2 & 3 & 4 & 2 & 1 & 2 \\
Flint & MI & 5 & 5 & 1 & 1 & 4 & 4 & 4 \\
Florence & SC & 2 & 3 & 2 & 3 & 2 & 1 & 3 \\
Florence-Muscle Shoals & AL & 5 & 3 & 1 & 3 & 2 & 3 & 4 \\
Fond du Lac & WI & 4 & 3 & 2 & 1 & 4 & 3 & 4 \\
Fort Collins-Loveland & CO & 4 & 1 & 2 & 3 & 5 & 1 & 2 \\
Fort Smith & AR-OK & 5 & 5 & 1 & 1 & 4 & 4 & 4 \\
Fort Wayne & IN & 5 & 5 & 1 & 3 & 4 & 4 & 4 \\
Fresno & CA & 1 & 2 & 3 & 4 & 1 & 3 & 2 \\
Gadsden & AL & 5 & 5 & 1 & 1 & 4 & 4 & 4 \\
Gainesville & FL & 4 & 1 & 2 & 5 & 2 & 1 & 2 \\
Gainesville & GA & 3 & 2 & 3 & 3 & 2 & 1 & 4 \\
Glens Falls & NY & 5 & 3 & 2 & 3 & 4 & 3 & 4 \\
Goldsboro & NC & 1 & 3 & 2 & 4 & 3 & 3 & 3 \\
Grand Forks & ND-MN & 4 & 3 & 2 & 3 & 3 & 3 & 5 \\
Grand Junction & CO & 4 & 2 & 4 & 2 & 3 & 2 & 2 \\
Grand Rapids-Wyoming & MI & 4 & 1 & 2 & 3 & 4 & 4 & 4 \\
Great Falls & MT & 4 & 3 & 2 & 2 & 3 & 3 & 2 \\
Greeley & CO & 1 & 2 & 4 & 2 & 3 & 1 & 2 \\
Green Bay & WI & 5 & 3 & 2 & 3 & 4 & 3 & 4 \\
Greensboro-High Point & NC & 2 & 5 & 1 & 1 & 4 & 4 & 4 \\
Greenville & NC & 2 & 3 & 2 & 3 & 2 & 3 & 3 \\
Greenville-Mauldin-Easley & SC & 2 & 5 & 1 & 3 & 4 & 1 & 4 \\
Gulfport-Biloxi & MS & 1 & 4 & 4 & 2 & 5 & 3 & 1 \\
Hagerstown-Martinsburg & MD-WV & 5 & 3 & 2 & 3 & 2 & 3 & 4 \\
Hanford-Corcoran & CA & 1 & 2 & 3 & 4 & 1 & 3 & 1 \\
Harrisburg-Carlisle & PA & 4 & 3 & 2 & 3 & 2 & 3 & 2 \\
Harrisonburg & VA & 4 & 3 & 2 & 3 & 2 & 3 & 4 \\
Hartford-West Hartford-East Hartford & CT & 5 & 1 & 2 & 1 & 5 & 1 & 5 \\
Hattiesburg & MS & 2 & 3 & 2 & 3 & 2 & 3 & 3 \\
Hickory-Lenoir-Morganton & NC & 5 & 5 & 1 & 1 & 4 & 4 & 4 \\
Hinesville-Fort Stewart & GA & 1 & 4 & 5 & 5 & 3 & 5 & 1 \\
Holland-Grand Haven & MI & 4 & 1 & 2 & 1 & 4 & 4 & 4 \\
Hot Springs & AR & 3 & 2 & 3 & 4 & 2 & 4 & 4 \\
Houma-Bayou Cane-Thibodaux & LA & 1 & 4 & 4 & 2 & 3 & 2 & 3 \\
Houston-Sugar Land-Baytown & TX & 4 & 4 & 4 & 2 & 3 & 2 & 3 \\
Huntington-Ashland & WV-KY-OH & 5 & 5 & 1 & 1 & 4 & 3 & 4 \\
Huntsville & AL & 2 & 1 & 5 & 5 & 2 & 1 & 1 \\
Idaho Falls & ID & 4 & 3 & 2 & 3 & 3 & 3 & 2 \\
Indianapolis-Carmel & IN & 2 & 5 & 2 & 3 & 4 & 1 & 4 \\
Iowa City & IA & 4 & 1 & 2 & 5 & 2 & 3 & 2 \\
Ithaca & NY & 4 & 1 & 2 & 5 & 2 & 3 & 2 \\
Jackson & MI & 5 & 5 & 1 & 1 & 4 & 4 & 4 \\
Jackson & MS & 2 & 3 & 2 & 5 & 2 & 1 & 3 \\
Jackson & TN & 2 & 3 & 2 & 3 & 2 & 4 & 3 \\
Jacksonville & FL & 2 & 1 & 2 & 3 & 2 & 1 & 3 \\
Jacksonville & NC & 1 & 4 & 5 & 5 & 3 & 3 & 1 \\
Janesville & WI & 5 & 5 & 1 & 1 & 4 & 4 & 4 \\
Jefferson City & MO & 4 & 3 & 2 & 5 & 2 & 3 & 2 \\
Johnson City & TN & 2 & 3 & 2 & 3 & 2 & 4 & 4 \\
Johnstown & PA & 5 & 5 & 1 & 1 & 4 & 3 & 4 \\
Jonesboro & AR & 2 & 3 & 2 & 3 & 2 & 3 & 4 \\
Joplin & MO & 5 & 3 & 2 & 3 & 2 & 4 & 4 \\
Kalamazoo-Portage & MI & 5 & 1 & 2 & 3 & 2 & 3 & 4 \\
Kankakee-Bradley & IL & 5 & 5 & 1 & 3 & 4 & 3 & 4 \\
Kansas City & MO-KS & 4 & 3 & 2 & 3 & 2 & 1 & 2 \\
Kennewick-Pasco-Richland & WA & 1 & 2 & 2 & 4 & 3 & 3 & 2 \\
Killeen-Temple-Fort Hood & TX & 1 & 4 & 5 & 5 & 3 & 3 & 1 \\
Kingsport-Bristol-Bristol & TN-VA & 5 & 5 & 1 & 3 & 4 & 3 & 4 \\
Kingston & NY & 3 & 5 & 1 & 1 & 4 & 3 & 4 \\
Knoxville & TN & 2 & 3 & 2 & 3 & 2 & 3 & 4 \\
Kokomo & IN & 5 & 5 & 1 & 1 & 4 & 4 & 4 \\
La Crosse & WI-MN & 4 & 3 & 2 & 3 & 2 & 3 & 4 \\
Lafayette & IN & 4 & 1 & 2 & 1 & 4 & 3 & 2 \\
Lafayette & LA & 4 & 4 & 4 & 2 & 3 & 2 & 3 \\
Lake Charles & LA & 1 & 4 & 4 & 2 & 3 & 3 & 3 \\
Lake Havasu City-Kingman & AZ & 3 & 2 & 3 & 4 & 5 & 1 & 2 \\
Lakeland-Winter Haven & FL & 3 & 2 & 3 & 4 & 5 & 1 & 3 \\
Lancaster & PA & 4 & 1 & 2 & 3 & 2 & 1 & 4 \\
Lansing-East Lansing & MI & 4 & 3 & 2 & 5 & 2 & 3 & 2 \\
Laredo & TX & 1 & 2 & 3 & 4 & 3 & 2 & 5 \\
Las Cruces & NM & 1 & 2 & 3 & 4 & 3 & 3 & 1 \\
Las Vegas-Paradise & NV & 3 & 2 & 3 & 4 & 5 & 1 & 2 \\
Lawrence & KS & 3 & 1 & 2 & 5 & 2 & 3 & 2 \\
Lawton & OK & 1 & 4 & 5 & 5 & 3 & 3 & 1 \\
Lebanon & PA & 5 & 5 & 1 & 1 & 4 & 3 & 4 \\
Lewiston & ID-WA & 1 & 3 & 2 & 3 & 2 & 3 & 4 \\
Lewiston-Auburn & ME & 2 & 5 & 1 & 3 & 4 & 1 & 4 \\
Lexington-Fayette & KY & 2 & 3 & 2 & 3 & 2 & 4 & 2 \\
Lima & OH & 5 & 5 & 1 & 1 & 4 & 4 & 4 \\
Lincoln & NE & 4 & 3 & 2 & 5 & 2 & 3 & 2 \\
Little Rock-North Little Rock-Conway & AR & 4 & 3 & 2 & 5 & 2 & 3 & 3 \\
Logan & UT-ID & 4 & 1 & 2 & 3 & 2 & 3 & 2 \\
Longview & TX & 1 & 4 & 4 & 2 & 3 & 2 & 3 \\
Longview & WA & 1 & 2 & 3 & 2 & 4 & 1 & 4 \\
Los Angeles-Long Beach-Santa Ana & CA & 3 & 2 & 3 & 1 & 5 & 1 & 5 \\
Louisville/Jefferson County & KY-IN & 2 & 3 & 2 & 3 & 2 & 1 & 4 \\
Lubbock & TX & 4 & 3 & 2 & 3 & 3 & 3 & 2 \\
Lynchburg & VA & 2 & 3 & 2 & 3 & 2 & 1 & 4 \\
Macon & GA & 2 & 3 & 2 & 3 & 2 & 3 & 3 \\
Madera-Chowchilla & CA & 1 & 2 & 3 & 4 & 1 & 3 & 4 \\
Madison & WI & 4 & 1 & 2 & 5 & 2 & 3 & 2 \\
Manchester-Nashua & NH & 2 & 1 & 2 & 5 & 5 & 1 & 4 \\
Manhattan & KS & 1 & 1 & 5 & 5 & 2 & 3 & 1 \\
Mankato-North Mankato & MN & 4 & 3 & 2 & 3 & 2 & 3 & 2 \\
Mansfield & OH & 5 & 5 & 1 & 1 & 4 & 4 & 4 \\
McAllen-Edinburg-Mission & TX & 1 & 2 & 3 & 4 & 3 & 3 & 5 \\
Medford & OR & 3 & 2 & 3 & 4 & 2 & 4 & 2 \\
Memphis & TN-MS-AR & 2 & 3 & 2 & 3 & 2 & 1 & 3 \\
Merced & CA & 1 & 2 & 3 & 4 & 1 & 3 & 1 \\
Miami-Fort Lauderdale-Pompano Beach & FL & 3 & 2 & 3 & 4 & 5 & 1 & 5 \\
Michigan City-La Porte & IN & 5 & 5 & 1 & 3 & 4 & 1 & 4 \\
Midland & TX & 1 & 4 & 4 & 2 & 3 & 2 & 2 \\
Milwaukee-Waukesha-West Allis & WI & 2 & 5 & 1 & 3 & 4 & 1 & 5 \\
Minneapolis-St. Paul-Bloomington & MN-WI & 2 & 1 & 2 & 5 & 5 & 1 & 2 \\
Missoula & MT & 3 & 2 & 3 & 3 & 2 & 3 & 2 \\
Mobile & AL & 2 & 3 & 2 & 2 & 2 & 3 & 3 \\
Modesto & CA & 1 & 2 & 3 & 4 & 1 & 3 & 4 \\
Monroe & LA & 2 & 3 & 2 & 3 & 2 & 3 & 3 \\
Monroe & MI & 5 & 5 & 1 & 1 & 4 & 4 & 4 \\
Montgomery & AL & 2 & 3 & 2 & 3 & 2 & 3 & 3 \\
Morgantown & WV & 4 & 3 & 2 & 5 & 2 & 3 & 2 \\
Morristown & TN & 2 & 3 & 1 & 1 & 4 & 4 & 4 \\
Mount Vernon-Anacortes & WA & 1 & 2 & 3 & 4 & 1 & 4 & 2 \\
Muncie & IN & 5 & 5 & 1 & 1 & 4 & 4 & 4 \\
Muskegon-Norton Shores & MI & 5 & 5 & 1 & 1 & 4 & 1 & 4 \\
Myrtle Beach-North Myrtle Beach-Conway & SC & 3 & 2 & 3 & 4 & 5 & 3 & 3 \\
Napa & CA & 3 & 2 & 3 & 4 & 1 & 3 & 2 \\
Naples-Marco Island & FL & 3 & 2 & 3 & 4 & 5 & 4 & 2 \\
Nashville-Davidson--Murfreesboro--Franklin & TN & 3 & 1 & 2 & 3 & 2 & 4 & 3 \\
New Haven-Milford & CT & 2 & 5 & 1 & 5 & 5 & 1 & 5 \\
New Orleans-Metairie-Kenner & LA & 3 & 4 & 4 & 2 & 5 & 3 & 3 \\
New York-Northern New Jersey-Long Island & NY-NJ-PA & 3 & 1 & 2 & 5 & 5 & 1 & 5 \\
Niles-Benton Harbor & MI & 5 & 5 & 1 & 3 & 4 & 4 & 4 \\
North Port-Bradenton-Sarasota & FL & 3 & 2 & 3 & 4 & 5 & 4 & 2 \\
Norwich-New London & CT & 1 & 4 & 5 & 1 & 5 & 3 & 2 \\
Ocala & FL & 3 & 2 & 3 & 4 & 2 & 1 & 3 \\
Ocean City & NJ & 3 & 2 & 3 & 4 & 5 & 3 & 2 \\
Odessa & TX & 1 & 4 & 4 & 2 & 3 & 2 & 2 \\
Ogden-Clearfield & UT & 4 & 3 & 2 & 3 & 5 & 3 & 1 \\
Oklahoma City & OK & 4 & 4 & 4 & 2 & 3 & 3 & 2 \\
Olympia & WA & 1 & 3 & 2 & 5 & 2 & 3 & 2 \\
Omaha-Council Bluffs & NE-IA & 4 & 1 & 2 & 3 & 2 & 3 & 2 \\
Orlando-Kissimmee-Sanford & FL & 3 & 2 & 3 & 4 & 5 & 1 & 2 \\
Oshkosh-Neenah & WI & 5 & 3 & 2 & 3 & 4 & 3 & 4 \\
Owensboro & KY & 5 & 3 & 1 & 3 & 4 & 1 & 4 \\
Oxnard-Thousand Oaks-Ventura & CA & 3 & 2 & 3 & 4 & 1 & 3 & 2 \\
Palm Bay-Melbourne-Titusville & FL & 3 & 1 & 5 & 5 & 5 & 1 & 2 \\
Palm Coast & FL & 3 & 2 & 3 & 4 & 5 & 1 & 3 \\
Panama City-Lynn Haven-Panama City Beach & FL & 3 & 2 & 3 & 4 & 5 & 3 & 1 \\
Parkersburg-Marietta-Vienna & WV-OH & 5 & 5 & 1 & 1 & 4 & 4 & 4 \\
Pascagoula & MS & 1 & 4 & 5 & 2 & 3 & 2 & 4 \\
Pensacola-Ferry Pass-Brent & FL & 1 & 4 & 3 & 2 & 3 & 3 & 1 \\
Peoria & IL & 5 & 5 & 1 & 1 & 4 & 3 & 4 \\
Philadelphia-Camden-Wilmington & PA-NJ-DE-MD & 2 & 5 & 2 & 3 & 5 & 1 & 5 \\
Phoenix-Mesa-Glendale & AZ & 3 & 2 & 3 & 4 & 5 & 1 & 2 \\
Pine Bluff & AR & 5 & 5 & 1 & 1 & 4 & 3 & 3 \\
Pittsburgh & PA & 5 & 5 & 1 & 1 & 4 & 3 & 4 \\
Pittsfield & MA & 5 & 5 & 1 & 1 & 4 & 1 & 4 \\
Pocatello & ID & 5 & 3 & 1 & 1 & 2 & 3 & 2 \\
Portland-South Portland-Biddeford & ME & 3 & 3 & 2 & 3 & 5 & 1 & 2 \\
Portland-Vancouver-Hillsboro & OR-WA & 3 & 1 & 3 & 3 & 5 & 1 & 2 \\
Port St. Lucie & FL & 3 & 2 & 3 & 4 & 5 & 1 & 3 \\
Poughkeepsie-Newburgh-Middletown & NY & 3 & 3 & 2 & 3 & 4 & 3 & 2 \\
Prescott & AZ & 3 & 2 & 3 & 4 & 3 & 4 & 2 \\
Providence-New Bedford-Fall River & RI-MA & 2 & 5 & 1 & 1 & 4 & 1 & 5 \\
Provo-Orem & UT & 4 & 1 & 2 & 5 & 5 & 1 & 2 \\
Pueblo & CO & 5 & 5 & 1 & 1 & 4 & 3 & 2 \\
Punta Gorda & FL & 3 & 2 & 3 & 4 & 5 & 4 & 2 \\
Racine & WI & 5 & 5 & 1 & 1 & 4 & 1 & 4 \\
Raleigh-Cary & NC & 3 & 1 & 2 & 5 & 5 & 1 & 3 \\
Rapid City & SD & 3 & 3 & 2 & 4 & 3 & 3 & 2 \\
Reading & PA & 5 & 5 & 1 & 1 & 4 & 1 & 4 \\
Redding & CA & 3 & 3 & 3 & 4 & 1 & 3 & 2 \\
Reno-Sparks & NV & 3 & 2 & 3 & 2 & 5 & 1 & 2 \\
Richmond & VA & 2 & 3 & 2 & 5 & 2 & 1 & 3 \\
Riverside-San Bernardino-Ontario & CA & 1 & 2 & 3 & 4 & 5 & 1 & 2 \\
Roanoke & VA & 5 & 3 & 2 & 3 & 2 & 1 & 4 \\
Rochester & MN & 4 & 1 & 2 & 5 & 2 & 3 & 2 \\
Rochester & NY & 5 & 5 & 1 & 1 & 2 & 3 & 4 \\
Rockford & IL & 5 & 5 & 1 & 1 & 4 & 4 & 4 \\
Rocky Mount & NC & 2 & 5 & 1 & 1 & 4 & 4 & 3 \\
Rome & GA & 2 & 3 & 2 & 3 & 2 & 4 & 4 \\
Sacramento--Arden-Arcade--Roseville & CA & 3 & 2 & 3 & 5 & 1 & 3 & 2 \\
Saginaw-Saginaw Township North & MI & 5 & 5 & 1 & 1 & 4 & 4 & 4 \\
St. Cloud & MN & 4 & 3 & 2 & 3 & 2 & 3 & 4 \\
St. George & UT & 3 & 2 & 3 & 4 & 5 & 1 & 2 \\
St. Joseph & MO-KS & 5 & 3 & 2 & 3 & 2 & 3 & 4 \\
St. Louis & MO-IL & 5 & 5 & 2 & 3 & 2 & 1 & 4 \\
Salem & OR & 3 & 3 & 3 & 4 & 2 & 3 & 2 \\
Salinas & CA & 3 & 2 & 3 & 4 & 1 & 3 & 1 \\
Salisbury & MD & 4 & 3 & 2 & 3 & 2 & 1 & 3 \\
Salt Lake City & UT & 4 & 1 & 2 & 5 & 5 & 1 & 2 \\
San Angelo & TX & 1 & 4 & 2 & 2 & 3 & 3 & 2 \\
San Antonio-New Braunfels & TX & 1 & 1 & 2 & 5 & 3 & 3 & 1 \\
San Diego-Carlsbad-San Marcos & CA & 3 & 1 & 3 & 5 & 5 & 1 & 1 \\
Sandusky & OH & 5 & 5 & 1 & 1 & 4 & 4 & 4 \\
San Francisco-Oakland-Fremont & CA & 3 & 1 & 3 & 1 & 5 & 1 & 5 \\
San Jose-Sunnyvale-Santa Clara & CA & 3 & 1 & 3 & 5 & 5 & 1 & 5 \\
San Luis Obispo-Paso Robles & CA & 3 & 2 & 3 & 4 & 1 & 3 & 2 \\
Santa Barbara-Santa Maria-Goleta & CA & 3 & 2 & 3 & 4 & 1 & 3 & 2 \\
Santa Cruz-Watsonville & CA & 3 & 2 & 3 & 4 & 1 & 3 & 2 \\
Santa Fe & NM & 3 & 2 & 3 & 5 & 3 & 3 & 2 \\
Santa Rosa-Petaluma & CA & 3 & 2 & 3 & 4 & 1 & 1 & 2 \\
Savannah & GA & 2 & 3 & 2 & 3 & 2 & 3 & 3 \\
Scranton--Wilkes-Barre & PA & 5 & 5 & 1 & 1 & 4 & 1 & 4 \\
Seattle-Tacoma-Bellevue & WA & 3 & 1 & 2 & 5 & 5 & 1 & 2 \\
Sebastian-Vero Beach & FL & 3 & 2 & 3 & 4 & 5 & 4 & 3 \\
Sheboygan & WI & 5 & 3 & 1 & 1 & 4 & 1 & 4 \\
Sherman-Denison & TX & 4 & 3 & 2 & 3 & 5 & 1 & 4 \\
Shreveport-Bossier City & LA & 1 & 4 & 4 & 2 & 3 & 3 & 3 \\
Sioux City & IA-NE-SD & 4 & 3 & 2 & 4 & 4 & 3 & 4 \\
Sioux Falls & SD & 4 & 1 & 2 & 3 & 2 & 3 & 2 \\
South Bend-Mishawaka & IN-MI & 5 & 5 & 1 & 1 & 4 & 4 & 4 \\
Spartanburg & SC & 2 & 5 & 1 & 1 & 4 & 1 & 4 \\
Spokane & WA & 4 & 3 & 2 & 3 & 2 & 1 & 2 \\
Springfield & IL & 5 & 3 & 2 & 5 & 2 & 3 & 2 \\
Springfield & MA & 5 & 5 & 1 & 1 & 4 & 1 & 4 \\
Springfield & MO & 4 & 3 & 2 & 3 & 2 & 1 & 4 \\
Springfield & OH & 5 & 5 & 1 & 1 & 4 & 4 & 4 \\
State College & PA & 4 & 1 & 2 & 5 & 2 & 3 & 2 \\
Steubenville-Weirton & OH-WV & 5 & 5 & 1 & 1 & 4 & 3 & 4 \\
Stockton & CA & 1 & 2 & 3 & 4 & 1 & 3 & 2 \\
Sumter & SC & 2 & 3 & 1 & 3 & 2 & 3 & 3 \\
Syracuse & NY & 5 & 5 & 1 & 1 & 4 & 3 & 2 \\
Tallahassee & FL & 3 & 1 & 2 & 5 & 2 & 1 & 3 \\
Tampa-St. Petersburg-Clearwater & FL & 3 & 2 & 3 & 4 & 5 & 1 & 2 \\
Terre Haute & IN & 5 & 5 & 1 & 1 & 4 & 3 & 4 \\
Texarkana, TX-Texarkana & AR & 5 & 3 & 2 & 3 & 3 & 1 & 3 \\
Toledo & OH & 5 & 5 & 1 & 1 & 4 & 4 & 4 \\
Topeka & KS & 4 & 3 & 2 & 3 & 2 & 3 & 2 \\
Trenton-Ewing & NJ & 2 & 5 & 1 & 5 & 2 & 1 & 5 \\
Tucson & AZ & 1 & 2 & 3 & 5 & 3 & 1 & 2 \\
Tulsa & OK & 4 & 4 & 4 & 2 & 3 & 2 & 4 \\
Tuscaloosa & AL & 2 & 3 & 2 & 1 & 2 & 3 & 3 \\
Tyler & TX & 4 & 3 & 2 & 3 & 3 & 3 & 3 \\
Utica-Rome & NY & 5 & 5 & 1 & 1 & 4 & 3 & 4 \\
Valdosta & GA & 2 & 3 & 2 & 3 & 2 & 3 & 3 \\
Vallejo-Fairfield & CA & 1 & 2 & 3 & 5 & 1 & 3 & 1 \\
Victoria & TX & 1 & 4 & 4 & 2 & 3 & 2 & 4 \\
Vineland-Millville-Bridgeton & NJ & 2 & 5 & 1 & 1 & 4 & 1 & 4 \\
Virginia Beach-Norfolk-Newport News & VA-NC & 1 & 4 & 5 & 2 & 5 & 3 & 1 \\
Visalia-Porterville & CA & 1 & 2 & 3 & 4 & 1 & 3 & 2 \\
Waco & TX & 4 & 3 & 2 & 3 & 3 & 3 & 3 \\
Warner Robins & GA & 1 & 4 & 5 & 5 & 3 & 3 & 1 \\
Washington-Arlington-Alexandria & DC-VA-MD-WV & 3 & 1 & 2 & 5 & 5 & 3 & 1 \\
Waterloo-Cedar Falls & IA & 5 & 3 & 2 & 1 & 4 & 3 & 4 \\
Wausau & WI & 4 & 3 & 2 & 3 & 2 & 3 & 4 \\
Wenatchee-East Wenatchee & WA & 3 & 2 & 3 & 4 & 1 & 3 & 2 \\
Wheeling & WV-OH & 5 & 5 & 1 & 1 & 4 & 3 & 4 \\
Wichita & KS & 5 & 3 & 1 & 1 & 3 & 1 & 2 \\
Wichita Falls & TX & 4 & 3 & 2 & 2 & 3 & 3 & 1 \\
Williamsport & PA & 5 & 5 & 1 & 1 & 4 & 1 & 4 \\
Wilmington & NC & 3 & 1 & 2 & 3 & 2 & 1 & 3 \\
Winchester & VA-WV & 4 & 3 & 2 & 3 & 2 & 4 & 4 \\
Winston-Salem & NC & 2 & 1 & 1 & 1 & 2 & 1 & 4 \\
Worcester & MA & 2 & 5 & 1 & 3 & 4 & 1 & 4 \\
Yakima & WA & 1 & 2 & 3 & 4 & 1 & 3 & 2 \\
York-Hanover & PA & 5 & 3 & 1 & 3 & 4 & 1 & 4 \\
Youngstown-Warren-Boardman & OH-PA & 5 & 5 & 1 & 1 & 4 & 4 & 4 \\
Yuba City & CA & 1 & 2 & 3 & 4 & 1 & 3 & 1 \\
Yuma & AZ & 1 & 2 & 3 & 4 & 1 & 3 & 1 \\
\end{longtable}
}

\subsection*{Cluster Descriptions by Method}

\paragraph{Corpus CBOW (CB).}
\begin{enumerate}[nosep]
    \item \emph{Amenity Migration and Population-Led Growth} (\sizeCBi{} CBSAs). Sun Belt and Western metros whose growth is driven by in-migration of retirees and lifestyle migrants rather than by industry. Construction, real estate, retail, and healthcare are the core employers, which leaves these labor markets acutely exposed to housing cycles and to the 2008 shock. Examples: Asheville, Austin, Boise~City, Cape~Coral, Dallas.

    \item \emph{University and Hospital Anchors} (\sizeCBii{} CBSAs). Flagship universities and major hospital systems dominate employment. Unemployment runs persistently below national averages and educational attainment is high, but the labor market is bifurcated between credentialed professionals and low-wage student and service workers. Examples: Ames, Ann~Arbor, Boston, Boulder, Champaign--Urbana.

    \item \emph{Regional Hubs in Manufacturing Transition} (\sizeCBiii{} CBSAs). Legacy manufacturing---textiles, paper, steel, timber---has declined since 1979 and been partly replaced by ``eds and meds,'' logistics along interstate and port corridors, and government or defense spending. These metros serve as retail and medical hubs for rural hinterlands, with below-average wages and uneven recoveries. Examples: Allentown, Birmingham, Charlotte, Chattanooga, Cincinnati.

    \item \emph{Severe Deindustrialization} (\sizeCBiv{} CBSAs). Steel, auto, textile, and tire dependence followed by heavy job losses after 1979. Loss of high-wage unionized work, wage stagnation, population decline, and out-migration of younger workers, with ``eds and meds,'' retail, and logistics as partial replacements. Examples: Akron, Buffalo, Chicago, Cleveland, Detroit.

    \item \emph{Agriculture and Resource Extraction} (\sizeCBv{} CBSAs). Economies anchored in agriculture, oil and gas, timber, or mining, with boom-bust cycles tied to commodity prices and weather. Seasonal low-wage work, chronically high unemployment and poverty, and large Hispanic and immigrant workforces, often near the border. Examples: Bakersfield, Brownsville, Corpus~Christi, El~Paso, Fresno.
\end{enumerate}

\paragraph{Corpus SGNS (CS).}
\begin{enumerate}[nosep]
    \item \emph{Knowledge and Anchor-Institution Economies} (\sizeCSi{} CBSAs). Universities, hospital systems and research employers dominate, with a shift out of manufacturing, resource extraction and defense into knowledge work and services. Unemployment is persistently low and recoveries---2008 in particular---are fast, but growth arrives with housing pressure and a labor market split between credentialed and low-wage service work. Examples: Albuquerque, Ames, Ann~Arbor, Atlanta, Austin.

    \item \emph{Amenity Migration and Tourism} (\sizeCSii{} CBSAs). Sun Belt, Western and coastal destinations growing through retiree and lifestyle in-migration rather than through industry. Tourism, hospitality, retail and healthcare dominate, construction and real estate amplify the cycle, and the 2008 housing shock hit hard. Seasonal low-wage work and large Hispanic workforces are common. Examples: Asheville, Bakersfield, Barnstable~Town, Bend, Cape~Coral.

    \item \emph{Regional Hubs in Transition} (\sizeCSiii{} CBSAs). The largest group: mid-sized metros that lost textiles, paper, timber, mining or tobacco employment and replaced part of it with ``eds and meds,'' logistics along interstate corridors, and state or federal anchors. Each serves as the retail and medical hub of a broad rural hinterland. Examples: Abilene, Albany~(GA), Albany--Schenectady--Troy, Amarillo, Appleton.

    \item \emph{Defense and Energy Single-Sector Metros} (\sizeCSiv{} CBSAs). One dominant sector---a military installation or oil, gas and petrochemicals---rather than a diversified base, so employment tracks Pentagon budgets or commodity prices. Blue-collar workforces, below-average wages, and Gulf Coast exposure to hurricanes and spills. Examples: Baton~Rouge, Beaumont--Port~Arthur, Bremerton--Silverdale, Casper, Colorado~Springs.

    \item \emph{Severe Deindustrialization} (\sizeCSv{} CBSAs). Steel, auto, textile, rubber, glass and coal employment collapsing from the early-1980s recessions onward, often around a single dominant employer. High-wage unionized work gives way to services, prisons and warehousing, alongside population loss and an aging workforce. Examples: Akron, Allentown, Altoona, Atlantic~City, Binghamton.
\end{enumerate}

\paragraph{Closed-Form SVD-$\beta$ at $r=50$ (SB).}
\begin{enumerate}[nosep]
    \item \emph{Deindustrialization} (\sizeSBi{} CBSAs). Collapse of steel, auto, textile, and machinery employment; loss of unionized blue-collar jobs, population decline, and weak post-2008 recoveries. Examples: Akron, Allentown, Binghamton, Buffalo, Canton.

    \item \emph{Anchor Institutions and Regional Hubs} (\sizeSBii{} CBSAs). ``Eds and meds'' as dominant, recession-resistant employers, often alongside state government, serving as service hubs for rural hinterlands. Low unemployment but below-average wages. Examples: Albuquerque, Ann~Arbor, Atlanta, Austin, Madison.

    \item \emph{Amenity, Tourism and Agriculture} (\sizeSBiii{} CBSAs). Sun Belt and coastal in-migration destinations driven by retirees and lifestyle migration, together with seasonal agricultural and border economies. Heavy tourism and construction exposure and a severe 2008 housing shock. Examples: Asheville, Bakersfield, Bend, Brownsville, Cape~Coral.

    \item \emph{Energy and Resource Extraction} (\sizeSBiv{} CBSAs). Oil, gas, refining, and coal dependence, with boom-bust cycles tied to commodity prices rather than the national business cycle. Examples: Baton~Rouge, Beaumont, Casper, Corpus~Christi, Houston.

    \item \emph{Military and Defense Installations} (\sizeSBv{} CBSAs). A dominant Army post, Navy base, or shipyard anchors each metro; fortunes track federal budgets and BRAC rounds rather than market conditions. Young, transient, diverse populations. Examples: Clarksville, Colorado~Springs, Fayetteville, Huntsville, Killeen.
\end{enumerate}

\paragraph{Pretrained Google News Vectors (PW).}
\begin{enumerate}[nosep]
    \item \emph{Deindustrialized Manufacturing} (\sizePWi{} CBSAs). Steel, auto, textile, rubber, and machinery dependence followed by sustained decline after 1979. High-wage unionized jobs give way to lower-paying service work, with single-employer vulnerability and sharp exposure to the 1981--82 and 2008--09 recessions. Examples: Akron, Allentown, Buffalo, Canton, Cleveland.

    \item \emph{Energy and Resource Extraction} (\sizePWii{} CBSAs). Oil, gas, coal, refining, and petrochemicals anchor these economies, so employment tracks commodity prices rather than the national cycle: the 1980s oil bust, the shale boom of the 2000s and 2010s, and the 2014--16 collapse. Diversification is incomplete. Examples: Baton~Rouge, Casper, Corpus~Christi, Houston, Lake~Charles.

    \item \emph{Regional Service Hubs} (\sizePWiii{} CBSAs). Small and mid-sized metros where legacy manufacturing and agriculture have given way to hospitals and anchor universities, retail and warehousing along interstate corridors, and tourism. They serve as hubs for multi-county rural hinterlands, with military bases and state government as stabilizers. Examples: Abilene, Asheville, Atlanta, Charlotte, Chattanooga.

    \item \emph{Seasonal and Amenity Economies} (\sizePWiv{} CBSAs). Tourism, hospitality, and recreation with sharp seasonal swings, retiree and ``snowbird'' in-migration, and heavy construction and real-estate dependence that left these metros badly exposed to the 2008 housing crash. Agriculture and extraction appear alongside. Examples: Bakersfield, Bend, Brownsville, Cape~Coral, Fresno.

    \item \emph{University and Government Centers} (\sizePWv{} CBSAs). Flagship universities, ``eds and meds,'' state capitals, federal agencies, national labs, and military bases dominate employment. Unemployment sits below national averages and these metros were cushioned in both the 1980s and 2008 downturns. Examples: Albuquerque, Ann~Arbor, Austin, Boston, Boulder.
\end{enumerate}

\paragraph{LLM Embeddings (LE).}
\begin{enumerate}[nosep]
    \item \emph{Agricultural Economies} (\sizeLEi{} CBSAs). Agriculture as the economic bedrock, often in high-value specialty crops, with limited diversification. Seasonal low-wage labor, pronounced employment cyclicality, large Latino and immigrant workforces, and above-average unemployment and poverty. Examples: Bakersfield, Fresno, Merced, Modesto, Salinas.

    \item \emph{Eds and Meds Regional Hubs} (\sizeLEii{} CBSAs). Universities and hospital systems as dominant, recession-resistant employers, following the decline of textiles, timber, steel, and tobacco. Retail, hospitality, and logistics replace blue-collar work, and these metros serve as hubs for rural hinterlands. Examples: Ann~Arbor, Asheville, Baltimore, Birmingham, Bloomington.

    \item \emph{Federal Anchors and Commodity Exposure} (\sizeLEiii{} CBSAs). Regional centers combining military bases, laboratories, and defense spending with boom-bust dependence on oil, gas, or agriculture. Healthcare and education are the growth sectors replacing extraction and manufacturing; wages sit below national averages. Examples: Albuquerque, Amarillo, Baton~Rouge, Colorado~Springs, Corpus~Christi.

    \item \emph{Heavy Industry and Deindustrialization} (\sizeLEiv{} CBSAs). Metros dependent in 1979 on steel, autos, textiles, paper, chemicals, or coal, often around a single dominant employer, that suffered plant closures and the early-1980s recession shock. Globalization, automation, and offshoring are the recurring drivers. Examples: Akron, Buffalo, Canton, Chicago, Cleveland.

    \item \emph{Diversified Service and Knowledge Economies} (\sizeLEv{} CBSAs). A shift from manufacturing and extraction toward services, with ``eds and meds'' as stable anchors and rising technology and professional-services sectors. In-migration drives growth, and construction exposure made these metros vulnerable in 2008. Examples: Atlanta, Austin, Boston, Boulder, Charlotte.
\end{enumerate}

\paragraph{Residual $k$-means (RK).}
This purely outcome-based clustering applies $k$-means with $L=5$ to the residuals from the pooled AR(2) specification with unit fixed effects (Section~\ref{sec:ar_panel}). Unlike the text-based methods, groups are formed to minimize within-group variation in employment dynamics rather than narrative similarity. The resulting clusters are less balanced than the text-based approaches---one group is a singleton and the largest contains \pctRKlargest\% of all CBSAs---but an ex-post examination of their composition reveals interpretable economic patterns:
\begin{enumerate}[nosep]
    \item \emph{Diversified, Higher-Density Metros} (\sizeRKi{} CBSAs, e.g., Atlanta, Austin, Akron, Allentown, Barnstable~Town). The densest group by some margin (1970 population density of 138 against 77--79 elsewhere), with moderate manufacturing ($20.6\%$ in 1980) and the second-highest college-educated share ($16.9\%$).

    \item \emph{Energy-Dependent} (\sizeRKii{} CBSAs, e.g., Houston, Beaumont--Port~Arthur, Corpus~Christi, Lafayette, Casper). Concentrated in Texas and Louisiana with the lowest manufacturing share of any non-singleton group ($11.8\%$); employment dynamics shaped by oil and gas cycles rather than the national one.

    \item \emph{Government-Anchored} (\sizeRKiii{} CBSAs, e.g., Albuquerque, Baltimore, Ames, Amarillo, Augusta--Richmond~County). The largest group, holding \pctRKlargest\% of the sample, with the highest government employment share ($23.4\%$) and the highest military share outside the singleton ($4.9\%$).

    \item \emph{Manufacturing-Heavy} (\sizeRKiv{} CBSAs, e.g., Battle~Creek, Bay~City, Chattanooga, Anderson, Asheville). The highest manufacturing share of any group ($26.2\%$) and the lowest college-educated share ($13.0\%$).

    \item \emph{Military Outlier} (\sizeRKv{} CBSA: Hinesville--Fort~Stewart, GA). A singleton with $61.0\%$ military and $79.3\%$ total government employment in 1980. Its impulse response is far from the pooled estimate, making it a residual outlier that no other CBSA resembles.

\end{enumerate}

The alignment between these ex-post labels and the text-based cluster themes---particularly the energy and manufacturing archetypes---provides a form of external validation: qualitatively similar economic groupings emerge whether cities are classified by narrative content or by the time-series behavior of their employment.

\clearpage

\section{Additional Numerical Results}\label{app-sec:simulations}

This appendix considers number of alternatives specifications for our numerical illustrations.

\subsection{Large Corpus}\label{app-sec:large_corpus}

Figures~\ref{app-fig:paper_scale_word_k2}--\ref{app-fig:paper_scale_doc_k3} repeat
Figures~\ref{fig:embeddings_large}--\ref{fig:doc_emb_noanchor_k3_alt} on the larger corpus
($D=1{,}000$, $N_d=10{,}000$, against $D=363$ and $N_d=487$). The comparison to make is panel against
its main-text counterpart, and what changes differs by panel.

For SVD-$\beta$ (left) only the sample size changes. We keep the full document as context, so the
number of ordered word--context pairs entering $\widehat R$ rises from $8.6\times 10^{7}$ to
$1.0\times 10^{11}$, a factor of about $1{,}200$, on $56$ times as many tokens. The panel is
therefore a pure sample-size comparison, and it is the cleanest picture of
Proposition~\ref{thm:factorization} in the paper: at $K=2$ the cloud collapses onto a line.

For SGNS (right) it does not. Full-document context is not feasible at this scale. We therefore run
$J=5$ (the default in \texttt{gensim}), and the two changes work against each other. \texttt{gensim} draws the effective half-window uniformly on $\{1,\dots,J\}$, so at $J=5$ each target sees about six contexts: the objective consumes roughly $6.0\times 10^{7}$ pairs per epoch, against $8.6\times 10^{7}$ in Figure~\ref{fig:embeddings_large}. So effectively, the training sample is about $30\%$ \emph{smaller} in pairs compared to its main-text counterpart. 

\begin{figure}[htbp]
\centering
\begin{subfigure}[b]{0.48\textwidth}
\centering
\includegraphics[width=\textwidth]{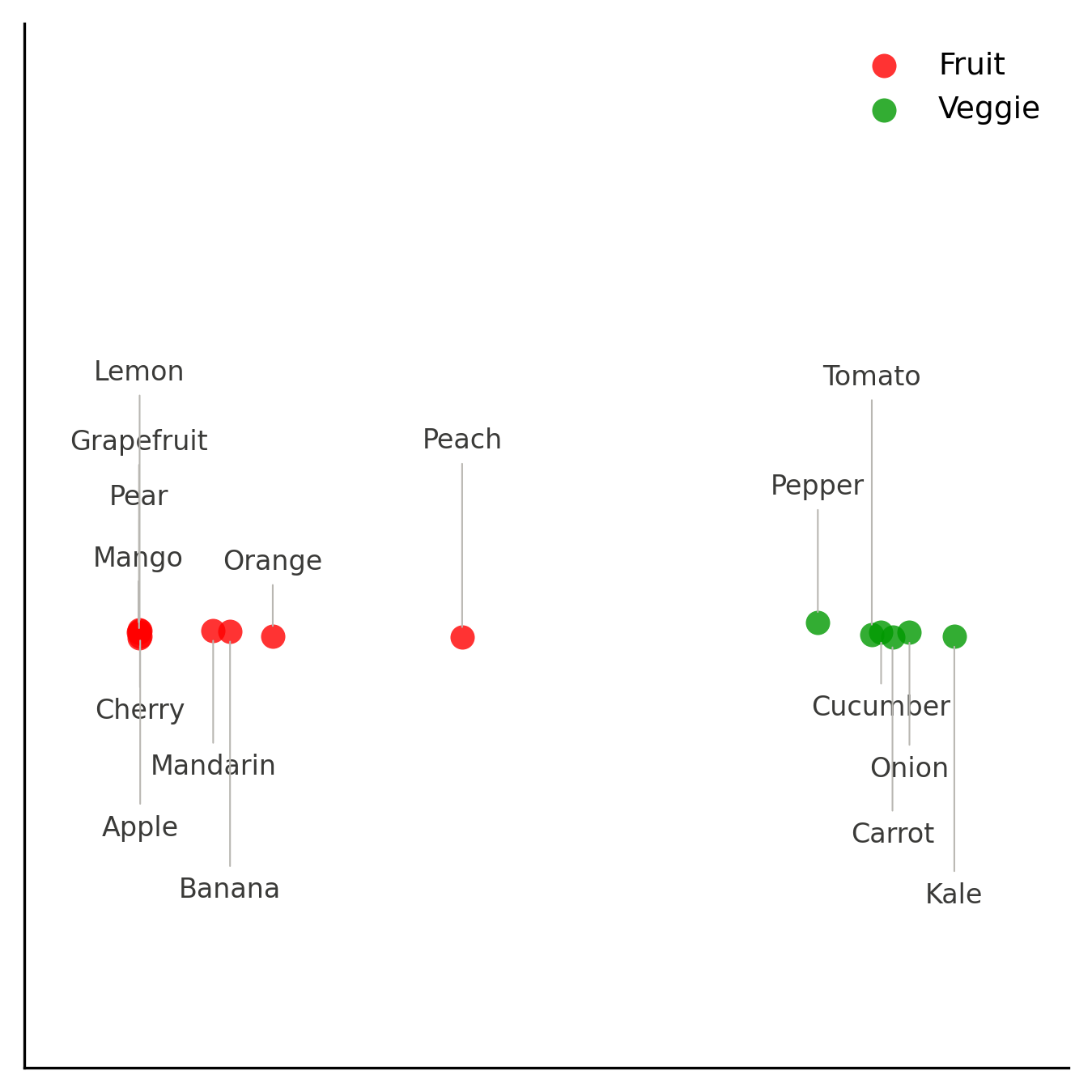}
\caption{SVD-$\beta$, full document}
\end{subfigure}
\hfill
\begin{subfigure}[b]{0.48\textwidth}
\centering
\includegraphics[width=\textwidth]{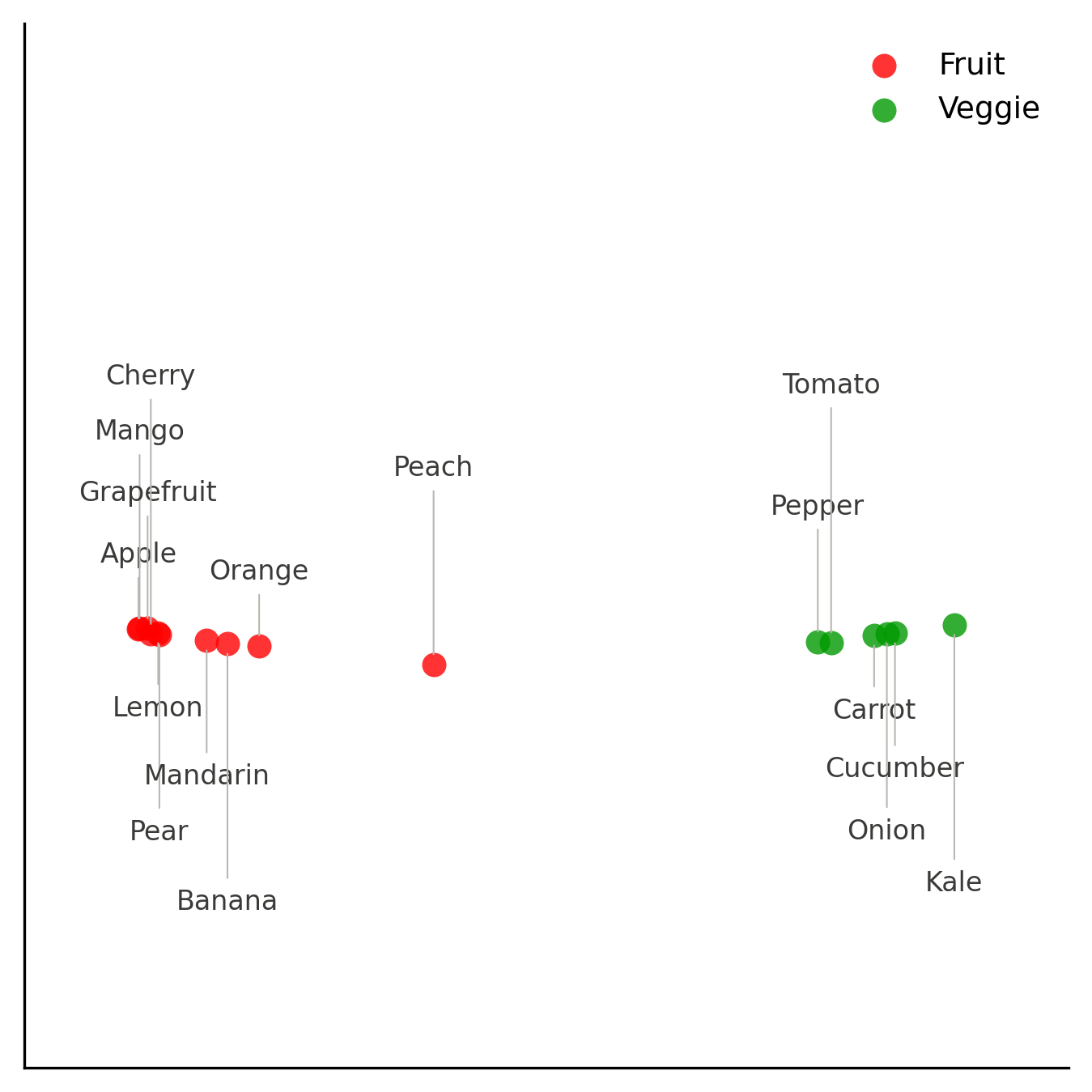}
\caption{SGNS, $J=5$}
\end{subfigure}
\caption[Word embeddings, Design 1 ($K=2$), large corpus]{Word embeddings for Design 1 ($K=2$) on the larger corpus ($D=1{,}000$, $N_d=10{,}000$). SVD-$\beta$ uses the full document as context, as in the main text; SGNS uses $J=5$, which is what is feasible at this scale. Counterpart of Figure~\ref{fig:embeddings_large}. Colors mark the dominant topic in the true $B$. Both panels are drawn at equal aspect: in population the embedding has rank $K-1=1$, so the cloud is one-dimensional and the vertical spread is estimation noise.}
\label{app-fig:paper_scale_word_k2}
\end{figure}

\begin{figure}[htb!]
\centering
\begin{subfigure}[b]{0.48\textwidth}
\centering
\includegraphics[width=\textwidth]{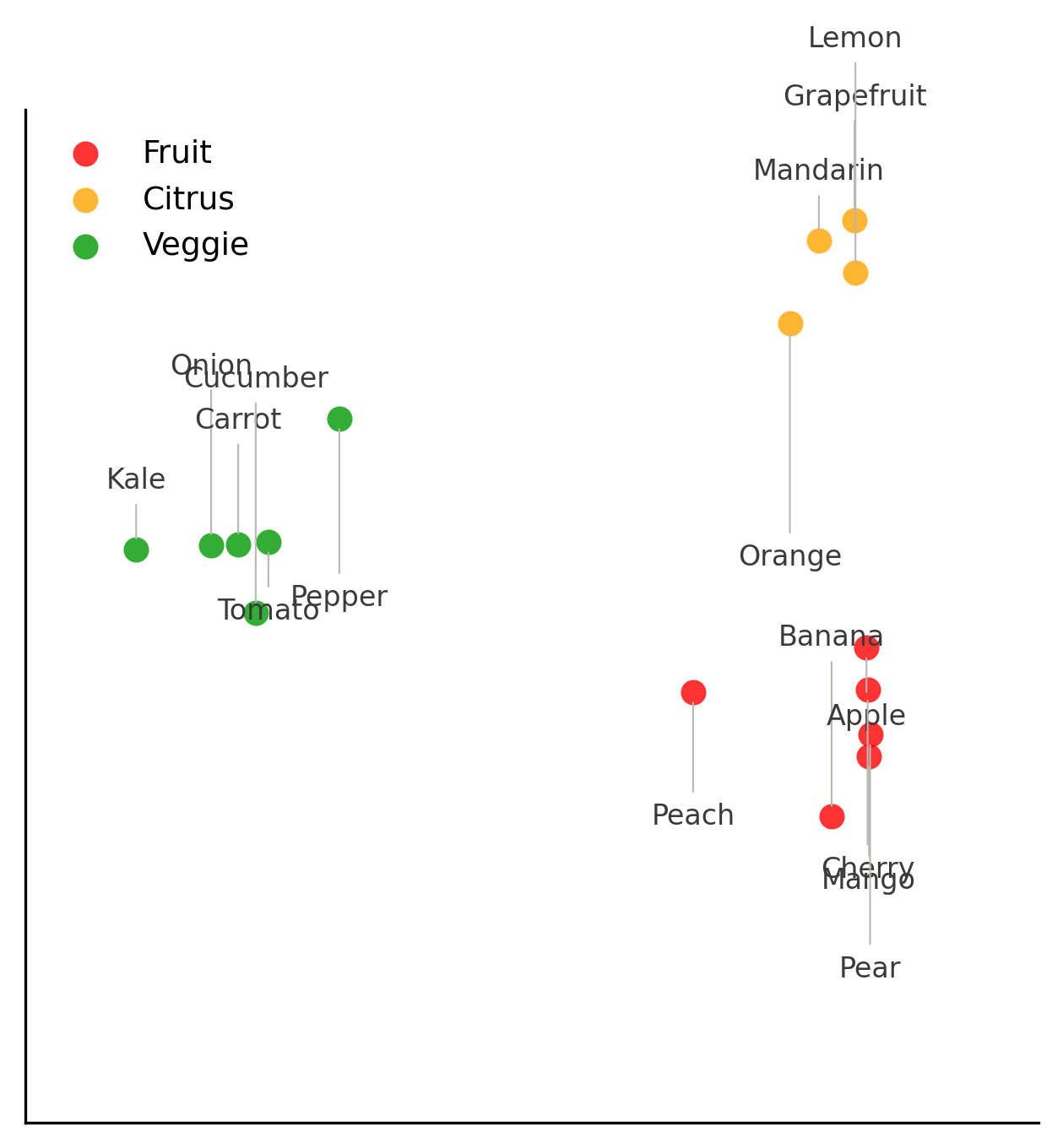}
\caption{SVD-$\beta$, full document}
\end{subfigure}
\hfill
\begin{subfigure}[b]{0.48\textwidth}
\centering
\includegraphics[width=\textwidth]{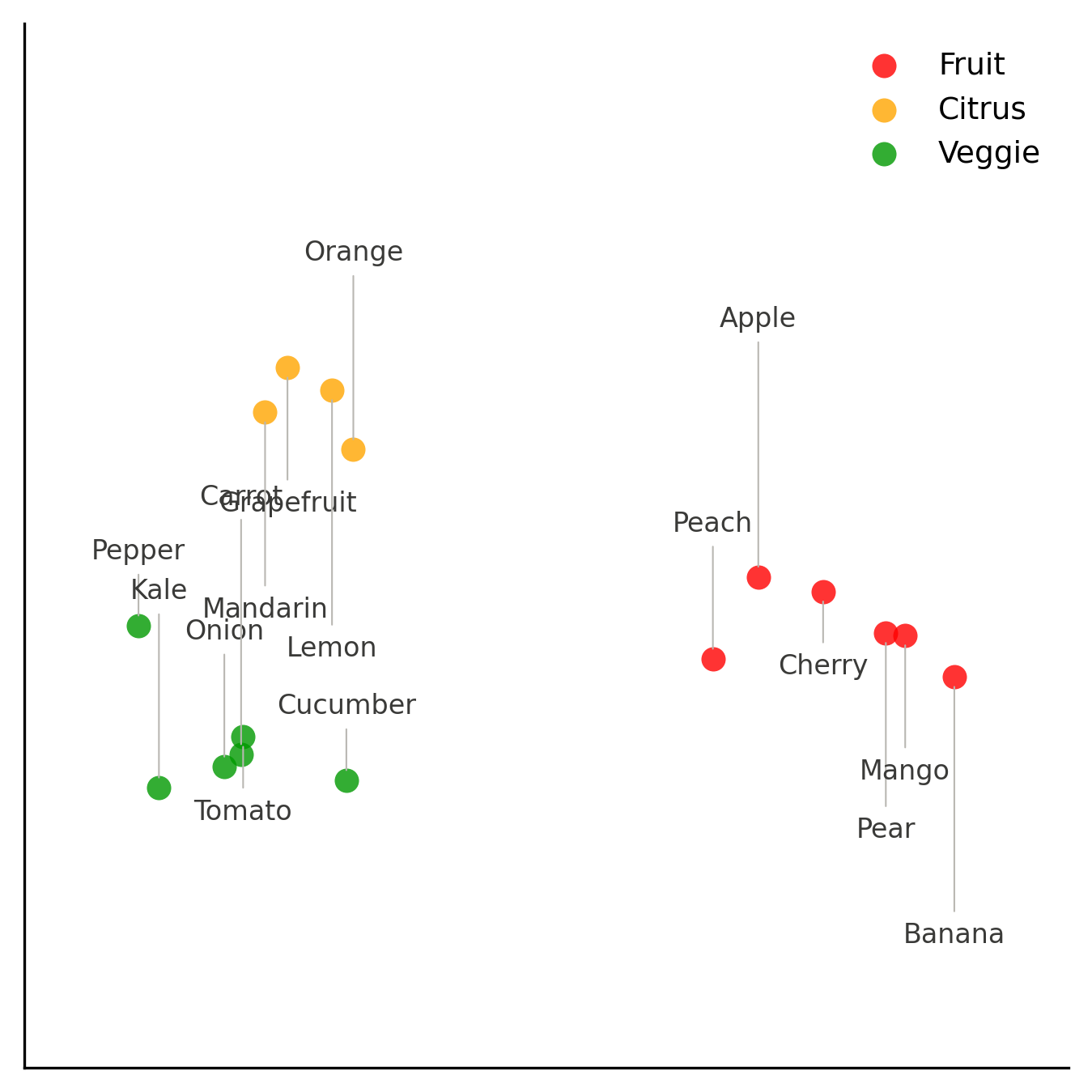}
\caption{SGNS, $J=5$}
\end{subfigure}
\caption[Word embeddings, Design 2 ($K=3$), large corpus]{Word embeddings for Design 2 ($K=3$) on the larger corpus. SVD-$\beta$ uses the full document as context; SGNS uses $J=5$. Counterpart of Figure~\ref{fig:embeddings_noanchor_k3}. Both methods separate the three topics. Each panel is projected onto its own leading two principal directions rather than onto the first two raw coordinates: for SVD-$\beta$ the two coincide, since $\beta$ is built from the ordered eigenvectors of $\widehat{R} - \mathbf{1}_V\mathbf{1}_V^\top$, but the CBOW coordinate basis is chosen by the optimizer and is arbitrary, so a raw two-coordinate slice would cut obliquely through the topic plane.}
\label{app-fig:paper_scale_word_k3}
\end{figure}

\begin{figure}[htb!]
\centering
\begin{subfigure}[b]{0.48\textwidth}
\centering
\includegraphics[width=\textwidth]{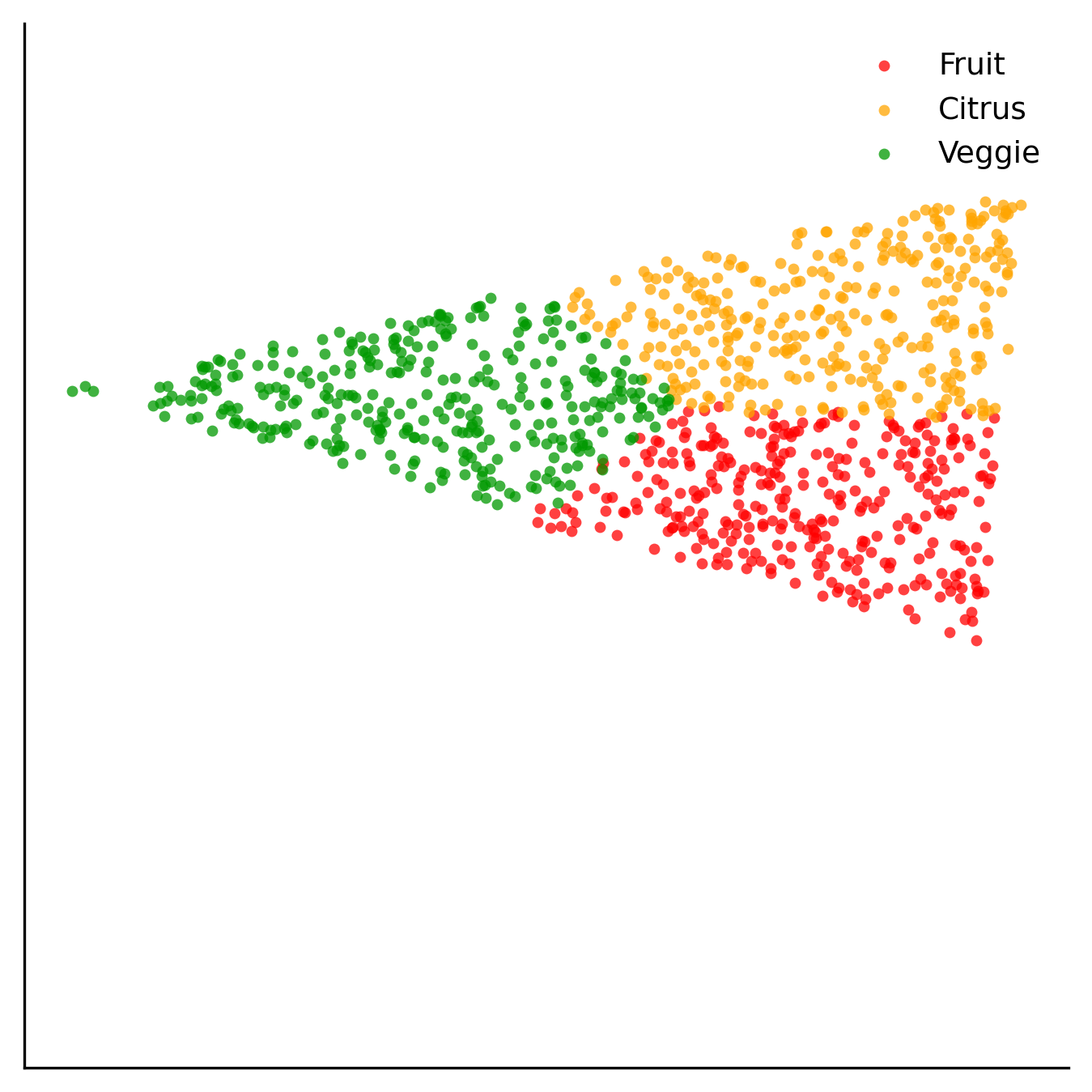}
\caption{SVD-$\beta$, full document}
\end{subfigure}
\hfill
\begin{subfigure}[b]{0.48\textwidth}
\centering
\includegraphics[width=\textwidth]{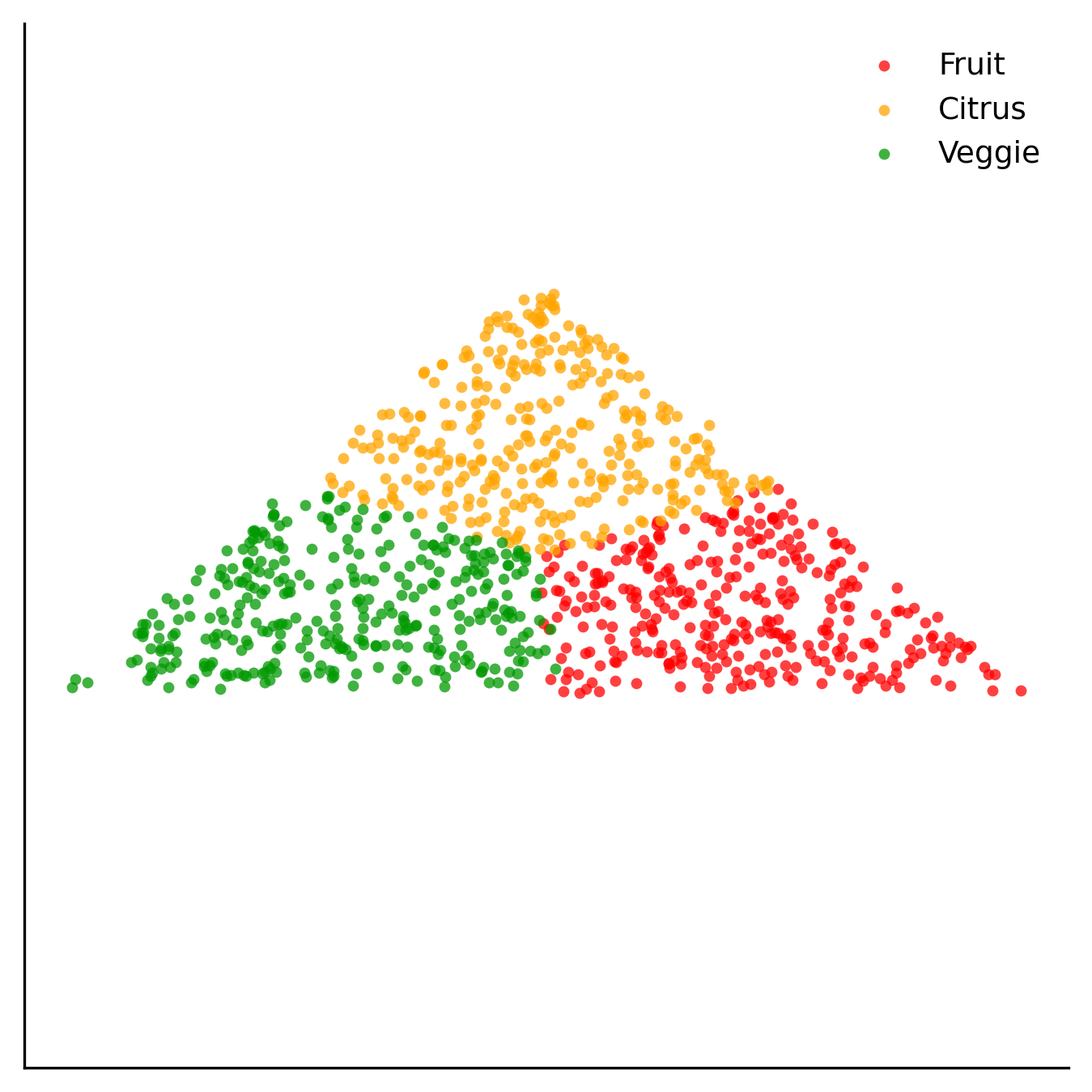}
\caption{SGNS, $J=5$}
\end{subfigure}
\caption[Document embeddings, Design 2 ($K=3$), large corpus]{Document embeddings for Design 2 ($K=3$) on the larger corpus, SVD-$\beta$ at full-document context and SGNS at $J=5$, colored by dominant topic. Counterpart of Figure~\ref{fig:doc_emb_noanchor_k3_alt}. With $\alpha=1$ the topic mixtures are spread across the simplex and the document embeddings fill the triangle spanned by the three topic centroids, as Proposition~\ref{prop:doc_embedding_general_freq} predicts.}
\label{app-fig:paper_scale_doc_k3}
\end{figure}

\clearpage

\subsection{Sparse Topic Mixtures: Dirichlet$(\alpha=0.01)$}\label{app-app:sim_sparse}

The simulations in Section~\ref{sec:simulations} use a symmetric Dirichlet$(\alpha=1)$ prior, which is uniform on the topic simplex. As a contrast, I repeat both designs with $\alpha=0.01$ (Holding $D=363$ and $N_d=487$). Under Dirichlet$(0.01)$, almost all mass concentrates at the simplex vertices: most documents are nearly pure-topic, with topic mixtures $\Theta_{\bullet d}$ close to a single $e_k$. By Proposition~\ref{prop:doc_embedding_general_freq}, document embeddings $\mu_d = C\,\Theta_{\bullet d}$ then concentrate at the topic centroids $\{c_k\}$ rather than filling the centroid simplex.

Figures~\ref{app-fig:sparse_word_emb_k2}--\ref{app-fig:sparse_doc_emb_k3_cluster} show the resulting embeddings. The $B$ matrices and pipeline are identical to Section~\ref{sec:simulations}; only the Dirichlet concentration parameter differs.

\begin{figure}[htbp]
\centering
\begin{subfigure}[b]{0.48\textwidth}
\centering
\includegraphics[width=\textwidth]{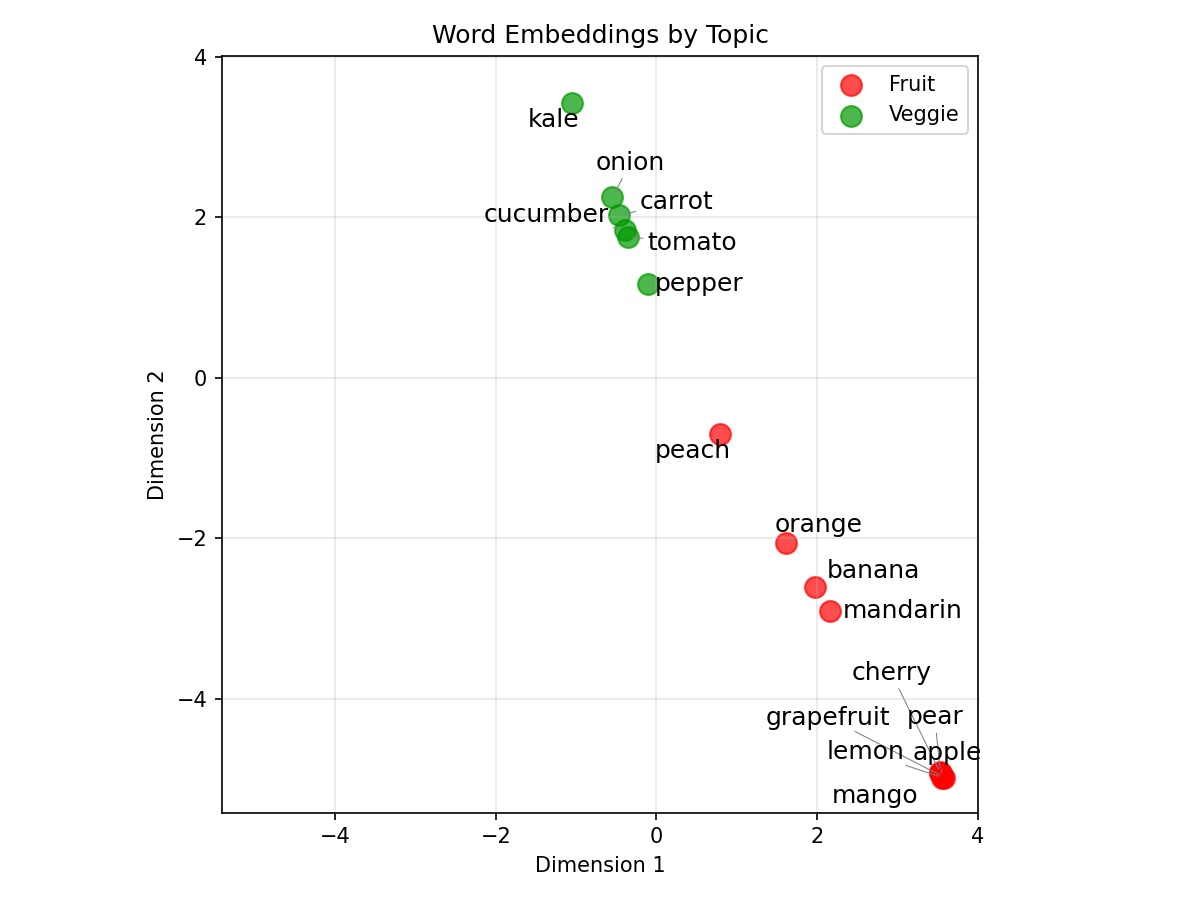}
\caption{CBOW}
\end{subfigure}
\hfill
\begin{subfigure}[b]{0.48\textwidth}
\centering
\includegraphics[width=\textwidth]{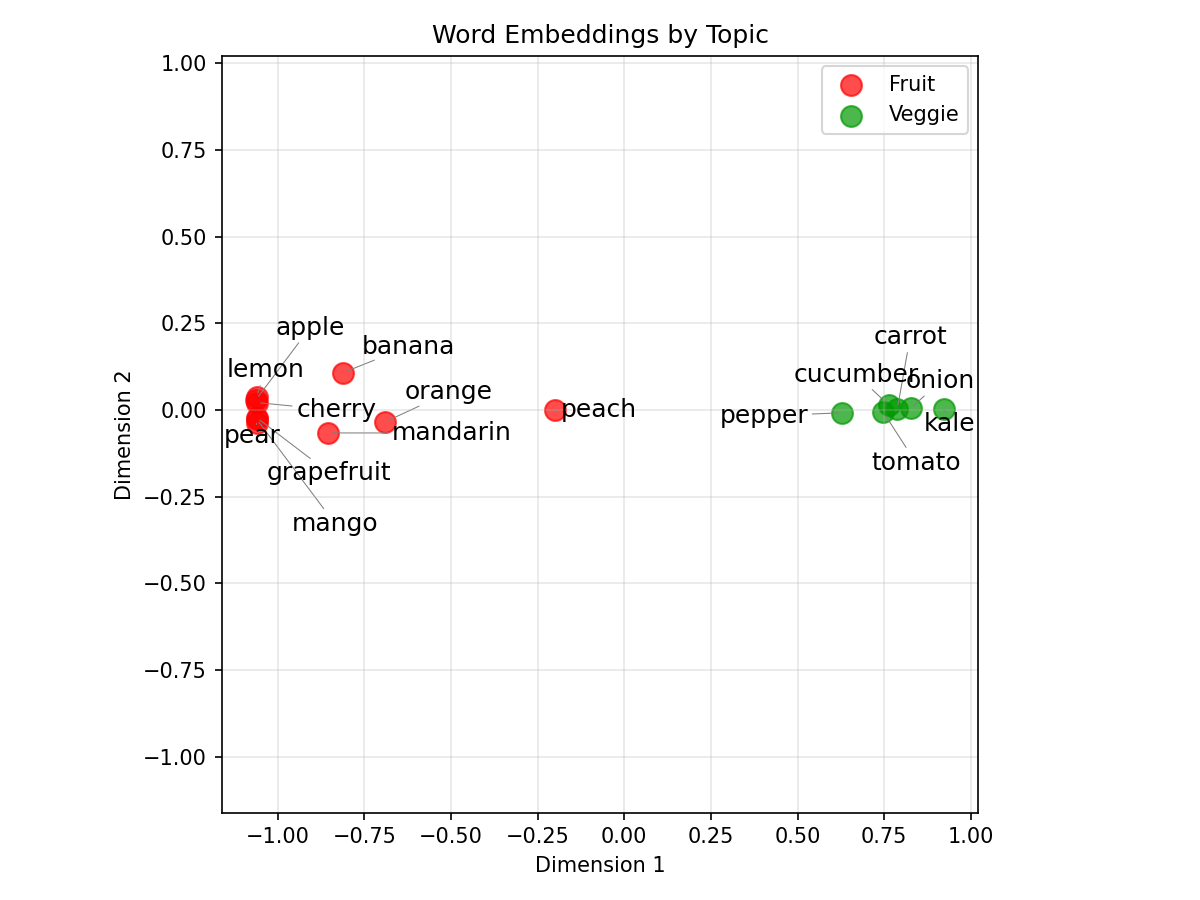}
\caption{SVD-$\beta$}
\end{subfigure}
\caption[Word embeddings, Design 1 ($K=2$)]{Word embeddings for Design 1 ($K=2$) under Dirichlet$(\alpha=0.01)$.}
\label{app-fig:sparse_word_emb_k2}
\end{figure}

\begin{figure}[ht!]
\centering
\begin{subfigure}[b]{0.48\textwidth}
\centering
\includegraphics[width=\textwidth]{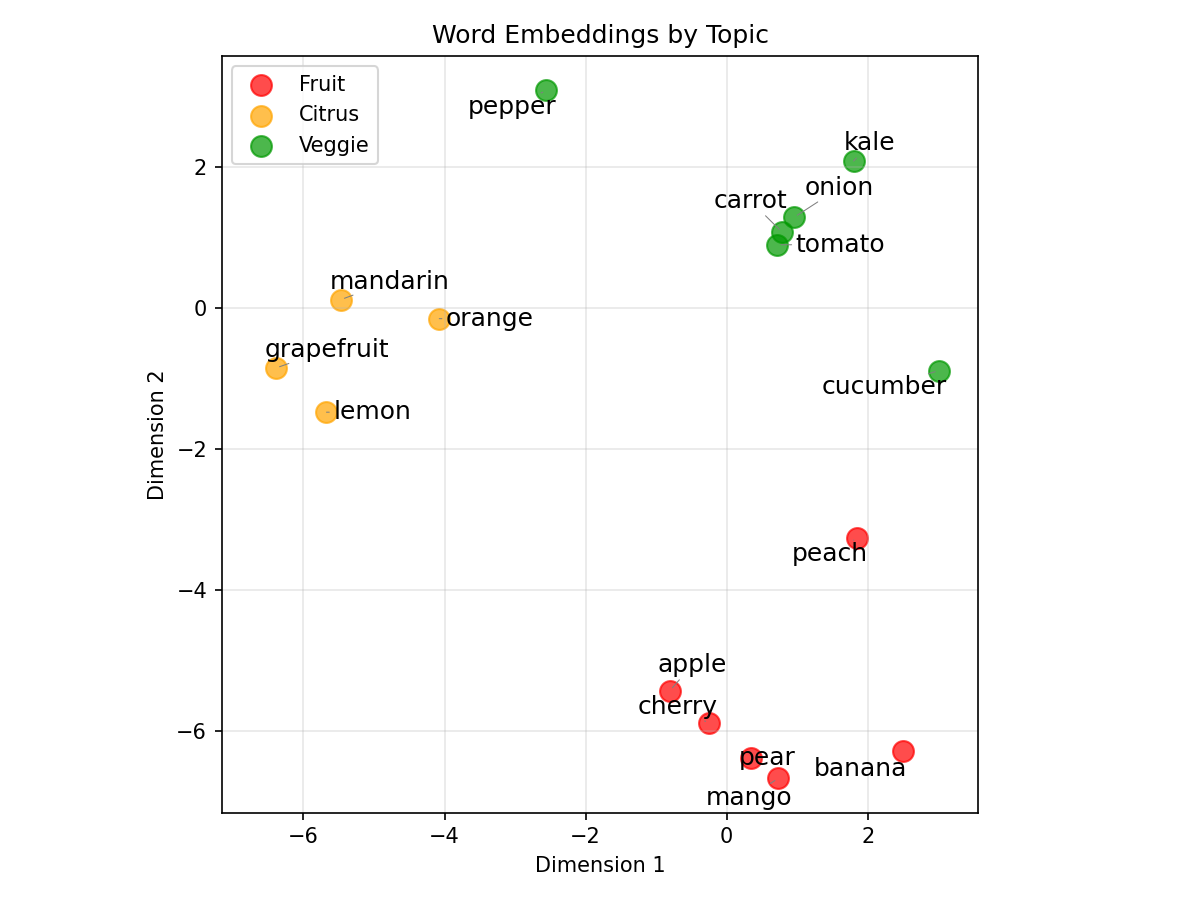}
\caption{CBOW}
\end{subfigure}
\hfill
\begin{subfigure}[b]{0.48\textwidth}
\centering
\includegraphics[width=\textwidth]{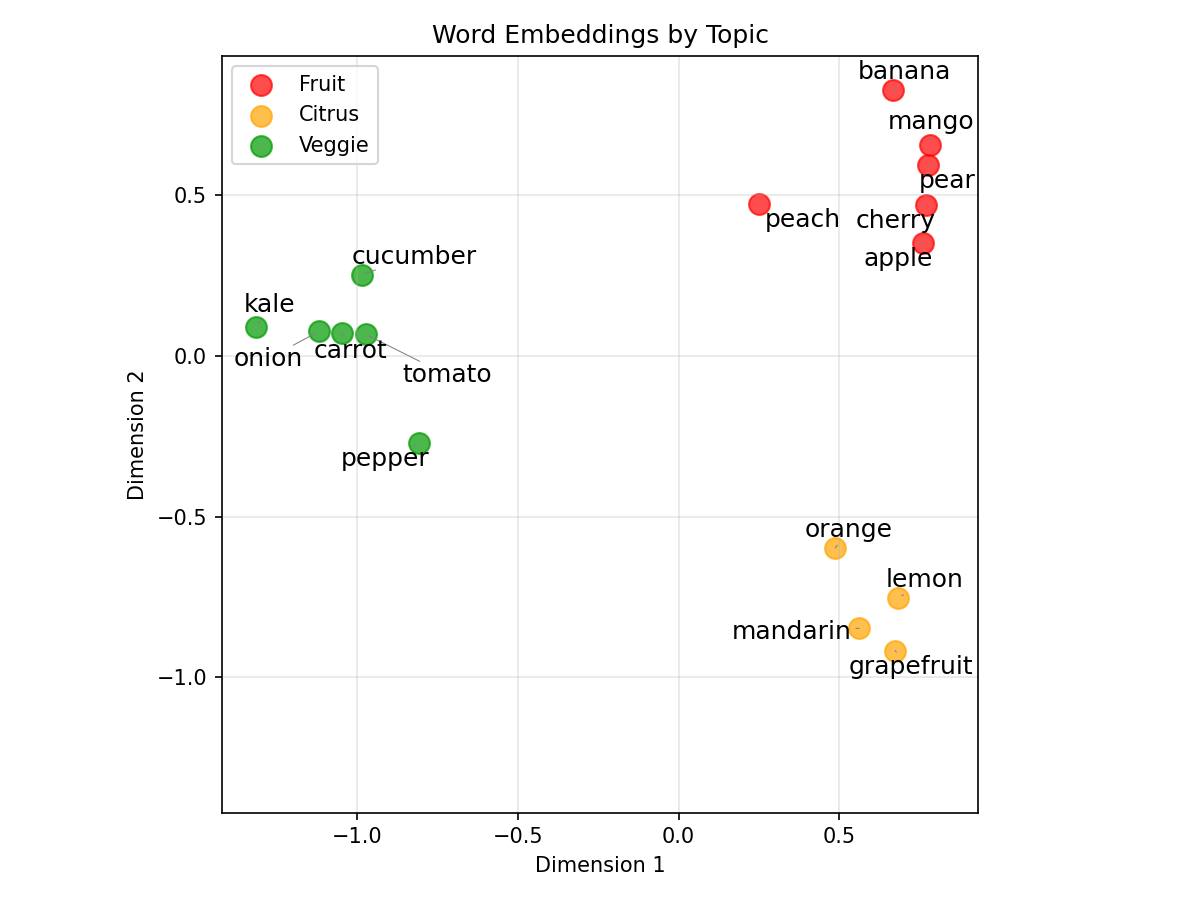}
\caption{SVD-$\beta$}
\end{subfigure}
\caption[Word embeddings, Design 2 ($K=3$)]{Word embeddings for Design 2 ($K=3$) under Dirichlet$(\alpha=0.01)$.}
\label{app-fig:sparse_word_emb_k3}
\end{figure}

\begin{figure}[hb!]
\centering
\begin{subfigure}[b]{0.48\textwidth}
\centering
\includegraphics[width=\textwidth]{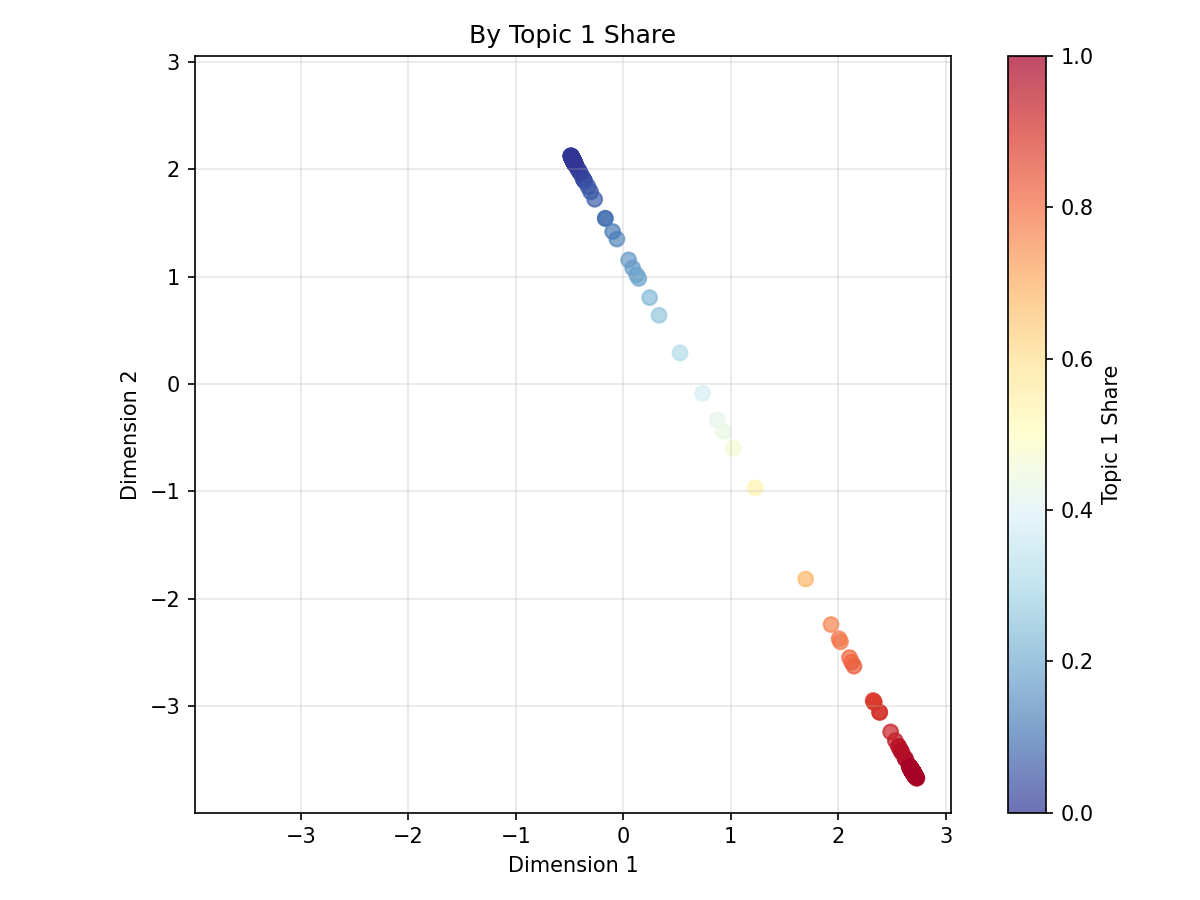}
\caption{CBOW}
\end{subfigure}
\hfill
\begin{subfigure}[b]{0.48\textwidth}
\centering
\includegraphics[width=\textwidth]{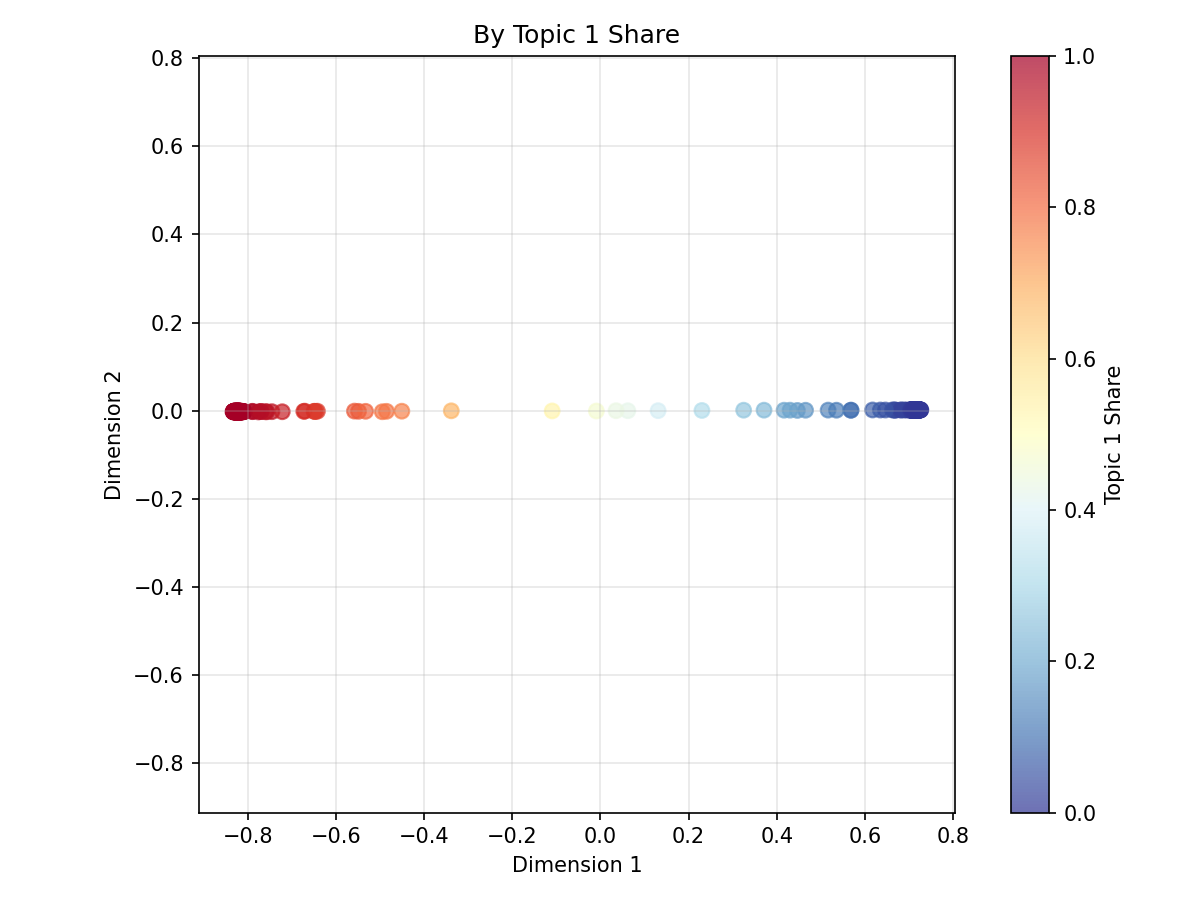}
\caption{SVD-$\beta$}
\end{subfigure}
\caption[Document embeddings, Design 1 ($K=2$)]{Document embeddings for Design 1 ($K=2$) under Dirichlet$(\alpha=0.01)$. Documents cluster at the two endpoints of the line segment $[c_1, c_2]$ rather than filling it; color indicates the share of topic~1.}
\label{app-fig:sparse_doc_emb_k2}
\end{figure}

\begin{figure}[htbp]
\centering
\begin{subfigure}[b]{0.48\textwidth}
\centering
\includegraphics[width=\textwidth]{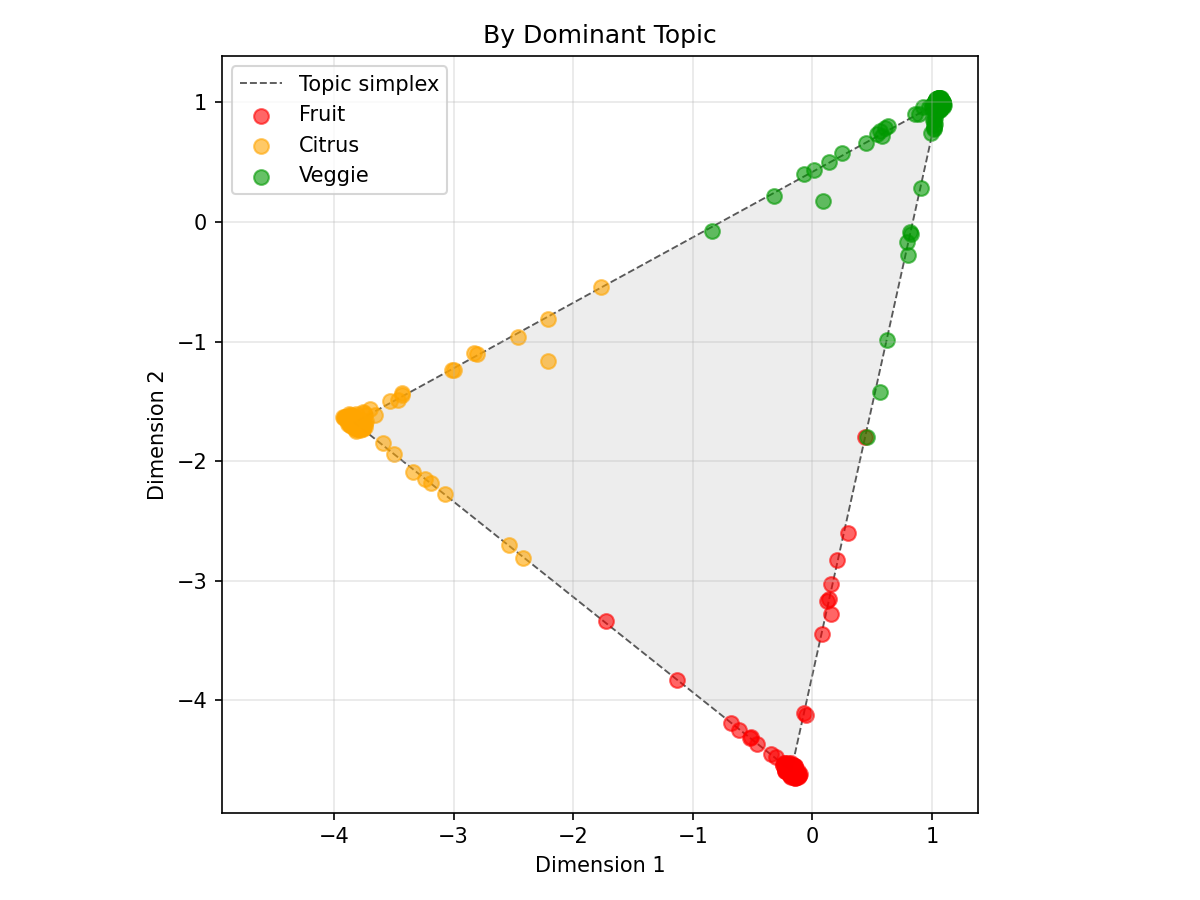}
\caption{CBOW}
\end{subfigure}
\hfill
\begin{subfigure}[b]{0.48\textwidth}
\centering
\includegraphics[width=\textwidth]{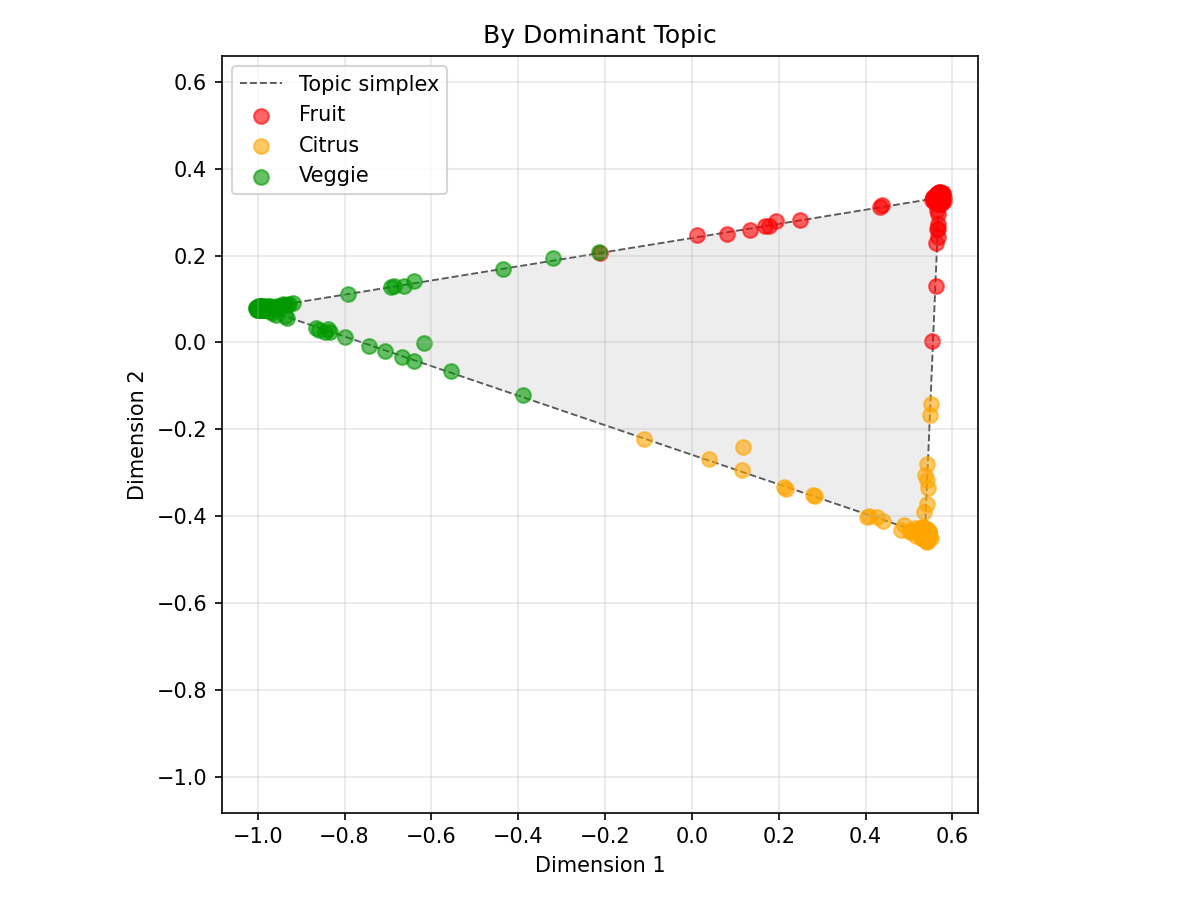}
\caption{SVD-$\beta$}
\end{subfigure}
\caption[Document embeddings, Design 2 ($K=3$)]{Document embeddings for Design 2 ($K=3$) under Dirichlet$(\alpha=0.01)$, colored by dominant topic. With most documents near-pure on a single topic, the dominant-topic partition is well-separated and is recovered cleanly by both embeddings ($k$-means with $K = 3$). The shaded region is the topic simplex $\mathrm{conv}\{c_1,c_2,c_3\}$; documents concentrate at its vertices rather than filling it, which is the contrast with Figure~\ref{fig:doc_emb_noanchor_k3_alt}.}
\label{app-fig:sparse_doc_emb_k3_cluster}
\end{figure}

The sparse-prior case illustrates the boundary regime for Lemma~\ref{lem:voronoi_membership}: when $\theta_{d,k^*(d)} \approx 1$, the dominant-topic threshold $\theta_{d,k^*(d)} > 1 - \eta$ is met for essentially every document, and $k$-means recovers the dominant-topic partition. At $\alpha = 0.01$, $\Pr(\theta_{d,k^*} > 0.5) \approx 1$ for both $K = 2$ and $K = 3$, consistent with the visible vertex concentration.

\clearpage

\subsection{Interior Archetypes ($L < K$)}\label{app-app:sim_archetype}

The simulations above illustrate the case $L = K$, where the number of clusters in document space matches the number of topics. Proposition~\ref{prop:archetype_identification} covers the more general case in which documents concentrate around $L$ \emph{archetype mixtures} $\pi^*_1, \ldots, \pi^*_L \in \Delta_{K-1}$ that need not coincide with the simplex vertices. The empirical application in Section~\ref{sec:application} sits in this regime: $k$-means on the CBSA corpus recovers $L = 5$ clusters whose centroids are interior mixtures of an unknown number of latent topics.

To illustrate this scaling, I repeat Design~2 ($K = 3$ topics, same $B$ matrix as in Section~\ref{sec:simulations}) but replace the Dirichlet$(\alpha = 1)$ prior on $\Theta_{\bullet d}$ with a two-archetype mixture. Each document is assigned uniformly to one of two archetypes,
\[
\pi^*_1 = (0.65,\, 0.25,\, 0.10), \qquad \pi^*_2 = (0.10,\, 0.10,\, 0.80),
\]
and its topic mixture is drawn from a tight Dirichlet centered on the assigned archetype, $\Theta_{\bullet d} \sim \text{Dirichlet}(20 \cdot \pi^*_\ell)$. Archetype~1 is fruit-heavy with some citrus; archetype~2 is nearly pure vegetable. Neither lies at a simplex vertex. The word geometry is not carried over unchanged. By Theorem~\ref{thm:word_embedding_guarantee} the word metric is $\Sigma_\Theta$-weighted. Figure \ref{app-fig:archetype_doc_emb_cluster} illustrates. The shaded region is the topic simplex $\mathrm{conv}\{c_1,c_2,c_3\}$, computed from the true $B$ and the estimated word embeddings and projected onto the same two coordinates as the documents; Proposition~\ref{prop:doc_embedding_general_freq} places every expected document embedding inside it. Crosses mark the two archetypes' expected embeddings, $\sum_k \pi^*_{\ell k}\, c_k$. Both sit in the interior, which is the $L < K$ regime of Proposition~\ref{prop:archetype_identification}. The simplex is far flatter than in Figure~\ref{app-fig:sparse_doc_emb_k3_cluster} even though $B$ is the same, for the reason given above: $\Sigma_\Theta$ is nearly rank one here. That flattening is present in both panels and in the full $K$ coordinates, not an artifact of the projection.

\begin{figure}[htbp]
\centering
\begin{subfigure}[b]{0.48\textwidth}
\centering
\includegraphics[width=\textwidth]{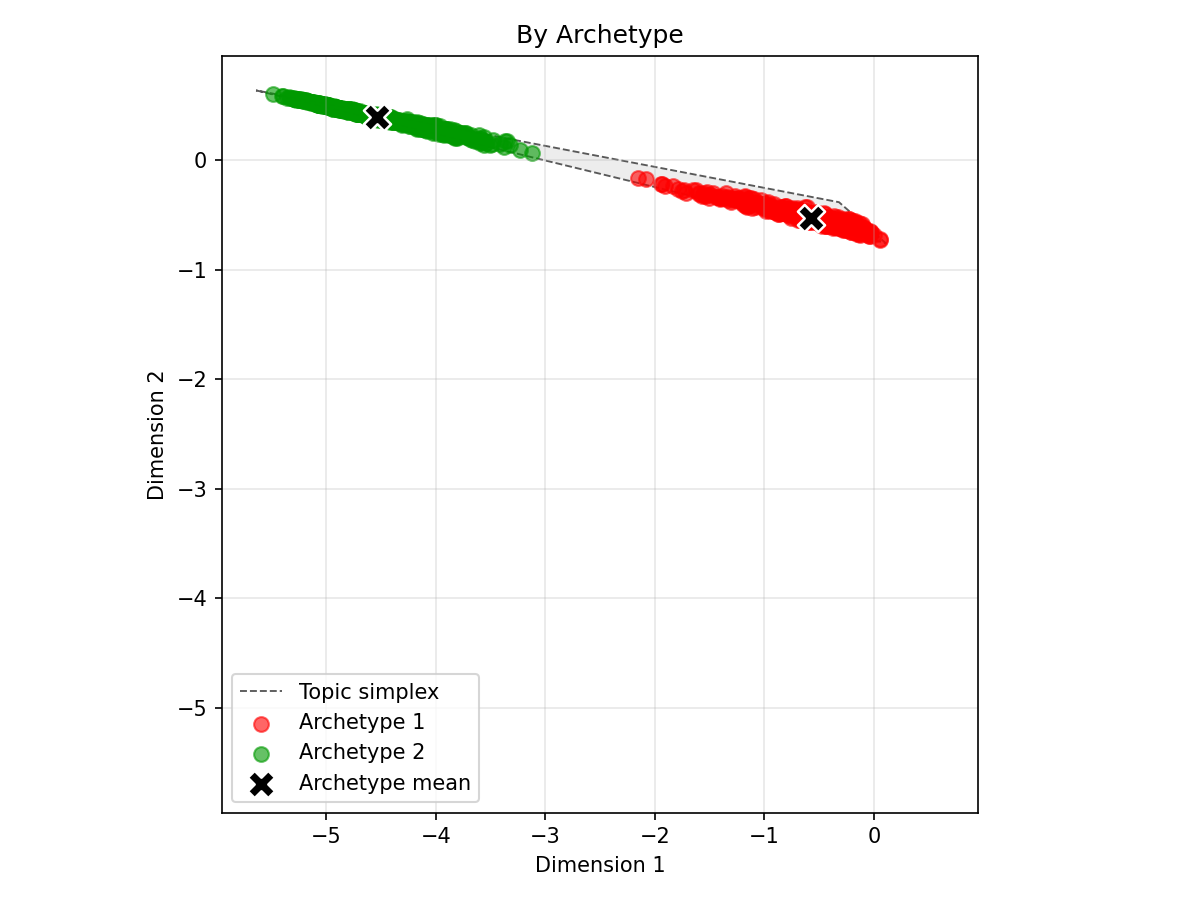}
\caption{CBOW}
\end{subfigure}
\hfill
\begin{subfigure}[b]{0.48\textwidth}
\centering
\includegraphics[width=\textwidth]{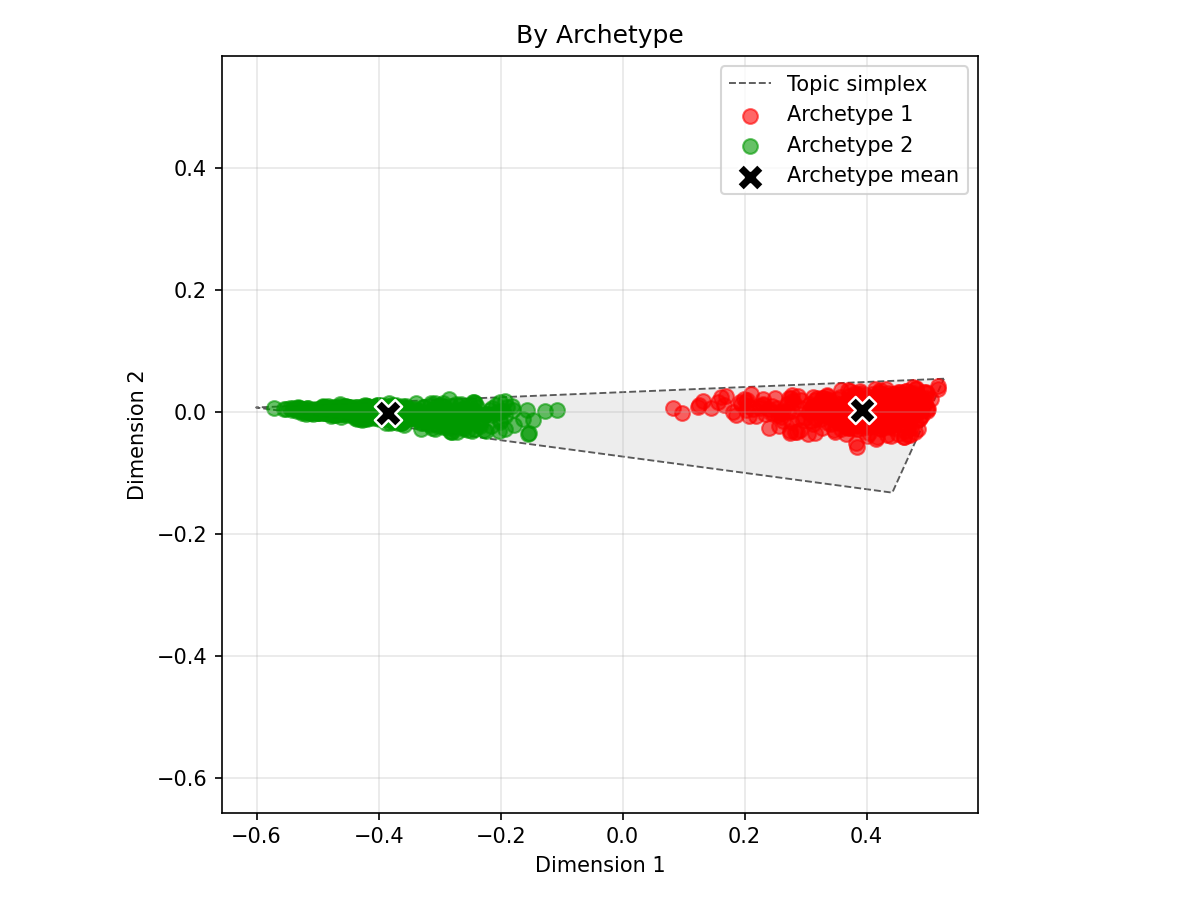}
\caption{SVD-$\beta$}
\end{subfigure}
\caption[Document embeddings by archetype label]{Document embeddings colored by archetype label. The shaded region is the topic simplex $\mathrm{conv}\{c_1,c_2,c_3\}$, computed from the true $B$ and the estimated word embeddings and projected onto the same two coordinates as the documents. Crosses mark the two archetypes' expected embeddings, $\sum_k \pi^*_{\ell k}\, c_k$.}
\label{app-fig:archetype_doc_emb_cluster}
\end{figure}

The archetype design illustrates Proposition~\ref{prop:archetype_identification} in the $L = 2 < K = 3$ regime: cluster identity in embedding space is governed by the archetype assignment, not by the dominant-topic vertex. This mirrors the empirical situation in Section~\ref{sec:application}, where the choice of $L = 5$ is governed by the resolvable archetypes in the corpus and need not equal the (unobserved) number of latent topics.

\clearpage

\subsection{CBOW}\label{app-sec:CBOW}

Here, I depict the counterpart of Figures~\ref{fig:embeddings_large}-Figure~\ref{fig:doc_emb_noanchor_k3_alt} using continuous bag-of-words (CBOW) embeddings.

\begin{figure}[htb!]
\centering
\begin{subfigure}[b]{0.48\textwidth}
\centering
\includegraphics[width=\textwidth]{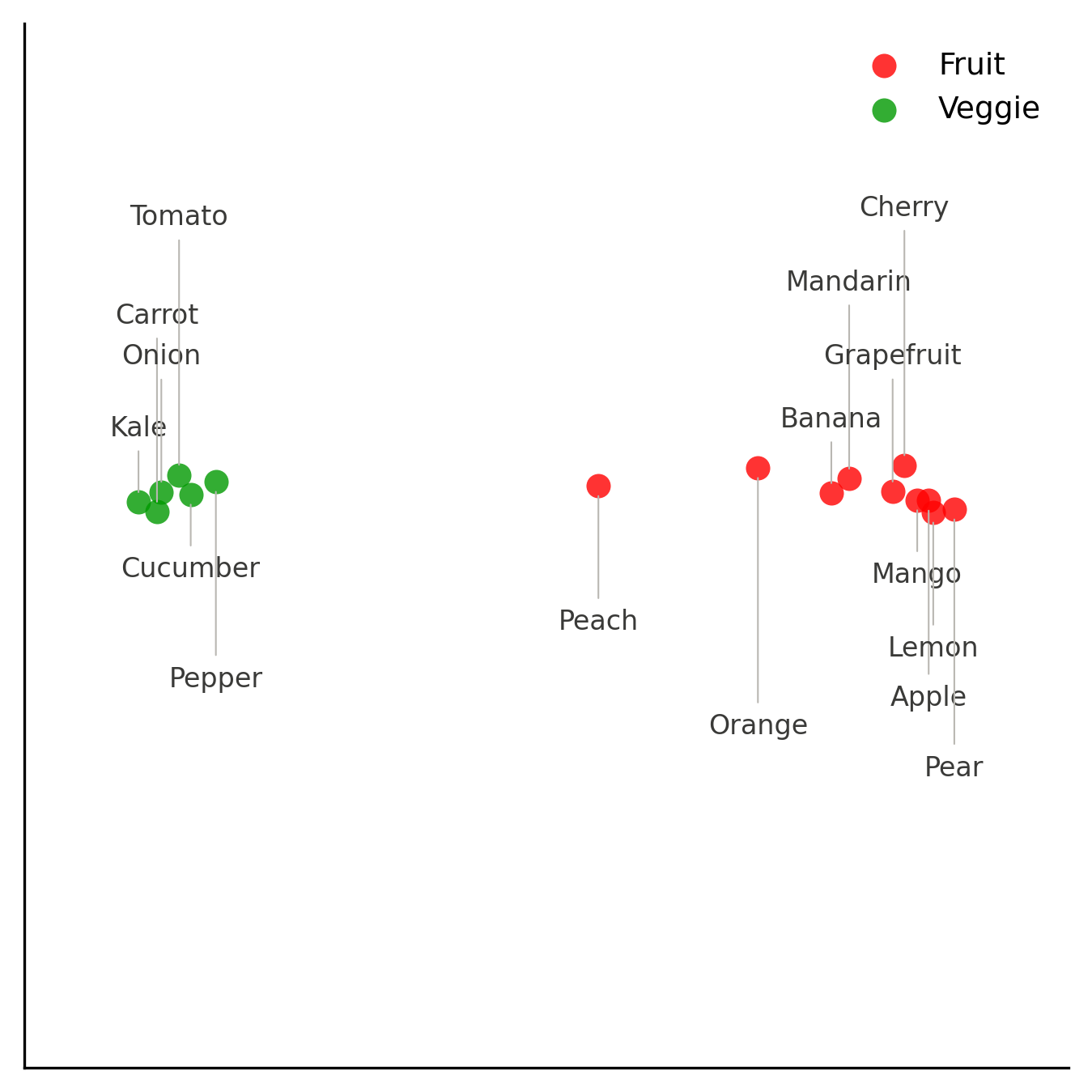}
\caption{Word embeddings, Design 1 ($K=2$)}
\label{app-fig:cbow_word_k2}
\end{subfigure}
\begin{subfigure}[b]{0.48\textwidth}
\centering
\includegraphics[width=\textwidth]{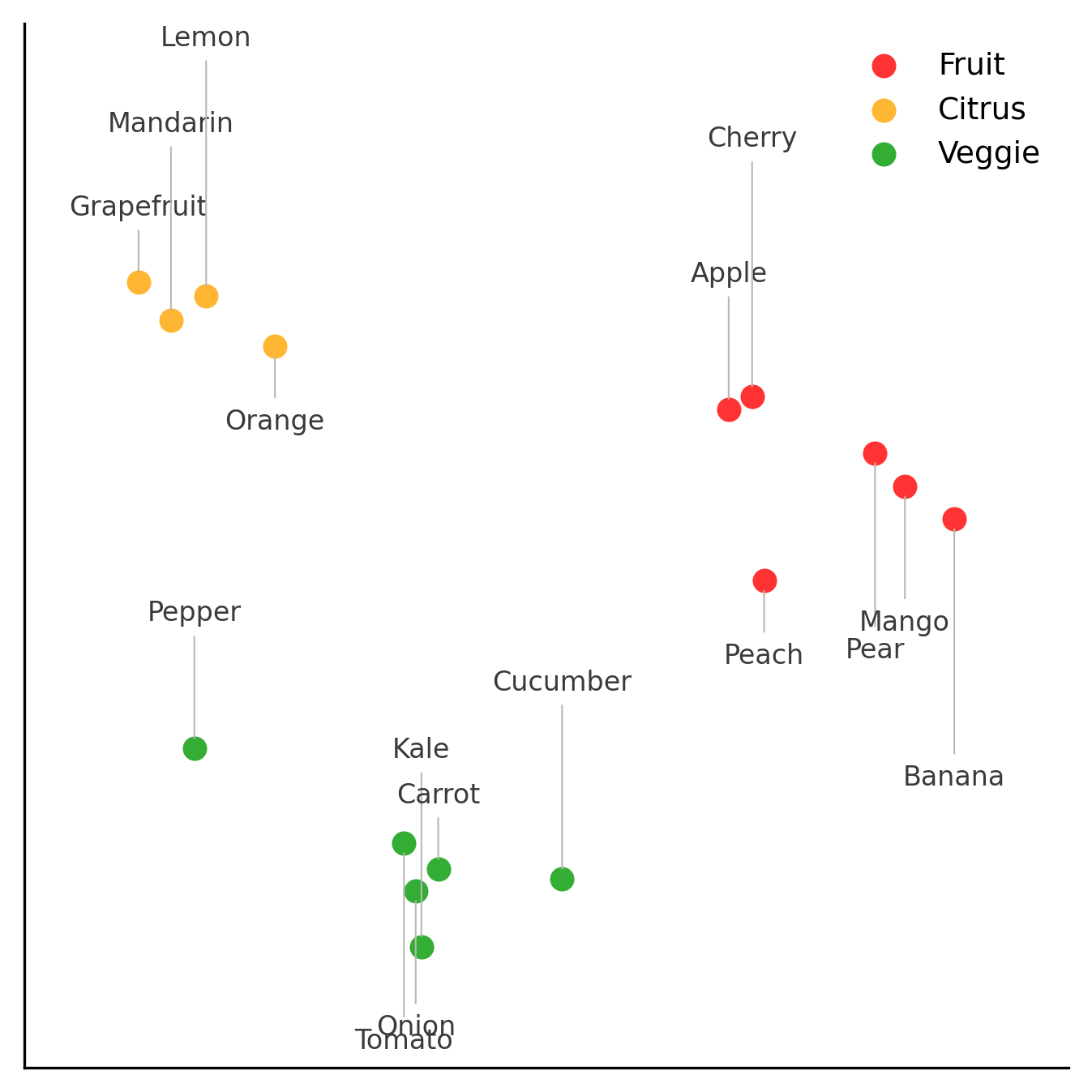}
\caption{Word embeddings, Design 2 ($K=3$)}
\label{app-fig:cbow_word_k3}
\end{subfigure}

\begin{subfigure}[b]{0.48\textwidth}
\centering
\includegraphics[width=\textwidth]{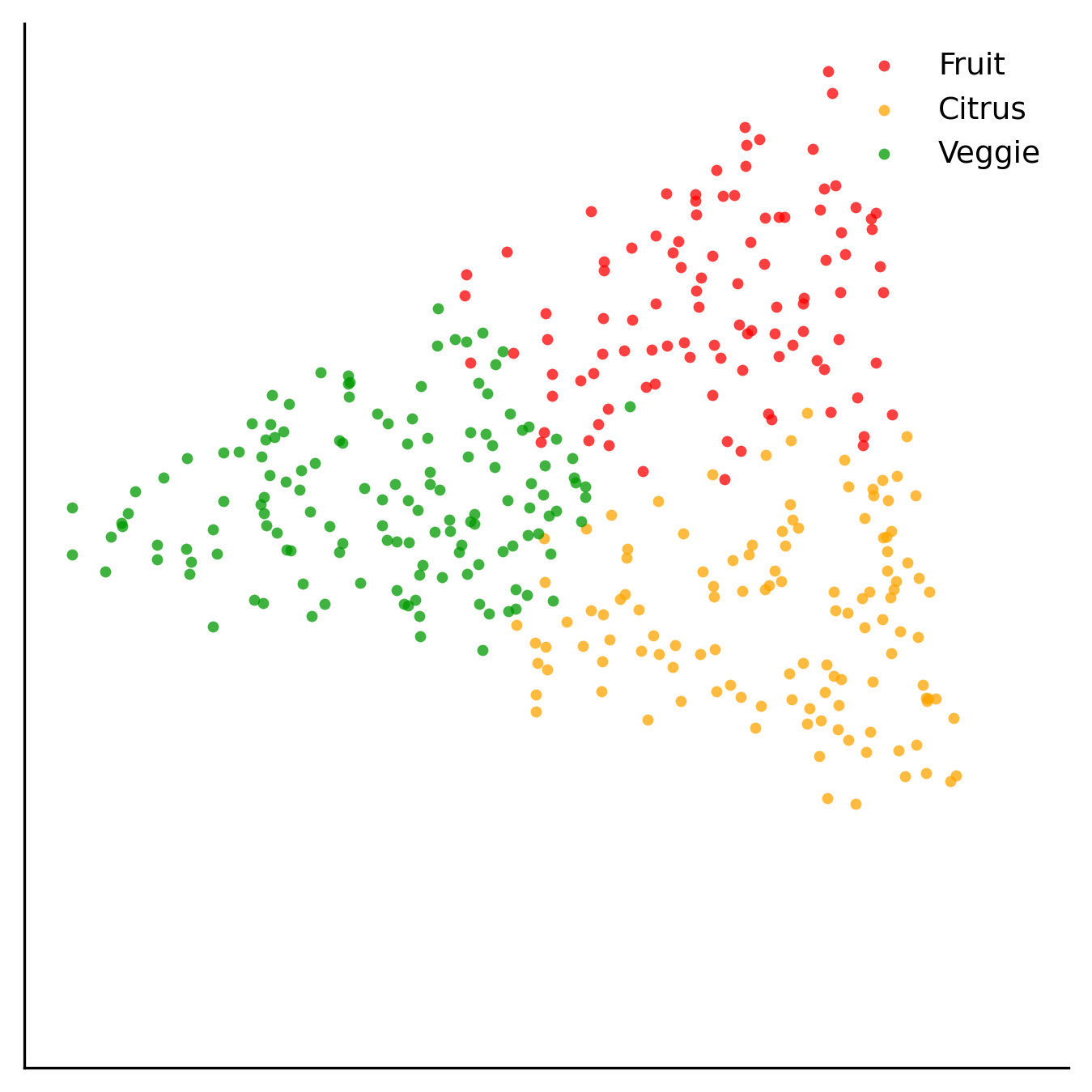}
\caption{Document embeddings, Design 2 ($K=3$)}
\label{app-fig:cbow_doc_k3}
\end{subfigure}
\caption[CBOW embeddings]{Estimated embedding using CBOW on applicaton-sized corpus. Colors mark the dominant topic in the true $B$. }
\end{figure}

\clearpage

\subsection{UMAP Projections}\label{app-app:sim_umap}

Figure \ref{fig:doc_emb_noanchor_k3_alt} projects the $3$-dimensional embeddings to two dimensions by retaining the first two coordinates. The empirical application in Section~\ref{sec:application} instead uses UMAP \citep{mcinnes2018umap}, a nonlinear dimensionality-reduction method, because the application's embeddings live in $50$ dimensions (closed-form SVD-$\beta$ and both corpus-trained arms), $300$ (pretrained Google News) or $3{,}072$ (the LLM embeddings). For visual parallelism with the application figures, this section repeats the Design~2 ($K=3$) embeddings under the same UMAP projection used in Section~\ref{sec:application}. UMAP's nonlinear transformation can curve or warp the simplex but generally preserves the cluster structure.

\begin{figure}[htbp]
\centering
\begin{subfigure}[b]{0.32\textwidth}
\centering
\includegraphics[width=\textwidth]{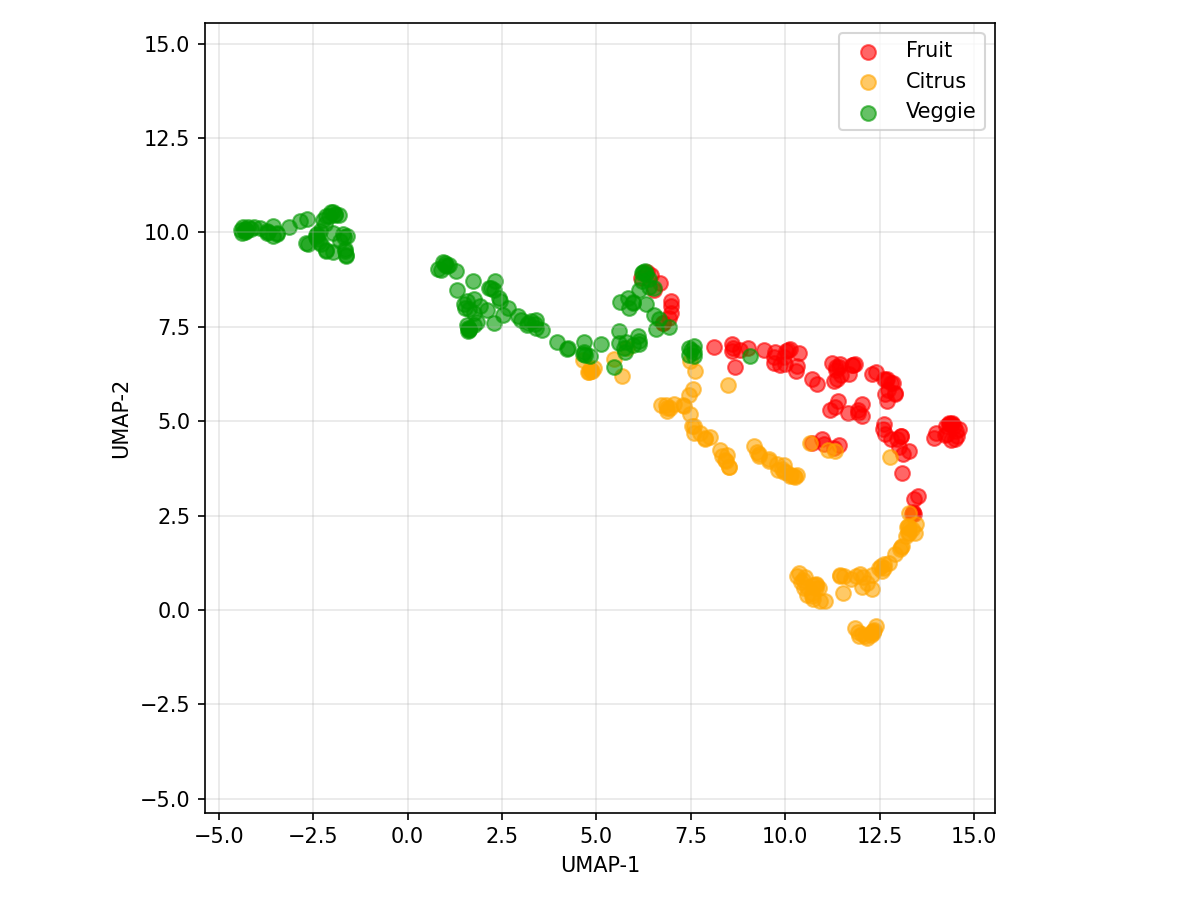}
\caption{SVD-$\beta$}
\end{subfigure}
\hfill
\begin{subfigure}[b]{0.32\textwidth}
\centering
\includegraphics[width=\textwidth]{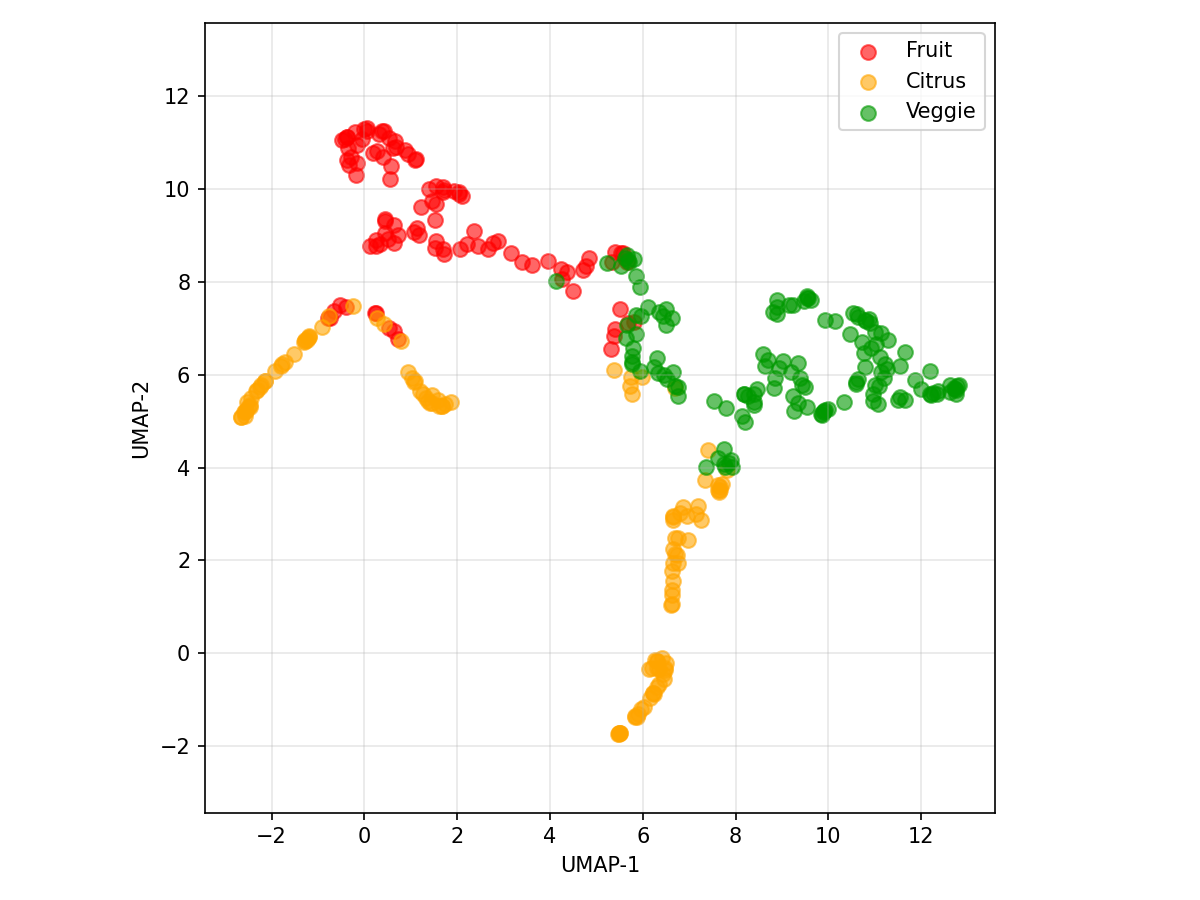}
\caption{SGNS}
\end{subfigure}
\hfill
\begin{subfigure}[b]{0.32\textwidth}
\centering
\includegraphics[width=\textwidth]{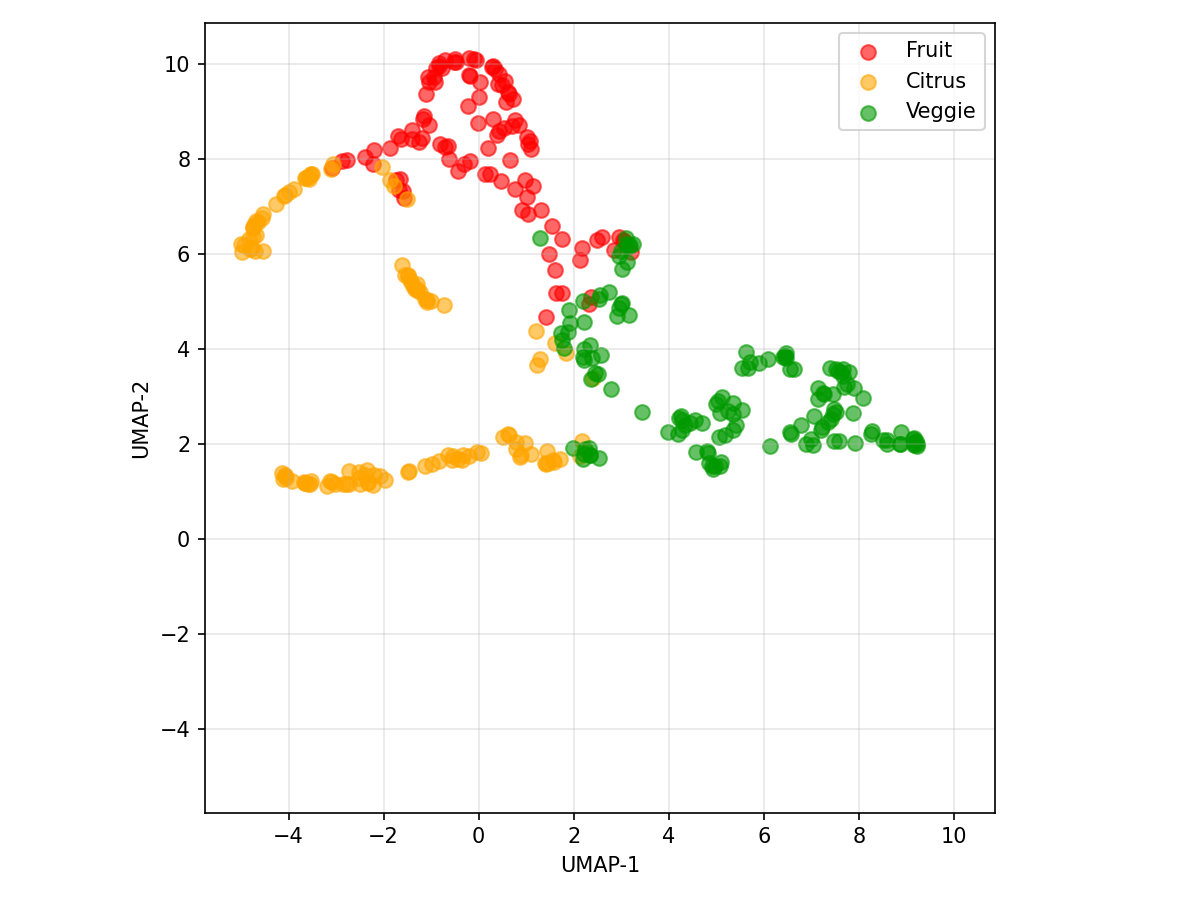}
\caption{CBOW}
\end{subfigure}

\vspace{0.6em}

\begin{subfigure}[b]{0.32\textwidth}
\centering
\includegraphics[width=\textwidth]{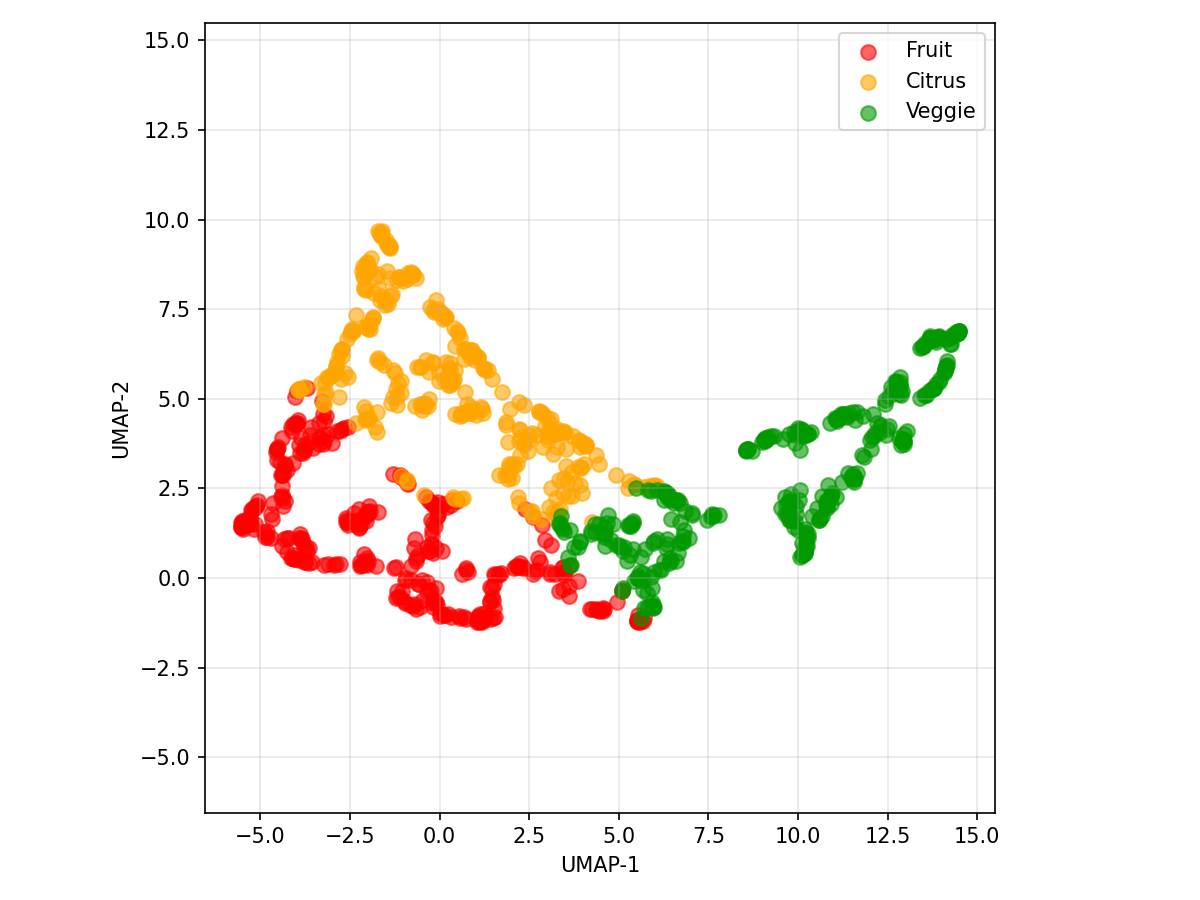}
\caption{SVD-$\beta$, large corpus}
\end{subfigure}
\hfill
\begin{subfigure}[b]{0.32\textwidth}
\centering
\includegraphics[width=\textwidth]{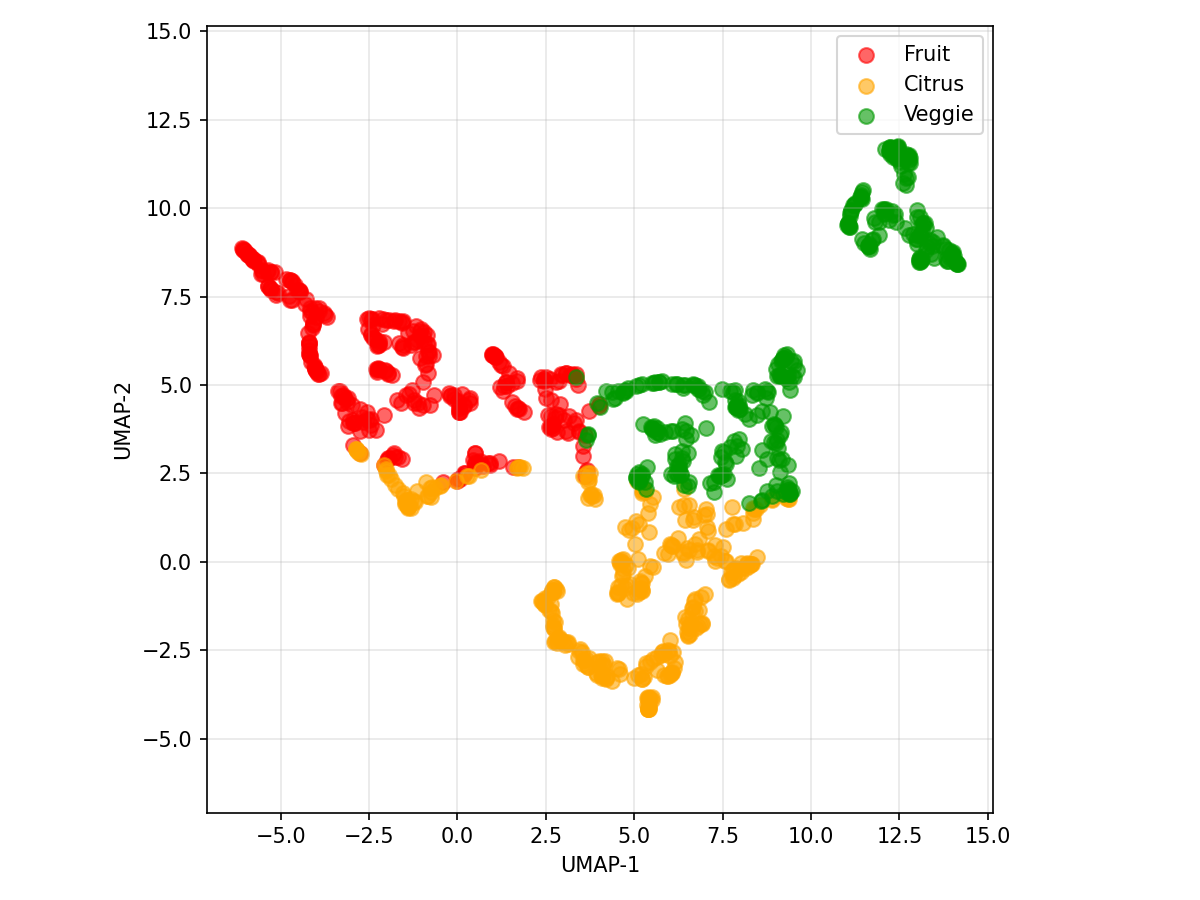}
\caption{SGNS, large corpus}
\end{subfigure}
\hfill
\begin{subfigure}[b]{0.32\textwidth}
\centering
\includegraphics[width=\textwidth]{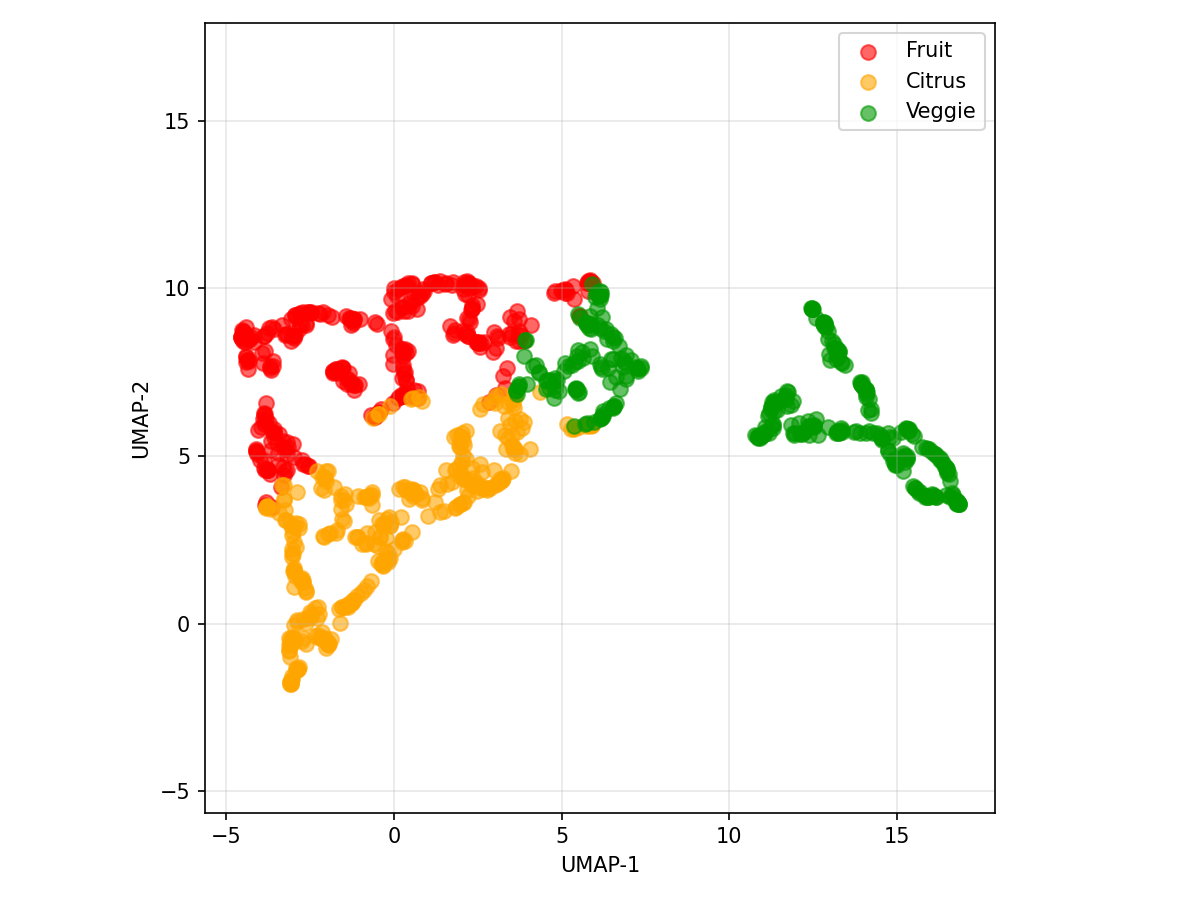}
\caption{CBOW, large corpus}
\end{subfigure}
\caption[Document embeddings, Design 2 ($K=3$), UMAP projection]{Document embeddings for Design 2 ($K=3$), UMAP-projected and colored by dominant topic. Top row: application scale ($D=363$, $N_d=487$, full document as context), the same embeddings Figure~\ref{fig:doc_emb_noanchor_k3_alt} projects linearly. Bottom row: the larger corpus ($D=1{,}000$, $N_d=10{,}000$, $J=5$), the counterpart of Figure~\ref{app-fig:paper_scale_doc_k3}. Cluster structure is preserved under UMAP at both scales.}
\label{app-fig:umap_doc_emb_k3_cluster}
\end{figure}

 \clearpage

 \subsection{Rank of the SGNS Target under a Sparser Prior}\label{app-app:target_rank_sparse}

Table~\ref{tab:target_rank_alpha1} draws every column of $B$ and every topic mixture from a symmetric Dirichlet$(\alpha=1)$, which puts no mass at the boundary of either simplex: each topic loads on every word and each document on every topic. Table~\ref{app-tab:target_rank_alpha01} repeats that grid at $\alpha = 0.1$, where both are more concentrated.

Two things change. The level falls sharply: $\pi_{K-1}$ runs from $0.68$ to $0.89$, against $0.91$ to $0.99$ at $\alpha=1$, so a rank-$(K-1)$ approximation now misses between a tenth and a third of the spectral mass of $\log R$. And the invariance to $V$ breaks for moderately sized $K$. This suggests the cost of the rank constraint grows with the vocabulary under a sparse prior.

Both follow from a larger departure from independence. The diagnostic
$\varrho = \|R - \mathbf{1}_V\mathbf{1}_V^\top\|_\infty$ stays below $0.6$ in every cell at $\alpha=1$,
but ranges from $0.84$ to $6.5$ here and exceeds one in $25$ of the $30$ cells.
Remark~\ref{rem:target_vs_embedding} expands $\log R = \beta\beta^\top + O(\varrho^2)$ under
$\varrho < 1$, so that expansion does not cover most of this table. 

\begin{table}[bh!]
\centering
\footnotesize
\setlength{\tabcolsep}{6pt}
\begin{tabular}{lrrrrrr}
\toprule
$V$ & $K=2$ & $K=3$ & $K=5$ & $K=10$ & $K=25$ & $K=50$ \\
\midrule
100 & $0.677$ & $0.693$ & $0.728$ & $0.779$ & $0.831$ & $0.890$ \\
250 & $0.677$ & $0.692$ & $0.715$ & $0.743$ & $0.788$ & $0.827$ \\
500 & $0.676$ & $0.691$ & $0.711$ & $0.731$ & $0.755$ & $0.808$ \\
1,000 & $0.677$ & $0.692$ & $0.705$ & $0.721$ & $0.726$ & $0.786$ \\
2,500 & $0.676$ & $0.691$ & $0.704$ & $0.712$ & $0.702$ & $0.747$ \\
\bottomrule
\end{tabular}
\caption[Rank of the SGNS target across $V$ and $K$]{Share of the spectral
mass of $\log R$ carried by its leading $K-1$ eigenvalues,
$\pi_{K-1} = \sum_{i \le K-1}|\lambda_i| / \sum_i |\lambda_i|$, at
$\alpha = 0.1$. Each column of $B$ and each topic mixture
$\Theta_{\bullet d}$ is drawn from a symmetric Dirichlet$(\alpha)$; $R$ is then
formed in population from $B$ and $G$, so no corpus is sampled. Means over 8
draws; the largest standard deviation in any cell is $0.016$. The
negative-sampling shift contributes one further eigenvalue, of size $V\log\nu$,
and is excluded.}
\label{app-tab:target_rank_alpha01}
\end{table}

\clearpage
\clearpage
\bibliographystyle{plainnat}
\bibliography{embeddings}
\end{document}